\documentclass[12pt]{article}

\PassOptionsToPackage{unicode}{hyperref}
\PassOptionsToPackage{hyphens}{url}
\PassOptionsToPackage{dvipsnames,svgnames,x11names}{xcolor}

\usepackage{amsmath,amssymb}
\IfFileExists{upquote.sty}{\usepackage{upquote}}{}
\IfFileExists{microtype.sty}{  \usepackage[]{microtype}
  \UseMicrotypeSet[protrusion]{basicmath}}{}
\makeatletter
\@ifundefined{KOMAClassName}{  \IfFileExists{parskip.sty}{    \usepackage{parskip}
  }{    \setlength{\parindent}{0pt}
    \setlength{\parskip}{6pt plus 2pt minus 1pt}}
}{  \KOMAoptions{parskip=half}}
\makeatother
\usepackage{xcolor}
\makeatletter
\ifx\paragraph\undefined\else
  \let\oldparagraph\paragraph
  \renewcommand{\paragraph}{
    \@ifstar
      \xxxParagraphStar
      \xxxParagraphNoStar
  }
 \newcommand{\xxxParagraphStar}[1]{\oldparagraph*{#1}\mbox{}}
  \newcommand{\xxxParagraphNoStar}[1]{\oldparagraph{#1}\mbox{}}
\fi
\ifx\subparagraph\undefined\else
  \let\oldsubparagraph\subparagraph
  \renewcommand{\subparagraph}{
    \@ifstar
      \xxxSubParagraphStar
      \xxxSubParagraphNoStar
  }
  \newcommand{\xxxSubParagraphStar}[1]{\oldsubparagraph*{#1}\mbox{}}
  \newcommand{\xxxSubParagraphNoStar}[1]{\oldsubparagraph{#1}\mbox{}}
\fi
\makeatother

\usepackage{longtable,booktabs,array}
\usepackage{calc}\usepackage{etoolbox}
\makeatletter
\patchcmd\longtable{\par}{\if@noskipsec\mbox{}\fi\par}{}{}
\makeatother
\IfFileExists{footnotehyper.sty}{\usepackage{footnotehyper}}{\usepackage{footnote}}
\makesavenoteenv{longtable}
\usepackage{graphicx}
\makeatletter
\def\maxwidth{\ifdim\Gin@nat@width>\linewidth\linewidth\else\Gin@nat@width\fi}
\def\maxheight{\ifdim\Gin@nat@height>\textheight\textheight\else\Gin@nat@height\fi}
\makeatother
\setkeys{Gin}{width=\maxwidth,height=\maxheight,keepaspectratio}
\makeatletter
\def\fps@figure{htbp}
\makeatother

\makeatletter
\@ifpackageloaded{caption}{}{\usepackage{caption}}
\AtBeginDocument{\ifdefined\contentsname
  \renewcommand*\contentsname{Table of contents}
\else
  \newcommand\contentsname{Table of contents}
\fi
\ifdefined\listfigurename
  \renewcommand*\listfigurename{List of Figures}
\else
  \newcommand\listfigurename{List of Figures}
\fi
\ifdefined\listtablename
  \renewcommand*\listtablename{List of Tables}
\else
  \newcommand\listtablename{List of Tables}
\fi
\ifdefined\figurename
  \renewcommand*\figurename{Figure}
\else
  \newcommand\figurename{Figure}
\fi
\ifdefined\tablename
  \renewcommand*\tablename{Table}
\else
  \newcommand\tablename{Table}
\fi
}
\@ifpackageloaded{float}{}{\usepackage{float}}
\floatstyle{ruled}
\@ifundefined{c@chapter}{\newfloat{codelisting}{h}{lop}}{\newfloat{codelisting}{h}{lop}[chapter]}
\floatname{codelisting}{Listing}

\makeatother
\makeatletter
\@ifpackageloaded{caption}{}{\usepackage{caption}}
\@ifpackageloaded{subcaption}{}{\usepackage{subcaption}}
\makeatother

\usepackage[]{natbib}
\usepackage{bookmark}

\IfFileExists{xurl.sty}{\usepackage{xurl}}{}\hypersetup{
  pdftitle={Title},
  pdfauthor={Author 1; Author 2},
  pdfkeywords={3 to 6 keywords, that do not appear in the title},
  colorlinks=true,
    linkcolor={black},
  filecolor={Maroon},
    citecolor={black},
    urlcolor={black},
  pdfcreator={LaTeX via pandoc}}

\newcommand{\anon}{1}

\usepackage[nodisplayskipstretch]{setspace}
\usepackage{bm}\usepackage{color}\usepackage{hyperref}\usepackage{graphicx}\usepackage{vmargin}\usepackage{multirow}
\usepackage{makecell}\usepackage{adjustbox}\usepackage{booktabs}
\usepackage{enumitem}
\usepackage{dsfont}
\usepackage{subcaption}
\RequirePackage{cleveref}
\usepackage{authblk}

\usepackage[toc,page,header]{appendix}
\usepackage{minitoc}

\usepackage{soul}
\usepackage{wrapfig}

\usepackage{algorithm}
\usepackage{algorithmic}
\usepackage{comment}
\usepackage{amsthm}

\hypersetup{colorlinks=true,linkcolor=black,urlcolor=black,citecolor=black,urlcolor=black}

\newtheorem{assumption}{Assumption}

\newtheorem{lemma}{Lemma}

\newtheorem{theorem}{Theorem}

\newtheorem{remark}{Remark}

\crefname{assumption}{assumption}{assumptions}

\setpapersize{USletter}
\setmarginsrb{1in}{1in}{1in}{1in}{0pt}{0mm}{0pt}{0mm}

\newcommand{\pseudodot}{{\lower 2.4pt\hbox{$\cdot$}}}

\newcommand\R{{\mathbb R}}
\newcommand\PP{{\mathbb{P}}}
\newcommand\Q{{\mathbb{Q}}}
\newcommand\X{{\mathbf X}}
\newcommand\Xn{\X_n}
\newcommand\Zn{\bZ_{n}}
\newcommand\Zm{\bZ_{m}}

\newcommand\Pn{P^{(n)}}
\newcommand\pn{p^{(n)}}
\newcommand\mN{{\mathcal{N}}}
\newcommand\E{{\mathbb E}}

\newcommand\bX{{\mathbf X}}
\newcommand\bZ{{\mathbf Z}}

\newcommand\mA{{\mathcal A}}

\newcommand\mD{{\mathcal D}}
\newcommand\mX{{\mathcal X}}
\newcommand\mO{{\mathcal O}}

\newcommand\wt{\widetilde}
\newcommand\wh{\widehat}

\newcommand\I{\mathds{1}}

\newcommand\diff{\mathrm d}

\newcommand{\abs}[1]{\left\vert#1\right\vert}
\newcommand{\norm}[1]{\left\lVert#1\right\rVert}
\newcommand{\normsm}[1]{\lVert#1\rVert}
\newcommand{\normbig}[1]{\big\lVert#1\big\rVert}
\newcommand{\normbigg}[1]{\bigg\lVert#1\bigg\rVert}

\newcommand\iid{\overset{\text{iid}}{\sim}}

\newcommand\MLE{\widehat{\theta}^{\text{\rm MLE}}_n}

\newcommand\blfootnote[1]{  \begingroup
  \renewcommand\thefootnote{}\footnote{#1}  \addtocounter{footnote}{-1}  \endgroup
}

\begin{document}

\def\spacingset#1{\renewcommand{\baselinestretch}{#1}\small\normalsize} \spacingset{1}

\doparttoc\faketableofcontents

\if1\anon
{
  \title{\bf Simulation-based Inference via Langevin Dynamics with Score Matching}

\author[1]{Haoyu Jiang}
\author[1]{Yuexi Wang\footnote{Corresponding author.}}
\author[2]{Yun Yang}

\affil[1]{{Department of Statistics}, {University of Illinois Urbana-Champaign}, {USA}}
\affil[2]{{Department of Mathematics}, {University of Maryland, College Park},
{USA}}
  \maketitle
  \blfootnote{Wang's research is support by NSF grant DMS-2515542. Codex and Claude Code were used to refine the language and improve the readability of the manuscript.}
} \fi

\if0\anon
{
  \begin{center}
    {\LARGE\bf  Simulation-based Inference via Langevin Dynamics with Score Matching}
\end{center}

} \fi

\begin{abstract}
Simulation-based inference (SBI) enables Bayesian analysis when the likelihood is intractable but model simulations are available. Recent advances in statistics and machine learning, including Approximate Bayesian Computation and deep generative models, have expanded the applicability of SBI, yet these methods often face { substantial computational} challenges { as the sample size and parameter dimension increase.}
{In this paper,} we propose a novel {scalable} SBI method that integrates score matching with Langevin dynamics, {while explicitly exploiting the statistical structure of log-likelihood functions. Our approach combines (i) a localization scheme that concentrates computation in regions of high posterior mass and (ii) a structured score network that embeds key properties of likelihood scores, including additivity across observations and Fisher information identities. We provide theoretical and empirical evidence demonstrating that the proposed structured score-matching approach improves statistical efficiency and computational scalability, achieving competitive or superior performance compared to existing SBI methods on both benchmark and challenging problems with large sample sizes and moderate-dimensional parameter spaces.}

\end{abstract}


\vspace{-2em}

\spacingset{1.7}
\section{Introduction\label{sec:intro}}
In fields such as ecology, biology, economics, and psychology, researchers often rely on complex structural models {whose parameters are essential for understanding scientific phenomena and guiding decisions.}
However, many such models involve richly structured parameter spaces and complex data-generating mechanisms, which create major challenges for inference and computation due to intractable likelihood functions. As a result, traditional likelihood-based inference is often computationally infeasible or entirely impractical.

Simulation-Based Inference (SBI) has emerged as a powerful framework for addressing this challenge. By relying on model simulations rather than explicit likelihood evaluations, SBI enables inference even when the likelihood is inaccessible. Within a Bayesian framework, prior knowledge about parameters or system behavior can be naturally incorporated through the choice of priors, and posterior distributions can be approximated using methods such as Approximate Bayesian Computation (ABC) \citep{beaumont2010approximate}, Bayesian Synthetic Likelihood (BSL) \citep{price2018bayesian,frazier2021robust}, or more recent neural SBI approaches, including flow-based methods \citep{papamakarios2016fast,papamakarios2019sequential},  generative adversarial networks \citep{wang2022adversarial}, and diffusion models \citep{sharrock2024sequential}, which directly approximate certain quantity of the posterior (c.f.~Section~\ref{sec:Existing_SBI} for a detailed discussion).

{Despite these advances, scalability remains a central challenge for SBI along two complementary axes. First, as the number of parameters increases, blanket exploration of the prior becomes increasingly inefficient, since only a small fraction of simulated parameters produce datasets resembling the observed dataset. The relevant scale of ``high-dimensional parameters'' in SBI differs from that in classical regression. While regression problems may routinely involve hundreds or thousands of covariates, SBI can already become challenging when the parameter dimension $d_\theta$ is on the order of dozens, depending on the simulator and prior. This is because SBI requires exploring and approximating a posterior that is often highly concentrated within a diffuse parameter space, without direct likelihood information, making effective search and coverage substantially more difficult \citep{lueckmann2021benchmarking,li2017extending,chatha2024neural}. Second, as the sample size of the observed dataset increases, both the number of simulations and the complexity of the approximation function may grow with the sample size in order to maintain accuracy. To address these scalability challenges, recent works introduce adaptive schemes such as sequential sampling \citep{papamakarios2019sequential,wang2022adversarial} and Thompson sampling \citep{ohagan2024tree}, which aim to focus simulations in regions near the true data-generating parameter $\theta^\ast$, as well as permutation-invariant architectures \citep{zaheer2017deep,deistler2025simulation} that exploit the exchangeable structure of the data; see \Cref{sec:score_matching_overview} for more details.}

{In this paper, we propose a structured score-matching approach to SBI. Our method combines score estimation with Langevin dynamics \citep{roberts1998optimal} to sample from the posterior distribution, while explicitly exploiting statistical structure in the likelihood score. The key idea is to replace global score learning over the full prior support with localized and structured score learning targeted at the posterior region relevant to the observed dataset.}

{The proposed method has two main components. First, we introduce a localization step that identifies a neighborhood around the true parameter and uses it to construct a proposal distribution for simulation. This concentrates score estimation in the region where posterior accuracy matters most, improving simulation efficiency. Second, we develop a statistically grounded score-matching framework that embeds structural properties of log-likelihood functions directly into the architecture and training of the score network. This restricts the space of score-matching networks to those consistent with likelihood theory, reducing estimation complexity and improving the reliability of downstream sampling.}

We impose two types of structural constraints.
The first constraint is additivity. We restrict the approximating score function $s(\theta,\, \bX_n)$ to admit an additive structure, applicable to i.i.d.~or weakly dependent datasets, by decomposing it as $\sum_{i=1}^n s(\theta,\, X_i)$. Each individual score $s(\theta,\, X_i)$ is evaluated using the same score function $s(\,\cdot, \,\cdot)$, which enables data-efficient learning without requiring the score-network complexity to grow with the sample size.
The second architectural regularization is based on two population-level constraints satisfied by the log-likelihood of regular parametric models (see, e.g., Chapter 2 of \cite{vanderVaart1998}): (i) the mean-zero property, $\mathbb{E}[\nabla \ell(\theta; X)] = 0$ for a single observation $X$, and (ii) the curvature property, $\mathbb{E}[\nabla \ell(\theta; X)\nabla \ell(\theta; X)^T] = \mathbb{E}[-\nabla^2 \ell(\theta; X)]$, whose value defines the Fisher information matrix (c.f.~Section~\ref{sec:regular}). {These two constraints are also known as Bartlett identities \citep{mykland1994bartlett}. Enforcing them also substantially relax the requirement on the quality of the estimated score (c.f. \Cref{rmk:sm_err}). }
{Together, these conditions control bias and enforce the correct local geometry of the score by aligning it with the Fisher information near the true parameter. As a result, they yield a more stable and accurate approximation of the posterior score, leading to more stable Langevin dynamics and better-calibrated uncertainty quantification. To our knowledge, this is the first work to systematically incorporate classical likelihood identities as inductive biases for score estimation, and this principle extends naturally to other score-based inference settings.}

The remainder of the paper is organized as follows. \Cref{sec:background} reviews existing SBI approaches and introduces key results on score-matching networks. \Cref{sec:method} describes our proposed score-based Langevin dynamics approach tailored for SBI. \Cref{sec:theory} presents the theoretical analysis of the proposed sampler. \Cref{sec:simulation} illustrates the performance of our method on several simulated examples. We conclude with future directions in \Cref{sec:discussion}.

\section{Background and Preliminary Results}\label{sec:background}

We consider a collection of $n$ observations $\Xn^\ast = (X_{1}^\ast, \ldots, X_{n}^\ast)^T$ drawn from a distribution $\Pn_{\theta^\ast}$ in the parametric family $\big\{\Pn_{\theta} : \theta \in \Theta \subset \R^{d_\theta}\big\}$, with each observation $X^\ast_{i} \in \R^p$ and $\theta^\ast$ denoting the true parameter. We assume that $\Pn_\theta$, for every $\theta \in \Theta$, admits a density $\pn_\theta$.
Our goal is to perform Bayesian inference on $\theta^\ast$ through its posterior distribution
\begin{equation}\label{eq:true_post}
\pi_n (\theta\mid \Xn^\ast)\, \propto\, \pn_\theta(\Xn^\ast)\cdot \pi(\theta),
\end{equation}
which is determined by the likelihood function $\pn_\theta(\Xn^\ast)$ and the prior distribution (density) $\pi(\theta)$. We focus on simulator-based models where the likelihood cannot be directly evaluated but simulation is feasible. That is, for any parameter $\theta^{(k)}\in \Theta$, one can readily generate pseudo-datasets $ \Xn^{(k)} = ( X^{(k)}_1, \ldots, X^{(k)}_{n}) \sim \Pn_\theta$. We use $\Xn$ to denote a generic dataset.
We also use $\normsm{\cdot}$ to denote $L_2$ norm unless otherwise noted.

\subsection{Challenges in existing SBI methods}\label{sec:Existing_SBI}
When the likelihood is computationally intractable, its evaluation in \eqref{eq:true_post} must be approximated through simulations. The core idea of SBI is to approximate the posterior distribution by identifying parameter values $\theta$ that generate simulated data resembling the observed dataset $\Xn^\ast$. { Existing SBI methods can be broadly categorized into three families: (1) Approximate Bayesian Computation (ABC) \citep{beaumont2010approximate}, which compares simulated and observed datasets using similarity measures; (2) Bayesian Synthetic Likelihood (BSL) \citep{price2018bayesian,frazier2021robust}, which approximates the likelihood of summary statistics with a Gaussian distribution; and (3) neural SBI methods, which use neural networks to learn an approximation to the posterior, likelihood, or likelihood-to-evidence ratio directly from simulated data. Representative examples include neural posterior estimation (NPE) \citep{papamakarios2016fast}, neural likelihood estimation (NLE) \citep{papamakarios2019sequential}, and neural ratio estimation (NRE) \citep{hermans2020likelihood,durkan2020contrastive,miller2022contrastive}.
For a general review of SBI methods, see \citet{cranmer2020frontier,lueckmann2021benchmarking,deistler2025simulation} and our discussion in \Cref{sec:related_work}.}

{
From the perspective of scalability, these methods often face two distinct but related challenges. The first concerns the parameter dimension $d_\theta$. Under blanket search from the prior, the number of simulations required to adequately explore the parameter space can grow rapidly with $d_\theta$. To mitigate this issue, several methods use adaptive or sequential proposals that refine the simulation distribution across rounds. Sequential and adaptive SBI methods address this by iteratively
refining the simulation proposal
\citep{sisson2007sequential,beaumont2009adaptive,del2012adaptive,
papamakarios2019sequential,wang2022adversarial,ohagan2024tree}. Our localization
step (c.f. \Cref{sec:preconditioning}) is motivated by the same principle.

The second challenge concerns the sample size $n$ of the observed dataset. When the dataset consists of many i.i.d. or exchangeable observations, both the cost of simulating a full dataset for each $\theta$ and the complexity of the approximating network can grow with $n$.
To address this issue, many methods rely on low-dimensional summaries of the data as input, which can lead to a loss of statistical efficiency. Neural SBI approaches often use architectures that exploit invariance or exchangeability, such as permutation-invariant embedding networks \citep{zaheer2017deep,wang2022adversarial,luciano2025permutations,deistler2025simulation}, or single-observation likelihood/ratio/score estimation \citep{papamakarios2019sequential,hermans2020likelihood,geffner2022score,linhart2024diffusion}, to reduce the computational burden.
 However, these approaches often lack a transparent way to analyze how approximation error accumulates with $n$, which we discuss in \Cref{rmk:iid_structure}.
}

\subsection{Score-based sampling and score-based SBI} \label{sec:score_matching_overview}

We adopt a different strategy, motivated by gradient-based sampling methods using Hamiltonian \citep{neal2011mcmc} or Langevin dynamics \citep{welling2011bayesian,roberts1996exponential,roberts1998optimal,cheng2018underdamped}, for
which both empirical success and strong theoretical properties have been established in high-dimensional models with tractable likelihoods
\citep{chen2023sampling,tang2024computational}. These methods exploit the local
geometry of the posterior distribution and can converge to the true posterior more efficiently
than random-walk exploration. { This suggests a natural route for SBI: estimate a score function from simulations and use the estimated score in a gradient-based sampler.}

{ More broadly, score matching has become a powerful tool for distribution learning. Foundational work by \citet{hyvarinen2005estimation,song2019generative,meng2020autoregressive} shows how score estimation can be performed without direct access to the likelihood (see \Cref{sec:conditional_sm} for more details), and a large body of work has since developed around this idea, including applications to SBI.

Several recent works pursue related directions. One line estimates score- or energy-related objects and uses them directly as inference targets without relying on a diffusion model, including \citet{pacchiardi2022score,zeghal2022neural,glaser2022maximum,khoo2025direct}. A more recent line combines score estimation with diffusion-based generative modeling. This includes neural score estimation and conditional diffusion approaches such as \citet{simons2023neural,sharrock2024sequential,nautiyalcondisim}, as well as joint modeling approaches such as \citet{gloeckler2024all}. In particular, \citet{geffner2023compositional,linhart2024diffusion} leverage the i.i.d.~data structure to construct score-based diffusion models that scale better with the sample size $n$. We provide a more thorough discussion of these related methods and comparisons with our approach in \Cref{rmk:iid_structure}.

}

\subsection{Conditional score matching and a generic baseline for SBI}\label{sec:conditional_sm}

{
To set up the specialized design in \Cref{sec:method}, we first describe how the likelihood score can be learned from simulations via conditional score matching. We then introduce a generic baseline that uses the estimated score as the key ingredient in a Langevin sampler for SBI.

Let $s_\phi(\theta,\Xn): \R^{d_\theta} \times \R^{np} \to \R^{d_\theta}$ denote a score model parametrized by $\phi$. The goal of conditional score matching is to find $\phi$ such that $s_\phi(\theta, \Xn)$ accurately approximates the true likelihood score $\nabla_\theta \log \pn_\theta(\Xn)$.
One may choose to estimate either the likelihood score
$\nabla_\theta \log \pn_\theta(\Xn)$ or the posterior score
$\nabla_\theta \log \pi_n(\theta \mid \Xn)$. In this work, we focus on the likelihood score because it provides a more convenient framework for incorporating statistical structures and avoids samples re-weighting after replacing the prior with a proposal distribution; further details are given in \Cref{sec:method}.

Score estimation can be formulated as minimizing the $L_2$ distance between the two score functions:
\begin{equation}\label{eq:obj_cond_score}
    \min_{\phi} \,
    \E_{(\theta, \Xn) \sim p(\theta)\,\pn_\theta(\Xn)}
    \Big[\big\| s_\phi(\theta, \Xn) - \nabla_\theta \log \pn_\theta(\Xn) \big\|^2\Big],
\end{equation}
where $p(\theta)$ denotes the {density} from which $\theta$ is drawn, which can be either the prior $\pi(\theta)$ or a proposal density $q(\theta)$.

Although the true score function $\nabla_\theta \log \pn_\theta(\Xn)$ is unknown, the conditional score-matching objective can be rewritten in a form that does not require direct access to the true score \citep{hyvarinen2007some,song2019generative,meng2020autoregressive}, as stated below.

\begin{theorem}[Adopted from Theorem 1 in \citet{hyvarinen2005estimation}]\label{thm:score_matching}
    Under mild boundary condition and finite-moment assumptions (see \Cref{pf:score_matching}), the objective in \eqref{eq:obj_cond_score} is equivalent to
    \begin{align*}
    \min_\phi  \E_{ (\theta, \Xn)\sim  p(\theta)\,\pn_\theta(\Xn)} \Big[
    \norm{s_\phi(\theta, \Xn)}^2
    +2s_\phi (\theta, \Xn)^T \nabla_\theta \log p(\theta)
    +2 \text{tr}\left(\nabla_\theta s_{\phi}(\theta, \Xn)\right)
    \Big],
    \end{align*}
    where $\text{tr}(\cdot)$ denotes the trace of a matrix.
\end{theorem}

\noindent We defer the assumptions and proof to \Cref{pf:score_matching}. The original proof in \citet{hyvarinen2005estimation} addresses the unconditional score case. We extend this proof to the conditional likelihood score and discuss the required boundary condition.
}

\Cref{thm:score_matching} allows us to train the score model $s_\phi$ (typically
a neural network) by minimizing the empirical objective below, using a reference
table of size $N$, $\mD = \{(\theta^{(k)}, \Xn^{(k)})\}_{k=1}^N$, generated from the joint density $\pi(\theta)\,\pn_\theta(\Xn)$:
{\small
\begin{equation}\label{eq:score_loss}
\min_\phi \frac{1}{N}\sum_{k=1}^N \Big[
\frac{1}{2} \norm{s_\phi(\theta^{(k)}, \Xn^{(k)})}^2
+ s_\phi (\theta^{(k)}, \Xn^{(k)})^T \nabla_\theta \log \pi(\theta)\big|_{\theta=\theta^{(k)}}
+ \text{tr}\Big(\nabla_\theta s_{\phi}(\theta, \Xn^{(k)})|_{\theta=\theta^{(k)}} \Big)
\Big].
\end{equation}
}

{Once an estimator of the likelihood score is obtained, it can be incorporated into a gradient-based sampler.} In this paper, we use Langevin Monte Carlo (LMC) to sample from the
posterior distribution $\pi_n(\theta \mid \Xn^\ast)$, which is given by
\begin{equation}\label{eq:langevin}
	\theta^{(k+1)}_\text{\tiny LMC} = \theta^{(k)}_\text{\tiny LMC} + \tau_n \nabla_\theta \log p(\theta^{(k)} \mid \Xn^\ast)
	+ \sqrt{2\tau_n}\,U_k,
\end{equation}
where $\tau_n$ is a fixed step size and $U_k \iid \mN(0, I_{d_\theta})$.
LMC combines deterministic gradient-based updates with stochastic noise to
guide exploration toward regions of high posterior probability. The method is
well suited for unimodal posteriors; in multimodal cases, it can be extended
using techniques such as simulated annealing to locate modes and
then applying our approach locally around each mode. In this work, however,
we focus on LMC for technical simplicity.

We now introduce the baseline generic LMC algorithm, which is essentially plugging the estimate score from \eqref{eq:score_loss} into the LMC sampler. The algorithm is summarized in \Cref{alg:generic_LMC}.

\begin{algorithm}[ht]
    \caption{Generic Langevin Monte Carlo for SBI}\label{alg:generic_LMC}
    \begin{algorithmic}
            \STATE {\bf Input:} Proposal {density} $p(\theta)$, observed dataset $\Xn^\ast$, number of particles $N$, number of Langevin steps $K$, step size $\tau_n$, score network $s_\phi(\theta, \Xn)$, initial value $\theta^{(0)}$.
        \end{algorithmic}
        \hrule
        \begin{algorithmic}
            \STATE {\bf 1. Reference Table:} Generate $\mD = \{(\theta^{(k)}, \Xn^{(k)})\}_{k=1}^N\iid p(\theta)\,\pn_\theta(\Xn)$
            \STATE {\bf 2. Network Training:} Train the likelihood score model $s_\phi(\theta, \Xn)$ on $\mD$ and obtain $\widehat \phi$.
            \STATE {\bf 3. Langevin Sampling:}
  {ß For} {$k = 1$ to $K$}
                \STATE \hspace{2em} $\theta^{(k)}_\text{\tiny LMC} \gets \theta^{(k-1)}_\text{\tiny LMC}  + \tau_n  \Big(s_{\widehat\phi}\big(\theta^{(k-1)}_\text{\tiny LMC} , \Xn^\ast\big)+\nabla_\theta\log\pi(\theta^{(k-1)}_\text{\tiny LMC}) \Big) + \sqrt{2\tau_n}\,U_k$, \quad $U_k \iid \mathcal{N}(0, I_{d_\theta})$.
            \end{algorithmic}
            \hrule
            \begin{algorithmic}
            \STATE {\bf Return} $\{\theta^{(k)}_\text{\tiny LMC}\}_{k=1}^K$ as approximated posterior samples
    \end{algorithmic}
    \end{algorithm}
\vspace{-0.2in}
There are several ways to implement the second step of network training in \Cref{alg:generic_LMC}. To motivate our specialized designs in \Cref{sec:method}, we first present a naive baseline, in which the prior $\pi(\theta)$ is used as the proposal for $\theta$, and $s_\phi(\theta, \Xn)$ is estimated using standard score-matching techniques.

\vspace{-0.2in}

\section{Structured Score-based Langevin Dynamics for SBI}\label{sec:method}
\vspace{-0.1in}

{
While the naive implementation in \Cref{alg:generic_LMC} is conceptually straightforward, it is not sufficient for accurate and scalable inference in simulator-based models.

We focus on three issues that arise in the naive implementation. First, score estimation can be poor when the training distribution places little mass near $(\theta^*,\Xn^*)$ \citep{koehlerstatistical}. Second, naively training the score network on the full dataset $\Xn$ can lead to computational and statistical complexity that grows unfavorably with the sample size $n$. Third, first-order score matching alone does not guarantee accurate Langevin dynamics in a neighborhood of $\theta^*$: for stable posterior sampling, the estimated score must also capture the local second-order geometry of the likelihood.

In this section, we propose specialized procedures that address these challenges in two steps. \Cref{sec:preconditioning} addresses the first issue through localization, while \Cref{sec:regular} addresses the latter two through a structured score-matching scheme that exploits classical properties of likelihood scores.
}
For notational simplicity, we denote the true likelihood score function by
$s^\ast(\theta, \Xn) = \nabla_\theta \log \pn_\theta(\Xn)$.

\vspace{-0.25in}
\subsection{Localization scheme}\label{sec:preconditioning}
\vspace{-0.05in}
{
We begin with the first limitation of the generic algorithm: poor score accuracy caused by sparse training data in the neighborhood of $\theta^\ast$. In the generic algorithm, the reference table is generated from the prior, which in many SBI problems places little mass near $\theta^\ast$. However, the score used by LMC is needed only at $\Xn^\ast$ and for $\theta$ in a neighborhood of $\theta^\ast$. As a result, uniform score matching over the full proposal support can waste simulation effort on regions that are irrelevant for posterior sampling and lead to poor local score accuracy. We illustrate this phenomenon using a simple Beta-Binomial example in \Cref{sec:preconditioning_details}.
}

To address this challenge, we introduce a computationally efficient localization
step that constructs a proposal distribution concentrated around $\theta^*$. Score matching is then performed under this localized proposal, which significantly improves score estimation accuracy in the neighborhood of $\theta^\ast$ and thus enhances the performance of Langevin sampling.

We assume that the simulation process admits a reparametrized representationn $\Xn^\theta = \tau(\theta, \Zn)$, where $\Zn$ is drawn from a known latent distribution and $\tau$ is a deterministic map. This assumption is satisfied by most simulator-based models \citep{meeds2015optimization,brehmer2020mining,cranmer2020frontier}. For example, in the queuing model of \Cref{sec:queuing_simu}, $Z_i$ corresponds to the quantiles used to sample  $u_{ik}$ and $w_{ik}$.

Motivated
by the simulated method of moments \citep{mcfadden1989method,pakes1989simulation} {and Optimization Monte Carlo \citep{meeds2015optimization,ikonomov2020robust},}
we generate $B$ independent latent draws $\Zm^{(b)}$. For each draw, we solve
\begin{equation}\label{eq:smm_sol}
    \widehat \theta^{(b)}
    = \arg\min_\theta d_\text{SW}\big(\tau(\theta,\Zm^{(b)}), \Xn^\ast\big),
    \qquad b = 1,\ldots, B,
\end{equation}
where $d_\text{SW}(\cdot, \cdot)$ denotes the sliced Wasserstein distance (SWD) \citep{bonneel2015sliced}.
 We use this discrepancy because it yields a complexity that grows linearly with both the data dimension $p$ and parameter dimension $d_\theta$. Additionally, it also allows for comparison of datasets with different sample sizes and has favorable theoretical properties (see \Cref{sec:theory}).

\begin{figure}[!ht]
\centering
\includegraphics[width=0.8\textwidth,height=0.4\textwidth]{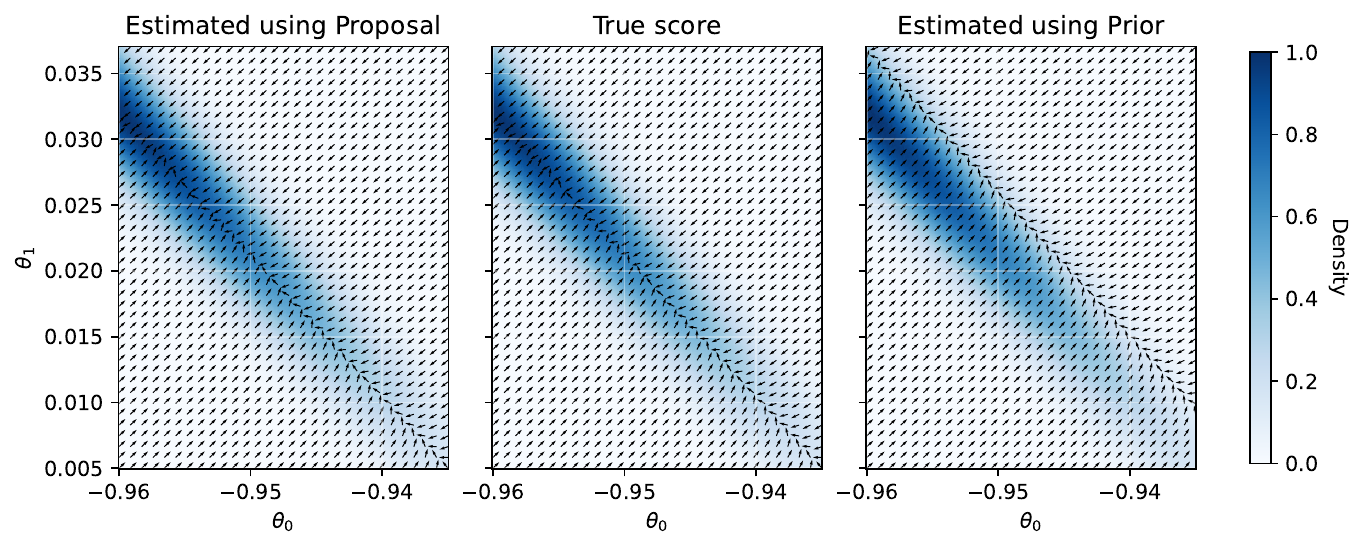}
\caption{\footnotesize Score estimation under the monotonic regression example in \Cref{sec:mono_reg}. Shown are score directions overlaid on the heatmap of the target posterior density on $(\theta_0, \theta_1)$. From left to right: scores estimated from the proposal $q(\theta)$, scores from the true likelihood, and scores estimated
from the prior $\pi(\theta)$.}\label{fig:score_direction}
\end{figure}

The resulting estimators $\{\widehat\theta^{(b)}\}_{b=1}^B$ are then summarized by a Gaussian proposal $q(\theta) = \mathcal{N}(\widehat\mu, \widehat\Sigma)$, where $\hat \Sigma$ is taken to be the diagonal part of the empirical covariance matrix. This localized proposal yields a more informative reference table for score matching and substantially improves the quality of the estimated score used in the subsequent Langevin sampler, as we illustrate in \Cref{fig:score_direction}. Further details of SWD and the localization step are provided in
\Cref{sec:preconditioning_details}. We also discuss the connection between our localized proposal and existing approaches in \Cref{sec:localization_discussion} and  non-diagonal Gaussian proposals in \Cref{sec:other_proposal}.

\subsection{Structured score matching}\label{sec:regular}

We now address the remaining two limitations of the generic algorithm: the unfavorable dependence on the sample size $n$, and the fact that first-order score matching alone does not fully capture the local geometry needed for stable LMC.

A distinctive feature of score-matching networks for Bayesian inference,
compared to networks that directly target generic conditional densities, is
that the true score function obeys universal structures from classical
statistical theory, which can be conveniently exploited to enhance estimation
efficiency and accuracy. For
simplicity, we focus on the case of i.i.d. observations, i.e.,
$\pn_\theta(\Xn) = \prod_{i=1}^n p_\theta(X_i)$. In this setting, the true
score function satisfies three fundamental properties:
1.~\emph{additive structure:}
   $s^\ast(\theta, \Xn) = \sum_{i=1}^n s^\ast(\theta, X_i)$.
2.~\emph{mean-zero structure:}
   $\E_{X_i \sim P_\theta}[s^\ast(\theta, X_i)] = 0$.
3.~\emph{curvature structure:}
   $\E_{X_i \sim P_\theta}\!\left[s^\ast(\theta, X_i) s^\ast(\theta, X_i)^T
   + \nabla_\theta s^\ast(\theta, X_i)\right] = 0$.
These hold for almost all likelihood functions under mild
conditions.\footnote{Note that standard conditions for the mean-zero and curvature structures require that the support of $X_i$ does not depend on $\theta$ (see, e.g., Chapter 2 of \cite{vanderVaart1998}).}

{
The additive structure is key to scalability with respect to the sample size $n$. The mean-zero structure controls the bias that accumulates when single-observation score estimates are summed across the dataset. The curvature structure enforces the local second-order geometry of the likelihood score and is essential for maintaining LMC accuracy in a neighborhood of $\theta^\ast$. Beyond these direct benefits, incorporating these statistical structures improves the generalizability of score-matching networks, which is crucial for the amortized training procedure in which a score network trained on a single observation is used to compute the score for the entire dataset.

We now describe how these three ingredients can be integrated into the score-matching scheme to improve scalability and accuracy. While we focus on the i.i.d.~case in this section, we also discuss in Appendix~\ref{sec:version1} how this approach can be generalized to dependent data settings, together with the corresponding theoretical results.
}

\subsubsection{Additive structure and single data score matching}\label{sec:additive_structure}

{
We begin with the additive structure. Instead of training a score network directly on the full dataset $\Xn$, we model the individual score function $s^*(\theta, X)$ by a network $s_\phi(\theta,X)$ and estimate the full-data score by summing the individual scores, i.e., $s_{\widehat\phi}(\theta, \Xn) = \sum_{i=1}^n s_{\widehat\phi}(\theta, X_i)$.

This structure has two advantages. First, it replaces a network whose input dimension grows with $n$ by a network defined on a single observation $X$, which significantly reduces the complexity of the score network and the number of simulations $N$ required for training. Second, it allows the reference table for score matching to be generated from single observations rather than full datasets. Specifically, we train on $\mD^S = \{(\theta^{(k)}, X^{(k)})\}_{k=1}^N \iid q(\theta)\,p_\theta(X)$, where $q(\theta)$ is the proposal distribution learned from the localization step in \Cref{sec:preconditioning}. This reduces the overall simulation cost from $O(Nn)$ to $O(N)$.

Therefore, the additive structure addresses the second issue raised at the beginning of the section: it makes score training substantially more scalable in the sample size $n$. Compared with more general exchangeable architectures such as Deep Sets \citep{zaheer2017deep}, it is also simpler to implement and train while still capturing the essential score structure.
}

However, single data score matching introduces a new difficulty. Even if $s_\phi(\theta, X)$ is a good estimator of $s^*(\theta, X)$, the error in the full data score estimate $s_{\widehat\phi}(\theta, \Xn) = \sum_{i=1}^n s_{\widehat\phi}(\theta, X_i)$ can accumulate with $n$. In the worst-case scenario,
the total score-matching loss
$\E \big[\big\| s_{\widehat\phi}(\theta, \Xn) - s^\ast(\theta, \Xn) \big\|^2 \big]
= \E \big[\big\| \sum_{i=1}^n \big( s_{\widehat\phi}(\theta, X_{i}) - s^\ast(\theta, X_{i}) \big) \big\|^2\big]$
may grow as large as $n^2$ times the single-observation score-matching loss
$\E \big[\big\| s_{\widehat\phi}(\theta, X_{1}) - s^\ast(\theta, X_{1}) \big\|^2\big]$.

\subsubsection{Mean-zero structure and debiasing}\label{sec:mean_zero_structure}

{To formally characterize the cumulative error issue in single data score matching,} we rewrite the overall score-matching
loss using the bias--variance decomposition as
{\small
\begin{equation*}
\E \big[\| s_{\widehat\phi}(\theta, \Xn) - s^\ast(\theta, \Xn) \|^2 \big]
= n \underbrace{\E \big[ \| s_{\widehat\phi}(\theta, X_{1}) - s^\ast(\theta, X_{1}) \|^2\big]}_{\text{Variance}}
+ n(n-1) \underbrace{\big\| \E [s_{\widehat\phi}(\theta, X_{1}) - s^\ast(\theta, X_{1})] \big\|^2}_{\text{Bias}^2}.
\end{equation*}}

In this decomposition, the variance term is $n$ times the single-observation
score-matching loss, which scales at most linearly in $n$. In contrast, the bias
term leads to quadratic growth in $n$ of the full-data score-matching loss.
Fortunately, due to the mean-zero structure of the true score $s^\ast$, the bias
term simplifies to $\text{Bias} = \| \E[s_{\widehat\phi}(\theta, X_{1})] \|$,
which can be explicitly computed and controlled (see below). This shows that enforcing the mean-zero structure
on the score-matching network can effectively control the bias term, thereby
providing a way to bypass the exploding cumulative score-matching error issue.

In this work, we introduce a post-processing debiasing step to center the estimated score-matching network by subtracting the expectation
$\widehat h(\theta) := \E_{X_{1}\sim p_\theta}[s_{\widehat\phi}(\theta, X_{1})]$ from
$s_{\widehat\phi}(\theta, x)$.
Specifically, we first train the score-matching network on the reference table $\mD^S$ using a loss function in \eqref{eq:score_loss}.
We then fit a regression model to approximate the mean $\widehat h(\theta)$ of $s_{\widehat\phi}(\theta, X)$
using another neural network $h_\psi(\theta): \R^{d_\theta} \to \R^{d_\theta}$,
parameterized by $\psi$. The corresponding \emph{mean-matching} optimization objective is
{
\begin{equation}\label{eq:score_demean_noreg}
\widehat \psi =\arg\min_{\psi} \mathbb{E}_{q(\theta)} \bigg[\normbig{h_\psi(\theta)- \E_{p_\theta} \big[s_{\widehat\phi} (\theta, X)\big] }^2\bigg].
\end{equation}
To implement \eqref{eq:score_demean_noreg}, we construct a separate regression
reference table $\mD^R = \{(\theta^{(l)}, \X^{(l)}_{m_R})\}_{l=1}^{N_R}$, where
$\X^{(l)}_{m_R}$ is a simulated dataset of size $m_R$ generated from $\theta^{(l)}$,
and the expectation $\E_{p_\theta}$ can thus be approximated by an empirical average.
}

The following lemma shows that incorporating this debiasing step never hurts the accuracy of the full-data score approximation.
\begin{lemma}\label{lem:debias_error}
The debiasing step never increases the score-matching error, i.e.
\[
\mathbb{E}_{(\theta,X)\sim q(\theta)p_\theta(X)} \Big[\big\| s_{\hat\phi}(\theta, X) - h_{\hat\psi}(\theta) - s^\ast(\theta, X) \big\|^2 \Big] \le \mathbb{E}_{(\theta,X)\sim q(\theta)p_\theta(X)} \Big[ \big\| s_{\hat\phi}(\theta, X) - s^\ast(\theta, X) \big\|^2 \Big].
\]
\end{lemma}
The intuition is straightforward: since $h_\psi(\theta) \equiv 0$ is always a feasible solution to \eqref{eq:score_demean_noreg}, the minimizer $h_{\hat\psi}$ can only further reduce the score-matching bias. A detailed proof is provided in \Cref{sec:debias_error_pf}. This lemma also plays an important role in guaranteeing that the overall
posterior approximation error in \Cref{thm:post_convergence_iid} in the next section scales linearly, rather than quadratically, in $n$. We also provide a numerical illustration of the benefit of debiasing in \Cref{fig:benefit_debias}.

\begin{figure}[!ht]
\centering
\includegraphics[width=0.8\textwidth,height=0.4\textwidth]{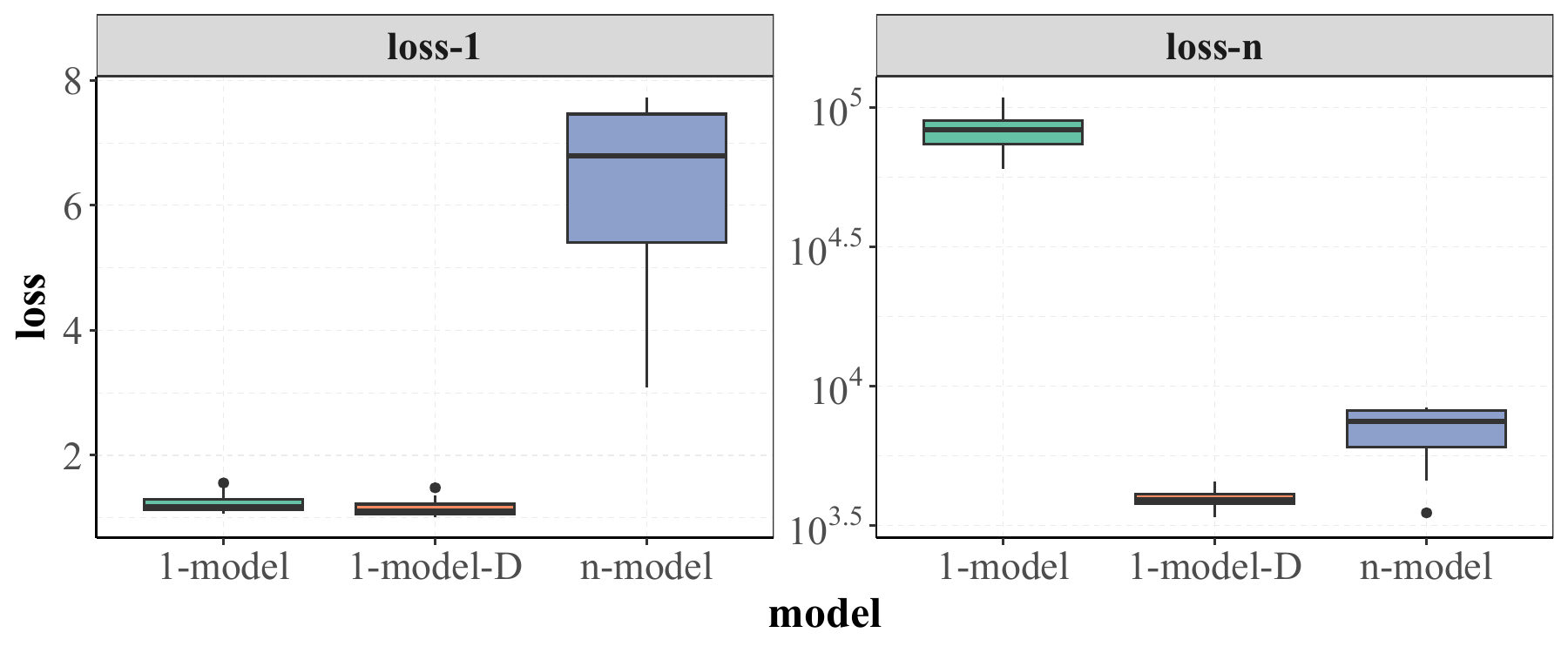}
\caption{\footnotesize Score estimation losses on a single observation (loss-1) and on $n$ observations (loss-n) under the monotonic regression example in \Cref{sec:mono_reg}. ``1-model'' and ``1-model-D'' refer to the models trained on single observation, without and with the debiasing step, respectively. ``n-model'' refers to the model trained on $n$ observations.}\label{fig:benefit_debias}
\end{figure}

{\begin{remark}[Scalable structure for i.i.d. data]\label{rmk:iid_structure}  For i.i.d., or more generally exchangeable, data, many existing SBI methods process the samples through a generic permutation-invariant representation, such as Deep Sets \citep{zaheer2017deep} or related embedding networks \citep{deistler2025simulation}. These architectures are flexible, but typically rely on a learned summation-based encoder. In favorable cases, such an encoder may approximate an informative summary statistic, but there is generally no guarantee that it recovers a sufficient statistic, especially when the latter is not of summation form. Moreover, this strategy does not reduce the simulation cost of generating the reference table. In contrast, our method builds the i.i.d. structure directly into the score model through the exact additive decomposition of the full-data likelihood score into a sum of single-observation scores.

A different line of work, including NLE \citep{papamakarios2019sequential}, NRE \citep{hermans2020likelihood}, and compositional score models \citep{geffner2023compositional,linhart2024diffusion}, leverages the i.i.d. structure by learning a likelihood, likelihood ratio, or posterior-related score for a single observation and then combining these contributions across samples. However, these methods do not provide a transparent mechanism for controlling how approximation errors accumulate with the sample size. Our approach addresses this issue directly: the debiasing step enforces the mean-zero identity of the single-observation likelihood score, and our theory tracks how the aggregated score error propagates to the posterior approximation. This explicit control of error accumulation is a key distinction of our method.
\end{remark}

}

\subsubsection{Curvature structure and local geometric accuracy}\label{sec:curvature_structure}

{
While the additive and mean-zero structures address scalability and bias accumulation, they do not fully resolve the third issue raised at the beginning of this section: first-order score matching alone is not sufficient to guarantee accurate LMC.

To see this, note that, by Bayes' rule, the standard score-matching objective in \eqref{eq:obj_cond_score} can be rewritten as
\[
\E_{(\Xn,\theta) \sim p(\Xn)\,p(\theta\,|\,\Xn)}
    \Big[\big\| s_\phi(\theta, \Xn) - \nabla_\theta \log \pn_\theta(\Xn) \big\|^2\Big].
\]
Standard posterior concentration results suggest that $p(\theta\,|\,\Xn)$ tends to concentrate around the parameter value that generated the data $\Xn$. Consequently, this objective primarily controls the pointwise score error near that parameter value. When the learned score is applied to the observed data $\Xn^\ast$, the standard score-matching objective mainly controls the score estimation error in the region where the training posterior concentrates, namely very close to $\theta^\ast$, i.e., $\norm{\hat s(\theta^\ast, \Xn^\ast)- s(\theta^\ast, \Xn^\ast)}$.

In contrast, the LMC trajectory may visit a larger neighborhood around $\theta^\ast$ than the region in which the pointwise score error is directly controlled by posterior concentration. Thus, accurate LMC requires the estimated score to remain accurate not only near $\theta^\ast$ but also throughout this neighborhood. In such a neighborhood, the accuracy of the estimated score is affected by both the pointwise score error at $\theta^\ast$ and the local behavior of the score as $\theta$ moves away from $\theta^\ast$. This can be seen from the following local linear approximations:
\begin{align*}
s(\theta, \Xn^\ast) &\approx s(\theta^\ast, \Xn^\ast)+ \nabla_\theta s(\theta, \Xn^\ast)\mid_{\theta=\theta^\ast}(\theta - \theta^\ast), \\
\hat s(\theta, \Xn^\ast)& \approx \hat s(\theta^\ast, \Xn^\ast) + \nabla_\theta \hat s(\theta, \Xn^\ast)\mid_{\theta=\theta^\ast}(\theta - \theta^\ast).
\end{align*}
Therefore, accurate Langevin dynamics depends not only on the pointwise score error, but also on whether the estimated score preserves the local geometry of the true score. This local geometry is governed by the discrepancy
$\norm{\nabla_\theta \hat s(\theta, \Xn^\ast)\mid_{\theta=\theta^\ast}- \nabla_\theta s(\theta, \Xn^\ast)\mid_{\theta=\theta^\ast}}$.
Stable and accurate posterior sampling therefore requires controlling both the score itself and its local curvature. A related idea of incorporating higher-order information in score estimation was considered by \citet{lu2022maximum}, although their goal and formulation differ from ours. A rigorous analysis of this point is provided in \Cref{thm:post_convergence_iid} in the next section.

\begin{figure}[!ht]
    \centering
    \includegraphics[width=0.8\textwidth]{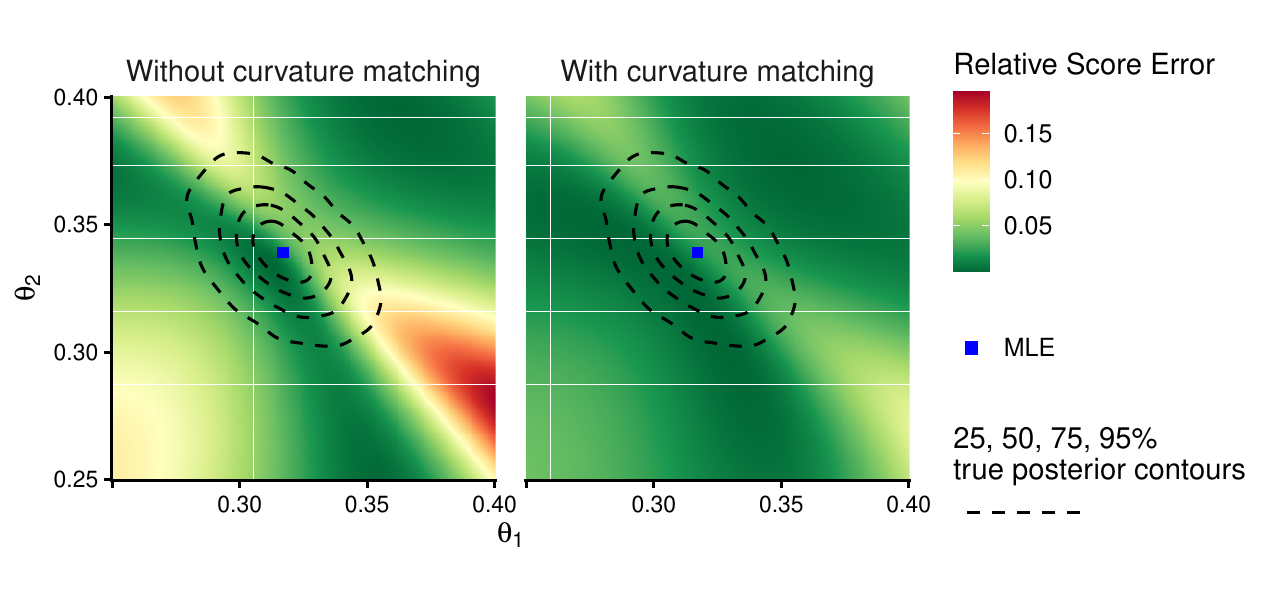}
    \caption{Relative error of $\normsm{s_{\hat\phi}(\theta, \Xn^*)- s^*(\theta, \Xn^*)}$ under a Dirichlet-Multinomial with 3 categories example (details in \Cref{sec:curvature_benefit}). The left panel shows the results from the score network without the curvature regularization, and the right panel shows the results with the curvature regularization.}\label{fig:relative_error_heatmap}
\end{figure}

We also provide an empirical illustration of this benefit using a Dirichlet-Multinomial example with three categories in \Cref{fig:relative_error_heatmap}; a detailed description is given in \Cref{sec:curvature_benefit}. The parameter $\theta$ is two-dimensional. In this example, we train two score networks, with and without curvature regularization, and evaluate the relative error of the estimated score $\norm{s_{\hat\phi}(\theta, \Xn^\ast)- s^\ast(\theta, \Xn^\ast)}$, normalized by $\norm{s^\ast(\theta, \Xn^\ast)}$, in a neighborhood of $\theta^\ast$. The results show that curvature regularization significantly improves the local accuracy of the estimated score when $\theta$ is near $\theta^\ast$. In \Cref{sec:curvature_benefit}, we also provide two additional plots showing how the relative score error changes with the distance $\norm{\theta-\theta^\ast}$ and how this affects the marginal density of the approximated posterior. These results further illustrate the importance of the curvature structure in improving local score accuracy and, consequently, the quality of the posterior approximation.

Motivated by this observation,
we impose the curvature identity satisfied by the true likelihood score as a regularization penalty in the score-matching objective.
Concretely, we train the single-observation score network using the following loss function, which combines the standard score-matching loss with a curvature-matching penalty:
{\small\begin{equation}\label{eq:score_loss_struct_single}
    \min_\phi  \E_{q(\theta)}  \bigg[ \E_{p_\theta}  \Big[
    \underbrace{\|s_\phi(\theta,X)- s^\ast(\theta,X)\|^2}_{\text{score-matching loss on a single $X$}} \Big]
    + \lambda_1 \,\underbrace{\Big\| \E_{ p_\theta}\!\big[s_\phi(\theta,X)s_\phi(\theta,X)^T
    + \nabla_\theta s_\phi(\theta,X)\big]\Big\|_F^2}_{\text{curvature-matching loss}} \bigg],
\end{equation}}
where $\lambda_1>0$ is a hyperparameter controlling the strength of the curvature penalty and $\norm{\cdot}_F$ denotes the Frobenius norm. In practice, we approximate the expectation in the curvature penalty by an empirical average.
}

{

This penalty encourages the estimated score to capture the correct local second-order geometry of the likelihood, which in turn improves the stability and accuracy of the resulting Langevin dynamics. In our implementation, we also use a curvature-aware version of the debiasing step, obtained by modifying \eqref{eq:score_demean_noreg}, so that the final debiased score function $\wt s(\theta, X) = s_{\widehat\phi}(\theta, X) - h_{\widehat\psi}(\theta)$ continues to preserve the curvature structure. The details of the modified loss and a discussion of alternative approaches are provided in \Cref{sec:alter_ways_demean}.
}

We provide an algorithm view of the overall procedure with all three structures and more implementation details in \Cref{sec:regularization_details}.
\begin{remark}[Generalization to dependent datasets] The debiasing idea here can be generalized to weakly dependent settings (see a detailed discussion in \Cref{sec:regularization_details}). For more general dependent data settings, one has to resort to full data score matching. In this case, the curvature structure is still helpful in improving the stability of the Langevin sampling procedure. We provide the full data score matching alternative and its theoretical analysis in \Cref{sec:version1}.
\end{remark}

\vspace{-2em}

\section{Theoretical Results}\label{sec:theory}
In this section, we study the theoretical properties of our proposed methods in
the i.i.d.\ setting and how the convergence of the approximated posterior to the true posterior is affected by different components of score-matching training.
In particular, this analysis highlights why incorporating the curvature and
mean-zero structures into the score-matching networks is essential for ensuring
both the accuracy and stability of the posterior approximation.
We also provide a theoretical analysis on how localization step enables
rapid identification of a neighborhood around $\theta^\ast$ in \Cref{sec:rate_loc} and an analysis of the full data score matching alternative in \Cref{sec:version1}.

For notational simplicity, we denote the estimated score function by
$\widehat{s}(\theta, \Xn)$, where $\widehat{s}(\theta, \Xn) = \sum_{i=1}^n \widehat{s}(\theta, X_i)$.
Recalling the Langevin sampling step in \eqref{eq:langevin},
we denote the approximated posterior distribution after $k$ steps by
$\widehat{\pi}^{k\tau_n}(\theta \mid \Xn^\ast)$, and the final approximated
posterior by $\widehat{\pi}_n(\theta \mid \Xn^\ast) = \widehat{\pi}^{K\tau_n}(\theta \mid \Xn^\ast)$.
Similarly, we denote by $\pi^{k\tau_n}(\theta \mid \Xn^\ast)$ the distribution of
the Langevin sampler using the true score function $s^\ast(\theta, \Xn)$ after
$k$ steps.

 Using the triangle inequality, we can bound the posterior approximation error under the total variation distance as $d_\text{TV}(\widehat \pi_n, \pi_n) \ \leq \ d_\text{TV}( \pi_n, \pi^{K\tau_n}) + d_\text{TV}(\pi^{K\tau_n},\widehat \pi^{K\tau_n} )$.
The discretization error, $d_\text{TV}( \pi_n, \pi^{K\tau_n})$, arises from the Euler-Maruyama discretization of the
Langevin diffusion with the true score and is unaffected by the choice of
score-matching strategy, whereas the score error, $d_\text{TV}(\pi^{K\tau_n},\widehat \pi^{K\tau_n} )$, results from replacing the true
score with the estimated score function in the drift term.

{For simplicity of presentation, we defer the formal statement of the standard regularity conditions to \Cref{sec:regularity_conditions_thm}, including the concentration of the true posterior and maximum likelihood estimator (MLE), the Lipschitz continuity of the true score function,  the log-Sobolev inequality for the posterior distribution, and the regularity conditions on  true and  estimated scores.}

Our final assumption concerns the score matching error,
which can be controlled by properly choosing the size
of the score network to optimally balance
the approximation error and the generalization bound.
Similar bounds have been extensively studied
in the statistical learning literature; see, for example, \cite{oko2023diffusion_minimax,tang2024conditional_diffusion_minimax}.

\begin{assumption}[Uniform score-matching error]\label{ass:uniform_sm_err_single}
Define the set $\mA_{n,1}:= \{\theta: \norm{\sqrt n(\theta- \theta^\ast)}_2 < {C_0}\sqrt{\log n}\}$ for $C_0 = \max \Big\{1, \sqrt{\frac{6}{C_1}} \Big\}$. The score-matching error, curvature-matching error and mean-matching error are all uniformly bounded as
\begin{align*}
& \wt\varepsilon_{N,1}^2:= \sup_{\theta\in \mA_{n,1}} \E_{X\sim P_{\theta}} \norm{\widehat s(\theta, X)- s^\ast(\theta, X)}^2 & \text{\rm(score-matching error)}\\
& \wt\varepsilon_{N_R,m_R,2}^2:=\sup_{\theta\in \mA_{n,1}} \norm{\E_{X\sim P_\theta} \big[ \nabla_\theta \widehat s(\theta, X) +\widehat s(\theta, X)\widehat s(\theta, X)^T \big] }_F^2 &\text{\rm(curvature-matching error)}\\
& \wt\varepsilon_{N_R,m_R,3}^2:=\sup_{\theta\in \mA_{n,1}} \norm{\E_{X\sim P_\theta} \widehat s(\theta, X)}^2 & \text{\rm(mean-matching error)}.
\end{align*}
\end{assumption}
Score matching error bounds are usually averaged over the sampling distribution
$q(\theta)$ of $\theta$.
According to the localization error bound established in \Cref{thm:precond_convergence},
the proposal distribution $q(\theta)$ is guaranteed to concentrate
in an $n^{-1/2}$ neighborhood of $\theta^\ast$.
This motivates us to localize the uniform estimation error bound
to the set $\mA_{n,1}$. Here the score-matching error depends on the complexity of the true score function and the size $N$ of the reference table $\mD^S$. For the curvature-matching error and mean-matching error, they are assessed using the reference table $\mD^R$ of size $(N_R, m_R)$, where the Monte Carlo approximation of expectations introduces an error that decays at rate  $1/\sqrt{m_R}$.

Denote the Fisher information matrix as $I(\theta):=\E_{P_\theta}[-\nabla_\theta s^*(\theta, X)]$ and the chi-square divergence between two distributions $P$ and $Q$ as $d_{\chi^2}(P||Q)=:\E_{Q} \big[\big(\frac{p(x)}{q(x)}-1\big)^2\big]$. We also write $f(x)\lesssim g(x)$ if there exists a constant $C>0$ such that $f(x)\leq C\cdot g(x)$.

\begin{theorem}[Posterior approximation error under single data score matching]\label{thm:post_convergence_iid}
Suppose \Cref{ass:uniform_sm_err_single} and the regularity conditions in \Cref{sec:regularity_conditions_thm}
hold and assume $\|I(\theta^*)\|_F<\infty$. If the step size $\tau_n$ and initial distribution of the Langevin Monte Carlo satisfy
 \[
 \tau_n = O\Big( \frac{1}{d_\theta C_{LSI}\lambda_L^2n}\Big) \quad  \mbox{and}\quad d_{\chi^2}(\widehat \pi^0_n\big(\cdot\mid \Xn^\ast), \pi_n(\cdot\mid \Xn^\ast)\big)\leq \eta_\chi^2,
\]
where $\eta_\chi>0$ is a constant, then we have
{\footnotesize\[
\E_{\Xn^\ast\sim \Pn_{\theta^\ast}}\Big[d_\text{TV}^2\big(\widehat\pi_n(\cdot\mid \Xn^\ast), \pi_n(\cdot\mid \Xn^\ast)\big) \Big]\ \lesssim \ \underbrace{\exp\Big(-\frac{K n \tau_n}{5C_\text{LSI}}\Big)\, \eta_\chi^2}_{\text{burn-in error}}\ + \ \underbrace{d_\theta\, C_\text{LSI} \,\lambda_L^2\, n\,\tau_n}_\text{discretization error}\ +\ \underbrace{\varepsilon_n (Kn \tau_n +\eta_\chi C_\text{LSI})}_\text{score error},
\]}
where $\varepsilon_n^2 :\,= \widetilde{\varepsilon}_{N,1}^2 (\log n)^2  +\widetilde{\varepsilon}_{N_R,m_R,2}^2  (\log n)^2 + n \,\widetilde{\varepsilon}_{N_R,m_R,3}^2\log n  + n^{-1} (\log n)^3$.
\end{theorem}

The proof is provided in \Cref{sec:post_convergence_iid_pf} and $C_\text{LSI}, \lambda_L$ are constants defined in the regularity conditions in \Cref{sec:regularity_conditions_thm}. The first burn-in error corresponds to the mixing bound of the continuous-time Langevin dynamics
run up to time $k\tau_n$.
The additional factor of $n$ in the exponent arises from
Assumption~\ref{ass:log_sobolev}, which implies that
the log-Sobolev constant of the posterior is $C_{\text{LSI}}/n$. \Cref{thm:post_convergence_iid} also suggests that we should choose the stepsize $\tau_n = O(\frac{1}{n})$ to control the discretization error. In practice, one can simply choose the initial distribution $\widehat\pi_n^0 (\cdot\mid \Xn^\ast)$ as the proposal distribution $q(\theta)$.

The score error is determined by three sources of error in the score estimation process.  To ensure a diminishing error $\varepsilon_n=o(1)$ as $n\to\infty$, it suffices for the score-matching error to decay at the rate $\widetilde{\varepsilon}_{N,1} = \mathcal{O}\!\left(\tfrac{1}{\log n}\right)$.  The Monte Carlo errors, $\widetilde{\varepsilon}_{N_R,m_R,2}$ and $\widetilde{\varepsilon}_{N_R,m_R,3}$, both scale as $\mO(1/\sqrt{m_R})$.  Thus, controlling these terms requires $m_R = \mathcal{O}(n \log n)$. In contrast, if we directly match the single-data score without the debiasing step,
it can be shown that the corresponding score error term then takes the form of $\widetilde{\varepsilon}_{n}^2:\,=n\widetilde{\varepsilon}_{N,1}^2(\log n)^2 + \widetilde{\varepsilon}_{N_R,m_R,2}^2 (\log n)^2+n^{-1} (\log n)^3$. This forces a much stricter condition
 $\widetilde{\varepsilon}_{N,1} = \mathcal{O}\!\left(\frac{1}{\sqrt n\log n}\right)$  in order to control $\widetilde{\varepsilon}_{n}$. The remark below indicates that pushing the score-matching error beyond the root-$n$ rate could lead to exponential sample complexity in $d_\theta$. This is consistent with our observation that the non-debiased method requires far more samples than its debiased counterpart.

\begin{remark}[Score-matching error] \label{rmk:sm_err}
For both approaches, the convergence rate of the approximated posterior
is governed by the decay rate of the score-matching error.
Under expressive neural networks, the score-matching error typically scales as
$L^{\frac{d}{2\beta+d}}N^{-\frac{\beta}{2\beta+d}}$,
where $N$ denotes the score matching sample size, $L$ the radius of the input domain,
$\beta$ the smoothness of the true score function,
and $d$ the input dimension
(see \cite{Zuowei_Shen_2020,schmidt2020nonparametric};
for general nonparametric estimation error and its dependence on $L$,
see, e.g., \cite{yang2015minimax_optimal_nonparametric_regression_high_dim}).
Our localization step reduces $L$ from $O(1)$ to $O(n^{-1/2})$.
Thus, by taking $N$ to be of the same order as $n$,
one can guarantee that the score-matching error is of order
$n^{-\frac{d/2}{2\beta+d}} \cdot n^{-\frac{\beta}{2\beta+d}} = n^{-1/2}$,
which does not suffer from the curse of dimensionality in the error exponent.
In addition, when the score function admits a low-dimensional structure,
$d$ can be replaced by the intrinsic dimension, leading to faster rates
\citep{bauer2019deep}.
Another advantage of our score network construction in \Cref{alg:langevin_single_obs}
is that the additive structure we impose reduces the input dimension for data
from $np$ to $p$ in the score-matching step.
Consequently, the effective input dimension decreases from $d_\theta + np$
to $d_\theta + p$, thereby further improving the scalability of our method
and yielding more favorable approximation behavior in practice.
\end{remark}

\vspace{-0.2in}
\section{Empirical Analysis}\label{sec:simulation}
In this section, we conduct a series of simulation studies to evaluate the performance of our proposed method.  We look into three examples, including (1) M/G/1-queuing model, which is a low-dimensional SBI benchmark model, (2) Bayesian monotonic regression, which is a high-dimensional model with a known posterior distribution, and
{(3) an mRNA transfection model with real data application. We also include an additional example of a stochastic epidemic model in \Cref{sec:sto_epi_details}.}

We compare our method with existing SBI methods, including ABC using 1-Wasserstein distance \citep{bernton2019approximate}, BSL \citep{price2018bayesian}, NPE \citep{papamakarios2016fast,lueckmann2017flexible}, {and NLE  \citep{papamakarios2019sequential}}, unless otherwise noted. {To better handle i.i.d. data, we implement NPE with a permutation-invariant embedding network \citep{radev2020bayesflow}, and choose NLE to estimate the likelihood contribution of each individual observation, which are then aggregated across all observations at inference time following \citet{papamakarios2019sequential}. For these methods, we use the same training distribution as the proposal distribution $q(\theta)$ for our method, whenever the proposal distribution is different from the prior.}

{Additionally, we also explore the sequential NPE (SNPE) in the first two examples. SNPE performs significantly better than amortized NPE in the queuing model, which is a low-dimensional example and no localization is applied. However, SNPE does not perform as well as NPE with the localized proposal distribution in the monotonic regression example, which is higher dimensional. We find that its first round trained with the prior can be estimated poorly, leading to worse performance in subsequent rounds (details in \Cref{sec:comp_seq_methods}.)}

For our methods, we include both the version with full data score matching (details in \Cref{sec:version1}), referred as n-model, and the version with single data score matching in \Cref{sec:regular}, referred as single-model.
  To compare the performance of different methods, we report: (1) average estimation bias $|\E(\widehat\theta)-\theta^\ast|$, (2) average width of the 95\% credible interval (CI95\_width), and (3) average coverage of the 95\% credible interval (CI width).
Details of implementation for all simulation examples are provided in \Cref{sec:additional_implementation}.

\vspace{-0.2in}

\subsection{M/G/1-queueing Model}\label{sec:queuing_simu}
\vspace{-0.1in}
We begin by applying our method to the M/G/1-queuing model, a classic example in the ABC literature. This model uses $3$ parameters $\theta = (\theta_1, \theta_2, \theta_3)$ to simulate customers' interdeparture times in a single-server system. We adopt the same setting as in \citet{jiang2018approximate}. We observe 500 independent time-series observations. Each observation is a 5-dimensional vector of inter-departure times $x_i = (x_{i1},x_{i2},x_{i3},x_{i4},x_{i5})^T$. In this model, the service times $u_{ik}\sim U[\theta_1,\theta_2]$ and the arrival times $w_{ik} \sim \text{Exp}(\theta_3)$. The observed  inter-departure times $X_i$ are given by the process $x_{ik} =u_{ik} + \max(0,\sum^k_{j=1} w_{ij}-\sum^{k-1}_{j=1} x_{ij})$.  The prior on $(\theta_1,\theta_2-\theta_1, \theta_3)$ is uniform on $[0,10]^2 \times [0.01,0.5]$.
The observed dataset $\Xn^\ast$ is generated under $\theta^\ast = (1, 5, 0.2)$. Since $\theta$ is low-dimensional here, we skip the localization step and directly use the prior to generate the reference table $\mD$ for all methods.

One point worth mentioning is that this model violates the boundary condition required in \Cref{ass:boundary}. There are two reasons: the prior density is uniform and not vanishing at boundary, and the support of $\theta_1$ depends on the data as it is easy to verify $\theta_1\leq \min_{i,j} \{x_{ij}\}$. This requires special treatments since the objective function in \eqref{eq:score_loss} is no longer valid. We consider two solutions in this work for this issue. First is to introduce a weight function $g(\theta, \Xn)$ such that the elementwise joint product $s_\phi(\theta, \Xn)\odot g(\theta, \Xn)$  can satisfy \Cref{ass:boundary}. We apply that to our n-model here. The second solution is to perturb the data with a random Gaussian noise to resolve the dependency between supports. A more detailed investigation is provided in \Cref{sec:boundary_sol}.

We repeat the experiment $100$ times
 (with distinct $\Xn^\ast$'s). The averaged results are reported in \Cref{queuing_10exp}, and a density plot of the approximated posterior in one experiment is shown in \Cref{fig:queuing}. We observe that our methods have smaller errors and tighter credible intervals
{than other methods.}
In particular, the n-model is doing exceptionally well on $\theta_1$ so we put in a separate plot. This is because the weight function $g(\theta, \Xn)$ supply the information that $\theta_1\leq \min \{x_{ij}\}$  and when $n=500$, the upper bound is almost 1. Other than that, the single-model
is performing better than the n-model.

\renewcommand{\arraystretch}{0.9}

\begin{table}[!ht]
\centering
\caption{Averaged results over 100 experiments under M/G/1-queuing model. We report the standard deviations of the statistics under the average.}
\resizebox{\textwidth}{!}{
\begin{tabular}{l|ccc|ccc|ccc}
\toprule
 & \multicolumn{3}{c|}{$\theta^\ast_{1} = 1$} & \multicolumn{3}{c|}{$\theta^\ast_{2} = 5$} & \multicolumn{3}{c}{$\theta^\ast_{3} = 0.2$} \\
 & $\lvert \widehat{\theta}_1 - \theta^\ast_{1} \rvert$ & CI95 Width & Cover95 & $\lvert \widehat{\theta}_2 - \theta^\ast_{2} \rvert$ & CI95 Width & Cover95 & \begin{tabular}{@{}c@{}} $\lvert \widehat{\theta}_3 - \theta^\ast_{3} \rvert$ \\ ($\times 10^{-2}$) \end{tabular} & CI95 Width & Cover95 \\
\midrule
single-model & \begin{tabular}{@{}c@{}} 0.025 \\ {\footnotesize (0.018)} \end{tabular} & \begin{tabular}{@{}c@{}} 0.152 \\ {\footnotesize (0.016)} \end{tabular} & \begin{tabular}{@{}c@{}} 0.98 \end{tabular} & \begin{tabular}{@{}c@{}} \textbf{0.044} \\ {\footnotesize (0.035)} \end{tabular} & \begin{tabular}{@{}c@{}} \textbf{0.283} \\ {\footnotesize (0.035)} \end{tabular} & \begin{tabular}{@{}c@{}} 0.97 \end{tabular} & \begin{tabular}{@{}c@{}} 0.379 \\ {\footnotesize (0.284)} \end{tabular} & \begin{tabular}{@{}c@{}} \textbf{0.017} \\ {\footnotesize (0.001)} \end{tabular} & \begin{tabular}{@{}c@{}} 0.94 \end{tabular} \\
\midrule
n-model & \begin{tabular}{@{}c@{}} \textbf{0.002} \\ {\footnotesize (0.002)} \end{tabular} & \begin{tabular}{@{}c@{}} \textbf{0.011} \\ {\footnotesize (0.002)} \end{tabular} & \begin{tabular}{@{}c@{}} 0.98 \end{tabular} & \begin{tabular}{@{}c@{}} 0.107 \\ {\footnotesize (0.082)} \end{tabular} & \begin{tabular}{@{}c@{}} 0.510 \\ {\footnotesize (0.056)} \end{tabular} & \begin{tabular}{@{}c@{}} 0.95 \end{tabular} & \begin{tabular}{@{}c@{}} 0.420 \\ {\footnotesize (0.356)} \end{tabular} & \begin{tabular}{@{}c@{}} 0.021 \\ {\footnotesize (0.002)} \end{tabular} & \begin{tabular}{@{}c@{}} 0.93 \end{tabular} \\
\midrule
ABC & \begin{tabular}{@{}c@{}} 0.594 \\ {\footnotesize (0.075)} \end{tabular} & \begin{tabular}{@{}c@{}} 2.953 \\ {\footnotesize (0.133)} \end{tabular} & \begin{tabular}{@{}c@{}} 1.00 \end{tabular} & \begin{tabular}{@{}c@{}} 0.259 \\ {\footnotesize (0.136)} \end{tabular} & \begin{tabular}{@{}c@{}} 4.117 \\ {\footnotesize (0.198)} \end{tabular} & \begin{tabular}{@{}c@{}} 1.00 \end{tabular} & \begin{tabular}{@{}c@{}} 1.172 \\ {\footnotesize (0.529)} \end{tabular} & \begin{tabular}{@{}c@{}} 0.058 \\ {\footnotesize (0.004)} \end{tabular} & \begin{tabular}{@{}c@{}} 1.00 \end{tabular} \\
\midrule
BSL & \begin{tabular}{@{}c@{}} 0.328 \\ {\footnotesize (0.223)} \end{tabular} & \begin{tabular}{@{}c@{}} 2.663 \\ {\footnotesize (0.592)} \end{tabular} & \begin{tabular}{@{}c@{}} 1.00 \end{tabular} & \begin{tabular}{@{}c@{}} 0.399 \\ {\footnotesize (0.258)} \end{tabular} & \begin{tabular}{@{}c@{}} 3.638 \\ {\footnotesize (1.233)} \end{tabular} & \begin{tabular}{@{}c@{}} 1.00 \end{tabular} & \begin{tabular}{@{}c@{}} 0.460 \\ {\footnotesize (0.343)} \end{tabular} & \begin{tabular}{@{}c@{}} 0.029 \\ {\footnotesize (0.023)} \end{tabular} & \begin{tabular}{@{}c@{}} 0.96 \end{tabular} \\
\midrule
NPE & \begin{tabular}{@{}c@{}} 0.067 \\ {\footnotesize (0.051)} \end{tabular} & \begin{tabular}{@{}c@{}} 0.465 \\ {\footnotesize (0.084)} \end{tabular} & \begin{tabular}{@{}c@{}} 1.00 \end{tabular} & \begin{tabular}{@{}c@{}} 0.201 \\ {\footnotesize (0.147)} \end{tabular} & \begin{tabular}{@{}c@{}} 1.018 \\ {\footnotesize (0.184)} \end{tabular} & \begin{tabular}{@{}c@{}} 0.99 \end{tabular} & \begin{tabular}{@{}c@{}} 0.448 \\ {\footnotesize (0.368)} \end{tabular} & \begin{tabular}{@{}c@{}} 0.021 \\ {\footnotesize (0.002)} \end{tabular} & \begin{tabular}{@{}c@{}} 0.94 \end{tabular} \\
\midrule
NLE & \begin{tabular}{@{}c@{}} 0.025 \\ {\footnotesize (0.017)} \end{tabular} & \begin{tabular}{@{}c@{}} 0.080 \\ {\footnotesize (0.009)} \end{tabular} & \begin{tabular}{@{}c@{}} 0.80 \end{tabular} & \begin{tabular}{@{}c@{}} 0.069 \\ {\footnotesize (0.050)} \end{tabular} & \begin{tabular}{@{}c@{}} 0.223 \\ {\footnotesize (0.046)} \end{tabular} & \begin{tabular}{@{}c@{}} 0.78 \end{tabular} & \begin{tabular}{@{}c@{}} 0.510 \\ {\footnotesize (0.384)} \end{tabular} & \begin{tabular}{@{}c@{}} 0.016 \\ {\footnotesize (0.001)} \end{tabular} & \begin{tabular}{@{}c@{}} 0.78 \end{tabular} \\
\midrule
SNPE & \begin{tabular}{@{}c@{}} 0.050 \\ {\footnotesize (0.032)} \end{tabular} & \begin{tabular}{@{}c@{}} 0.249 \\ {\footnotesize (0.016)} \end{tabular} & \begin{tabular}{@{}c@{}} 0.96 \end{tabular} & \begin{tabular}{@{}c@{}} 0.077 \\ {\footnotesize (0.062)} \end{tabular} & \begin{tabular}{@{}c@{}} 0.394 \\ {\footnotesize (0.034)} \end{tabular} & \begin{tabular}{@{}c@{}} 0.97 \end{tabular} & \begin{tabular}{@{}c@{}} \textbf{0.362} \\ {\footnotesize (0.311)} \end{tabular} & \begin{tabular}{@{}c@{}} \textbf{0.017} \\ {\footnotesize (0.001)} \end{tabular} & \begin{tabular}{@{}c@{}} 0.93 \end{tabular} \\
\bottomrule
\end{tabular}
}
\label{queuing_10exp}
\end{table}

\renewcommand{\arraystretch}{1.0}

\begin{figure}[!ht]
    \centering
    \includegraphics[width=0.99\textwidth]{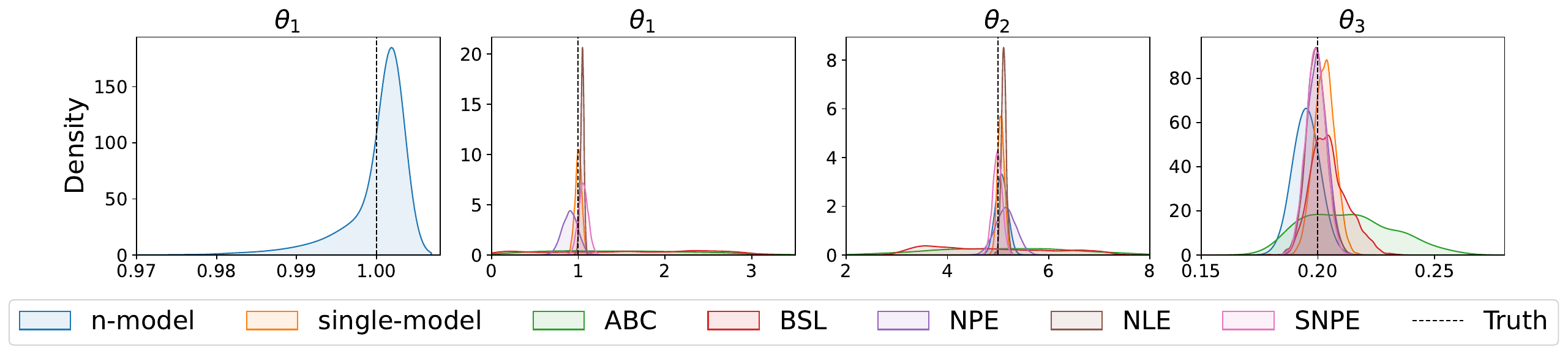}
    \caption{Posterior density plot of one experiment under the M/G/1-queuing model.}
    \label{fig:queuing}
\end{figure}

\vspace{-0.3in}
\subsection{Bayesian Monotonic Regression}\label{sec:mono_reg}
\vspace{-0.1in}
We consider the Bayesian monotonic regression with Bernstein polynomials proposed by \citet{mckay2011variable}. Since this model has a tractable likelihood, we compare all approximated posteriors against the {Gibbs sampling posteriors}. Additionally, as the true score is available, we evaluate the accuracy of estimated score under different implementation and a comprehensive comparison is provided in \Cref{sec:mono_reg_comp}. In \Cref{fig:score_direction} we have shown that the localization step is critical for learning the right Langevin direction in this high-dimensional example.

Following \citet{mckay2011variable}, we consider  i.i.d. observations $\{(x_i, y_i): i = 1,...,n\}$ generated by the following process $y_i = \text{tanh}(4x_i-2) + \varepsilon_i, \text{ with }  x_i \iid \text{U}(0, 1),\,  \varepsilon_i \iid \mN(0,0.1^2)$ for every $i=1, \ldots, n$. We set $n=1\,000$ and approximate the true function by Bernstein polynomials of order $M=10$, which leads to 11 parameters $\beta=(\beta_0, \ldots, \beta_M)^T$. The prior is set to be uniform on $[-5,5]\times[0,1]^M$ and the resulting posterior is truncated normal. Details on the polynomials and the true posterior are provided in \Cref{sec:mono_reg_details}.

 Since we are approximating a function in the example, for each experiment, we evaluate $y$ at $x \in \{0.00, 0.01, \ldots, 1.00\}$ ($101$ points), and obtain posterior predictive distribution of {$p(y\mid x)$} using the approximated posterior draws of $\theta$. We exclude the estimation bias and instead compute the Kolmogorov–Smirnov (KS) distance \citep{massey1951kolmogorov} and the 1-Wasserstein (W1) distance between the conditional distribution from the approximated posteriors and {the posterior from  Gibbs sampling}.  For each test, the final statistics is averaged over all $101$ $x$ values.

\begin{table}[!ht]
    \centering
\caption{Averaged results over 10 experiments in the monotonic regression example.}
\begin{tabular}{l|cccc}
\toprule
 & KS & W1 ($\times 10^{-2}$) & Cover95 & CI95 Width \\
\midrule
single-model & \textbf{0.095} & \textbf{0.210} & 0.976 & \textbf{0.034} \\
n-model & 0.159 & 0.395 & 0.985 & 0.042 \\
n-model-5x & 0.118 & 0.269 & 0.981 & 0.038 \\
ABC-W1 & 0.442 & 1.884 & 0.999 & 0.097 \\
BSL & 0.516 & 2.956 & 0.944 & 0.148 \\
NPE & 0.156 & 0.347 & 0.971 & 0.037 \\
NLE & 0.329 & 0.702 & 0.827 & 0.032 \\
\midrule
True posterior & - & - & 0.965 & 0.035 \\
\bottomrule
\end{tabular}
    \label{monoBP_res_10exp}
\end{table}

The averaged results in $10$ experiments are shown in Table \ref{monoBP_res_10exp}.
{We also present the posterior predictive $95\%$ credible band of $f(x; \theta) - \E[y \mid x]$ for one experiment in \Cref{sampling_res_error}.}
It can be seen that our methods significantly outperform the other methods in terms of closeness to the true posterior, and also achieve desirable coverage rates and tighter credible interval. Note that we have another version n-model-5x in \Cref{sampling_res_error}, which is the same algorithm with n-model but 5 times bigger the reference table size. We observe that increasing the diversity of $\theta^{(i)}$ helps improving the performance of n-model. However, the single-model still outperforms thanks to the debiasing strategy and the rich collection of $\theta^{(i)}$ in its training.

\begin{figure}[!ht]
    \centering
    \includegraphics[width=0.9\textwidth]{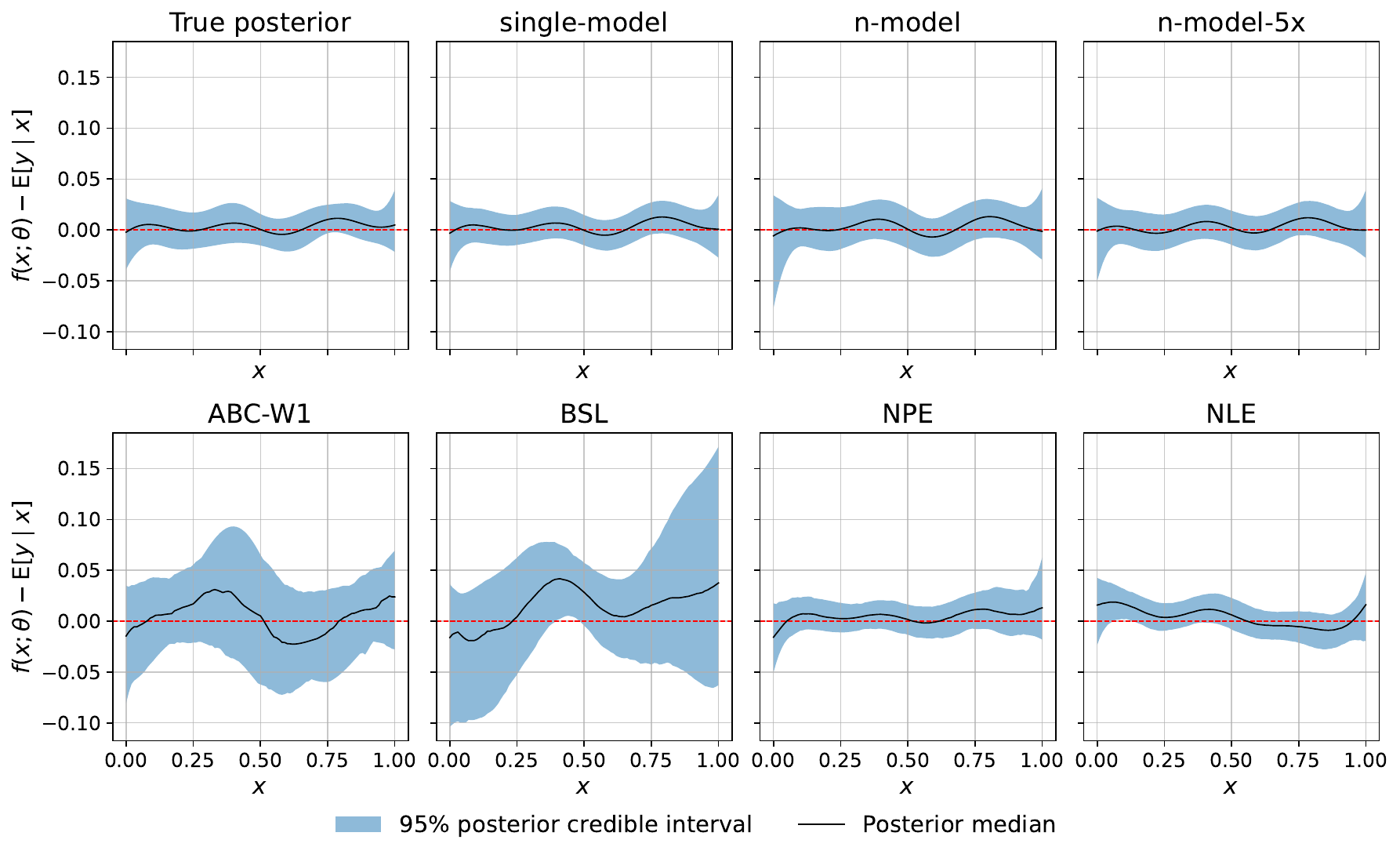}
    \caption{{95\% predictive bands of $f(x; \theta) - \E[y \mid x]$ by different methods from one monotonic regression experiment.}}
    \label{sampling_res_error}
\end{figure}

\subsection{mRNA transfection model}\label{sec:mRNA}
{

We consider the mRNA transfection model of \citet{haggstrom2026simulation}, a stochastic hierarchical model for translation kinetics after mRNA transfection. For each cell, latent mRNA and green fluorescent protein (GFP) trajectories evolve according to a two-dimensional SDE with cell-specific degradation, translation, and transfection-time parameters. The observed data are noisy log-transformations of the GFP trajectory, recorded at $180$ time points for each of $n$ i.i.d. cells. The full data-generating process, observation model, discretization scheme, and parameterization are given in \Cref{sec:mrna_details}. The model has 12 parameters, denoted by
$\theta = (\log m_0, \log \mathrm{scale}, \log \mathrm{offset}, \log \sigma, \mu_\delta, \mu_\gamma, \mu_k, \mu_{t_0}, \log \tau_\delta, \log \tau_\gamma, \log \tau_k, \log \tau_{t_0})$.

\textbf{Simulation study.}
We consider the same model as in \cite{haggstrom2026simulation}, with the true parameter
$\theta^\ast = (5.7, 0.7, 2.08, -1.6, -0.694, -3, 0.027, 0, -1.15, -1.15, -1.15, -1.15)$ and sample size $n = 200$. As discussed in \cite{pieschner2022identifiability}, the three  parameters $(\log \mathrm{scale}, \log m_0, \mu_k)$ are not individually identifiable and only their sum is identifiable. Therefore, we  report the posterior distribution only for their sum. Empirically, we also find that the marginal Fisher information matrix associated with these three parameters, evaluated from our estimated score function, is nearly rank one, with eigenvalues $(0.01, 0.09, 25)$. We  exclude BSL  because the Metropolis-Hastings sampler mixes very slowly in this example; further details are provided in \Cref{sec:SDE_BSL}.

We compare all methods across $100$ independently generated datasets and summarize the results in \Cref{tab:SDEMEM_simu_table}.  The posterior density plot from one representative dataset are shown in \Cref{fig:kde_SDEMEM_simu}. Overall, our method, NPE, and NLE outperform ABC, likely thanks to the flexibility of neural approximation, although NLE exhibits some undercoverage. Among all methods, our method achieves the lowest estimation errors and the narrowest credible intervals for most parameters.

\renewcommand{\arraystretch}{0.9}

\begin{table}[!ht]
    \centering
\caption{Averaged results in $100$ experiments, with standard deviations in the parentheses}
\resizebox{\textwidth}{!}{
\begin{tabular}{c|cccc|cccc|cccc}
\toprule
 & \multicolumn{4}{c|}{$\lvert \wh\theta - \theta^\ast \rvert$} & \multicolumn{4}{c|}{CI95 Width} & \multicolumn{4}{c}{Cover95} \\
 & Ours & ABC & NLE & NPE & Ours & ABC & NLE & NPE & Ours & ABC & NLE & NPE \\
\midrule
\begin{tabular}{@{}c@{}} $\log \mathrm{offset}$ \end{tabular} & \begin{tabular}{@{}c@{}} \textbf{0.004} \\ {\footnotesize (0.003)} \end{tabular} & \begin{tabular}{@{}c@{}} 0.008 \\ {\footnotesize (0.005)} \end{tabular} & \begin{tabular}{@{}c@{}} 0.006 \\ {\footnotesize (0.004)} \end{tabular} & \begin{tabular}{@{}c@{}} 0.005 \\ {\footnotesize (0.003)} \end{tabular} & \begin{tabular}{@{}c@{}} \textbf{0.024} \\ {\footnotesize (0.001)} \end{tabular} & \begin{tabular}{@{}c@{}} 0.291 \\ {\footnotesize (0.037)} \end{tabular} & \begin{tabular}{@{}c@{}} 0.023 \\ {\footnotesize (0.001)} \end{tabular} & \begin{tabular}{@{}c@{}} 0.025 \\ {\footnotesize (0.001)} \end{tabular} & \begin{tabular}{@{}c@{}} 0.98 \end{tabular} & \begin{tabular}{@{}c@{}} 1.00 \end{tabular} & \begin{tabular}{@{}c@{}} 0.86 \end{tabular} & \begin{tabular}{@{}c@{}} 0.98 \end{tabular} \\
\midrule
\begin{tabular}{@{}c@{}} $\log \sigma$ \end{tabular} & \begin{tabular}{@{}c@{}} \textbf{0.004} \\ {\footnotesize (0.003)} \end{tabular} & \begin{tabular}{@{}c@{}} 0.289 \\ {\footnotesize (0.034)} \end{tabular} & \begin{tabular}{@{}c@{}} 0.005 \\ {\footnotesize (0.004)} \end{tabular} & \begin{tabular}{@{}c@{}} 0.011 \\ {\footnotesize (0.009)} \end{tabular} & \begin{tabular}{@{}c@{}} \textbf{0.020} \\ {\footnotesize (0.001)} \end{tabular} & \begin{tabular}{@{}c@{}} 0.583 \\ {\footnotesize (0.051)} \end{tabular} & \begin{tabular}{@{}c@{}} 0.016 \\ {\footnotesize (0.001)} \end{tabular} & \begin{tabular}{@{}c@{}} 0.058 \\ {\footnotesize (0.019)} \end{tabular} & \begin{tabular}{@{}c@{}} 0.97 \end{tabular} & \begin{tabular}{@{}c@{}} 0.49 \end{tabular} & \begin{tabular}{@{}c@{}} 0.80 \end{tabular} & \begin{tabular}{@{}c@{}} 0.95 \end{tabular} \\
\midrule
\begin{tabular}{@{}c@{}} $\mu_{\delta}$ \end{tabular} & \begin{tabular}{@{}c@{}} 0.017 \\ {\footnotesize (0.016)} \end{tabular} & \begin{tabular}{@{}c@{}} \textbf{0.013} \\ {\footnotesize (0.011)} \end{tabular} & \begin{tabular}{@{}c@{}} 0.026 \\ {\footnotesize (0.019)} \end{tabular} & \begin{tabular}{@{}c@{}} 0.018 \\ {\footnotesize (0.016)} \end{tabular} & \begin{tabular}{@{}c@{}} \textbf{0.101} \\ {\footnotesize (0.006)} \end{tabular} & \begin{tabular}{@{}c@{}} 0.482 \\ {\footnotesize (0.040)} \end{tabular} & \begin{tabular}{@{}c@{}} 0.098 \\ {\footnotesize (0.009)} \end{tabular} & \begin{tabular}{@{}c@{}} \textbf{0.101} \\ {\footnotesize (0.007)} \end{tabular} & \begin{tabular}{@{}c@{}} 0.95 \end{tabular} & \begin{tabular}{@{}c@{}} 1.00 \end{tabular} & \begin{tabular}{@{}c@{}} 0.86 \end{tabular} & \begin{tabular}{@{}c@{}} 0.96 \end{tabular} \\
\midrule
\begin{tabular}{@{}c@{}} $\mu_{\gamma}$ \end{tabular} & \begin{tabular}{@{}c@{}} 0.019 \\ {\footnotesize (0.015)} \end{tabular} & \begin{tabular}{@{}c@{}} \textbf{0.018} \\ {\footnotesize (0.013)} \end{tabular} & \begin{tabular}{@{}c@{}} 0.022 \\ {\footnotesize (0.018)} \end{tabular} & \begin{tabular}{@{}c@{}} 0.019 \\ {\footnotesize (0.015)} \end{tabular} & \begin{tabular}{@{}c@{}} \textbf{0.092} \\ {\footnotesize (0.006)} \end{tabular} & \begin{tabular}{@{}c@{}} 0.360 \\ {\footnotesize (0.039)} \end{tabular} & \begin{tabular}{@{}c@{}} 0.088 \\ {\footnotesize (0.006)} \end{tabular} & \begin{tabular}{@{}c@{}} \textbf{0.092} \\ {\footnotesize (0.006)} \end{tabular} & \begin{tabular}{@{}c@{}} 0.92 \end{tabular} & \begin{tabular}{@{}c@{}} 1.00 \end{tabular} & \begin{tabular}{@{}c@{}} 0.88 \end{tabular} & \begin{tabular}{@{}c@{}} 0.93 \end{tabular} \\
\midrule
\begin{tabular}{@{}c@{}} $\mu_{t_0}$ \end{tabular} & \begin{tabular}{@{}c@{}} 0.019 \\ {\footnotesize (0.015)} \end{tabular} & \begin{tabular}{@{}c@{}} \textbf{0.013} \\ {\footnotesize (0.008)} \end{tabular} & \begin{tabular}{@{}c@{}} 0.026 \\ {\footnotesize (0.021)} \end{tabular} & \begin{tabular}{@{}c@{}} 0.019 \\ {\footnotesize (0.015)} \end{tabular} & \begin{tabular}{@{}c@{}} \textbf{0.089} \\ {\footnotesize (0.005)} \end{tabular} & \begin{tabular}{@{}c@{}} 0.220 \\ {\footnotesize (0.019)} \end{tabular} & \begin{tabular}{@{}c@{}} 0.110 \\ {\footnotesize (0.028)} \end{tabular} & \begin{tabular}{@{}c@{}} \textbf{0.089} \\ {\footnotesize (0.005)} \end{tabular} & \begin{tabular}{@{}c@{}} 0.94 \end{tabular} & \begin{tabular}{@{}c@{}} 1.00 \end{tabular} & \begin{tabular}{@{}c@{}} 0.89 \end{tabular} & \begin{tabular}{@{}c@{}} 0.93 \end{tabular} \\
\midrule
\begin{tabular}{@{}c@{}} $\log \tau_{\delta}$ \end{tabular} & \begin{tabular}{@{}c@{}} 0.057 \\ {\footnotesize (0.039)} \end{tabular} & \begin{tabular}{@{}c@{}} 0.073 \\ {\footnotesize (0.043)} \end{tabular} & \begin{tabular}{@{}c@{}} 0.066 \\ {\footnotesize (0.054)} \end{tabular} & \begin{tabular}{@{}c@{}} \textbf{0.050} \\ {\footnotesize (0.036)} \end{tabular} & \begin{tabular}{@{}c@{}} \textbf{0.265} \\ {\footnotesize (0.037)} \end{tabular} & \begin{tabular}{@{}c@{}} 1.404 \\ {\footnotesize (0.082)} \end{tabular} & \begin{tabular}{@{}c@{}} 0.314 \\ {\footnotesize (0.035)} \end{tabular} & \begin{tabular}{@{}c@{}} 0.270 \\ {\footnotesize (0.013)} \end{tabular} & \begin{tabular}{@{}c@{}} 0.96 \end{tabular} & \begin{tabular}{@{}c@{}} 1.00 \end{tabular} & \begin{tabular}{@{}c@{}} 0.93 \end{tabular} & \begin{tabular}{@{}c@{}} 0.98 \end{tabular} \\
\midrule
\begin{tabular}{@{}c@{}} $\log \tau_{\gamma}$ \end{tabular} & \begin{tabular}{@{}c@{}} \textbf{0.039} \\ {\footnotesize (0.031)} \end{tabular} & \begin{tabular}{@{}c@{}} 0.071 \\ {\footnotesize (0.039)} \end{tabular} & \begin{tabular}{@{}c@{}} 0.058 \\ {\footnotesize (0.039)} \end{tabular} & \begin{tabular}{@{}c@{}} 0.042 \\ {\footnotesize (0.032)} \end{tabular} & \begin{tabular}{@{}c@{}} \textbf{0.220} \\ {\footnotesize (0.005)} \end{tabular} & \begin{tabular}{@{}c@{}} 1.350 \\ {\footnotesize (0.096)} \end{tabular} & \begin{tabular}{@{}c@{}} 0.241 \\ {\footnotesize (0.016)} \end{tabular} & \begin{tabular}{@{}c@{}} 0.224 \\ {\footnotesize (0.007)} \end{tabular} & \begin{tabular}{@{}c@{}} 0.97 \end{tabular} & \begin{tabular}{@{}c@{}} 1.00 \end{tabular} & \begin{tabular}{@{}c@{}} 0.91 \end{tabular} & \begin{tabular}{@{}c@{}} 0.97 \end{tabular} \\
\midrule
\begin{tabular}{@{}c@{}} $\log \tau_{k}$ \end{tabular} & \begin{tabular}{@{}c@{}} \textbf{0.039} \\ {\footnotesize (0.032)} \end{tabular} & \begin{tabular}{@{}c@{}} 0.069 \\ {\footnotesize (0.041)} \end{tabular} & \begin{tabular}{@{}c@{}} 0.059 \\ {\footnotesize (0.045)} \end{tabular} & \begin{tabular}{@{}c@{}} 0.046 \\ {\footnotesize (0.036)} \end{tabular} & \begin{tabular}{@{}c@{}} \textbf{0.223} \\ {\footnotesize (0.005)} \end{tabular} & \begin{tabular}{@{}c@{}} 1.306 \\ {\footnotesize (0.088)} \end{tabular} & \begin{tabular}{@{}c@{}} 0.263 \\ {\footnotesize (0.026)} \end{tabular} & \begin{tabular}{@{}c@{}} 0.256 \\ {\footnotesize (0.012)} \end{tabular} & \begin{tabular}{@{}c@{}} 0.96 \end{tabular} & \begin{tabular}{@{}c@{}} 1.00 \end{tabular} & \begin{tabular}{@{}c@{}} 0.92 \end{tabular} & \begin{tabular}{@{}c@{}} 0.96 \end{tabular} \\
\midrule
\begin{tabular}{@{}c@{}} $\log \tau_{t_0}$ \end{tabular} & \begin{tabular}{@{}c@{}} 0.041 \\ {\footnotesize (0.028)} \end{tabular} & \begin{tabular}{@{}c@{}} \textbf{0.040} \\ {\footnotesize (0.028)} \end{tabular} & \begin{tabular}{@{}c@{}} 0.062 \\ {\footnotesize (0.043)} \end{tabular} & \begin{tabular}{@{}c@{}} 0.041 \\ {\footnotesize (0.027)} \end{tabular} & \begin{tabular}{@{}c@{}} \textbf{0.200} \\ {\footnotesize (0.003)} \end{tabular} & \begin{tabular}{@{}c@{}} 0.692 \\ {\footnotesize (0.061)} \end{tabular} & \begin{tabular}{@{}c@{}} 0.251 \\ {\footnotesize (0.074)} \end{tabular} & \begin{tabular}{@{}c@{}} 0.201 \\ {\footnotesize (0.006)} \end{tabular} & \begin{tabular}{@{}c@{}} 0.96 \end{tabular} & \begin{tabular}{@{}c@{}} 1.00 \end{tabular} & \begin{tabular}{@{}c@{}} 0.91 \end{tabular} & \begin{tabular}{@{}c@{}} 0.97 \end{tabular} \\
\midrule
\begin{tabular}{@{}c@{}} $\log m_0 + \mu_k$ \\ $+\,\log \mathrm{scale}$ \end{tabular} & \begin{tabular}{@{}c@{}} \textbf{0.019} \\ {\footnotesize (0.016)} \end{tabular} & \begin{tabular}{@{}c@{}} 0.031 \\ {\footnotesize (0.024)} \end{tabular} & \begin{tabular}{@{}c@{}} 0.026 \\ {\footnotesize (0.018)} \end{tabular} & \begin{tabular}{@{}c@{}} 0.021 \\ {\footnotesize (0.016)} \end{tabular} & \begin{tabular}{@{}c@{}} 0.116 \\ {\footnotesize (0.003)} \end{tabular} & \begin{tabular}{@{}c@{}} 0.593 \\ {\footnotesize (0.042)} \end{tabular} & \begin{tabular}{@{}c@{}} 0.094 \\ {\footnotesize (0.006)} \end{tabular} & \begin{tabular}{@{}c@{}} \textbf{0.104} \\ {\footnotesize (0.006)} \end{tabular} & \begin{tabular}{@{}c@{}} 0.98 \end{tabular} & \begin{tabular}{@{}c@{}} 1.00 \end{tabular} & \begin{tabular}{@{}c@{}} 0.85 \end{tabular} & \begin{tabular}{@{}c@{}} 0.93 \end{tabular} \\
\bottomrule
\end{tabular}
}
\label{tab:SDEMEM_simu_table}
\end{table}

\renewcommand{\arraystretch}{1.0}

\begin{figure}[!ht]
    \centering
    \includegraphics[width=0.99\textwidth]{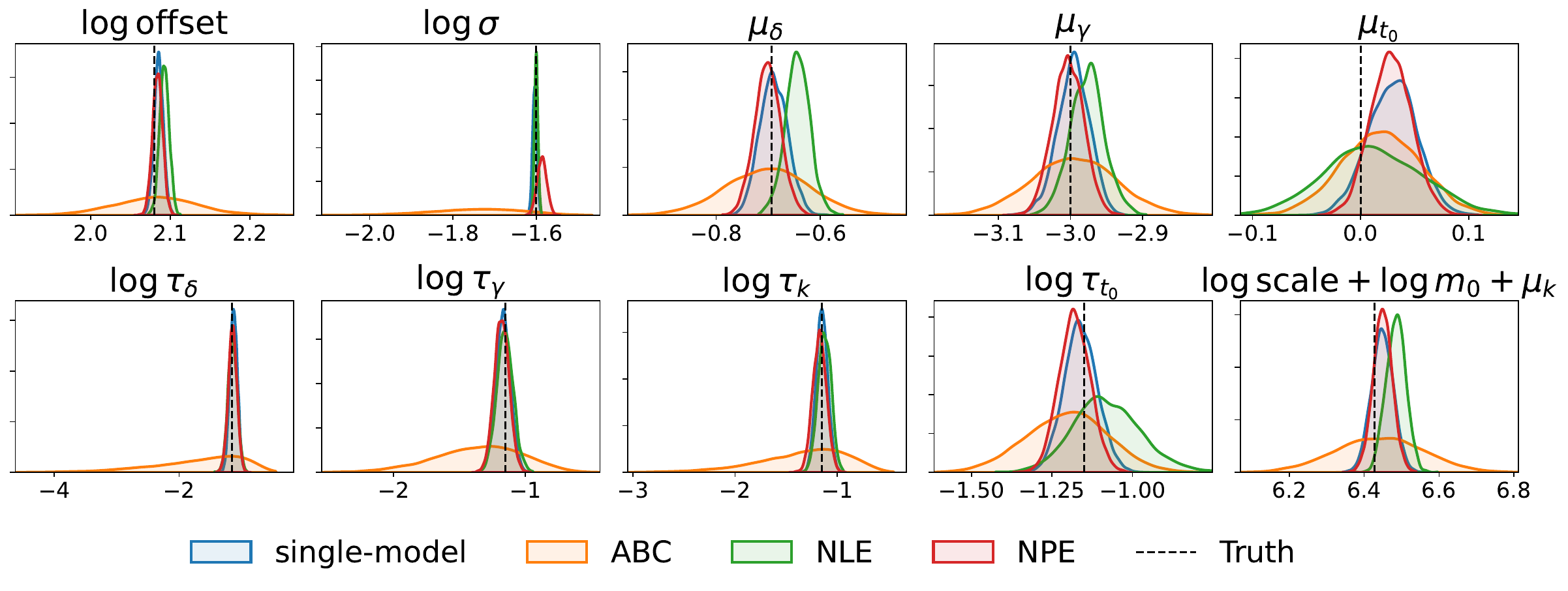}
    \caption{Posterior densities from one simulation of the mRNA transfection model.}
    \label{fig:kde_SDEMEM_simu}
\end{figure}

\textbf{Real data study.}
We now compare different methods using a real-world dataset from \citet{frohlich2018multi}, and follow exactly the same setting as \cite{haggstrom2026simulation}.  We exclude NLE from the comparison because it performs poorly, likely due to model misspecification and the accumulation of approximation errors when combining the estimated likelihood contributions across all observations; additional results are provided in \Cref{sec:more_res_SDErealdata}.

\Cref{fig:kde_SDEMEM_realdata} presents the posterior densities obtained by our method, ABC, and NPE. As expected, our method and NPE yield more concentrated posteriors than ABC, with similar posterior centers for most parameters. However, there are several noticeable differences between our method and NPE. First, our method produces larger posterior estimates of $\log \sigma$, which leads to improved posterior predictive distributions of the GFP trajectories during the initial time period ($t \in [0,2]$), without substantially affecting the fit over later time periods. Second, the posterior distributions of $\log \tau_\delta$ and $\log \tau_\gamma$ obtained by our method and NPE have different centers, with ours closer to those from ABC. Third, the joint posterior distribution of $(\mu_\delta,\mu_\gamma)$ produced by our method is bimodal, whereas NPE identifies only one mode. In fact, \citet{pieschner2022identifiability} also show that, in certain regimes, $(\mu_\delta,\mu_\gamma)$ are identifiable only up to exchangeability, which is consistent with the bimodal posterior behavior observed here. However, our posterior is not exactly symmetric between the two modes, possibly due to model misspecification. We further verify in \Cref{sec:more_res_SDErealdata} that both modes can generate data that fit the observed trajectory equally well.
Overall, these results suggest that our method provides a more faithful characterization of the posterior uncertainty than NPE while retaining the sharper posterior concentration relative to ABC.

\begin{figure}[!ht]
    \centering
    \includegraphics[width=0.99\textwidth]{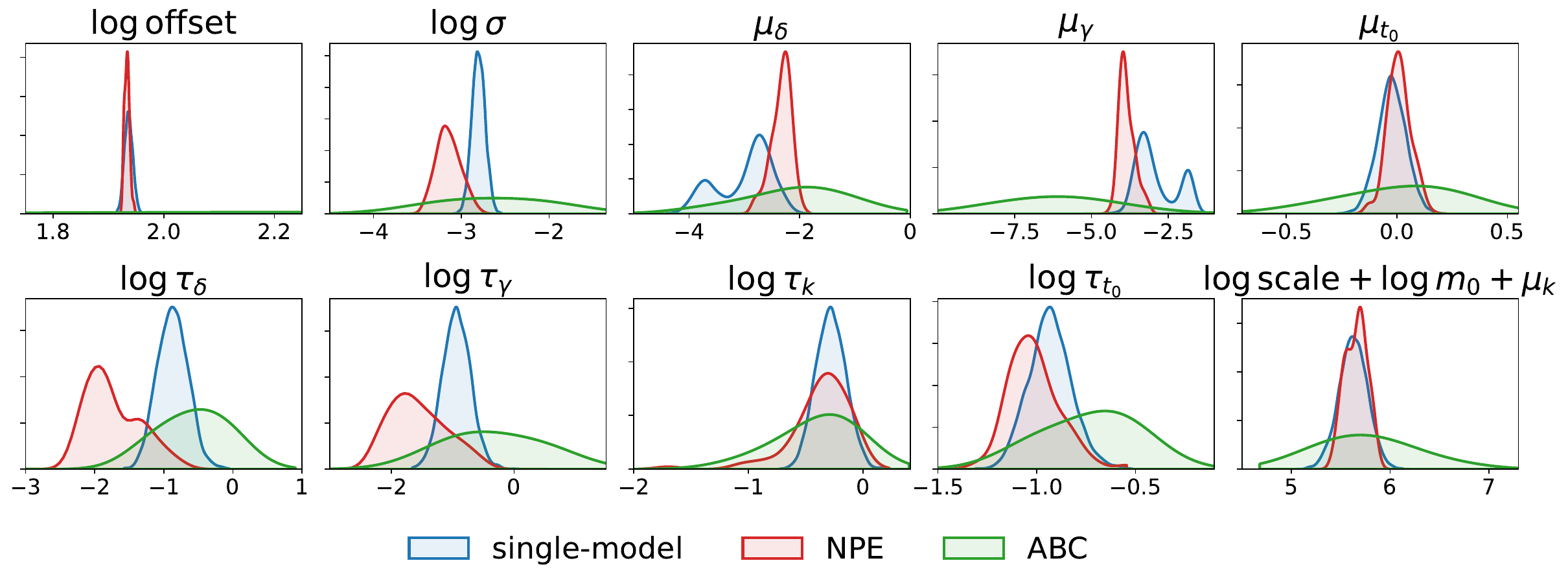}
    \caption{Posterior densities in the real data application.}
    \label{fig:kde_SDEMEM_realdata}
\end{figure}

\begin{figure}[!ht]
    \centering
    \includegraphics[width=0.99\textwidth]{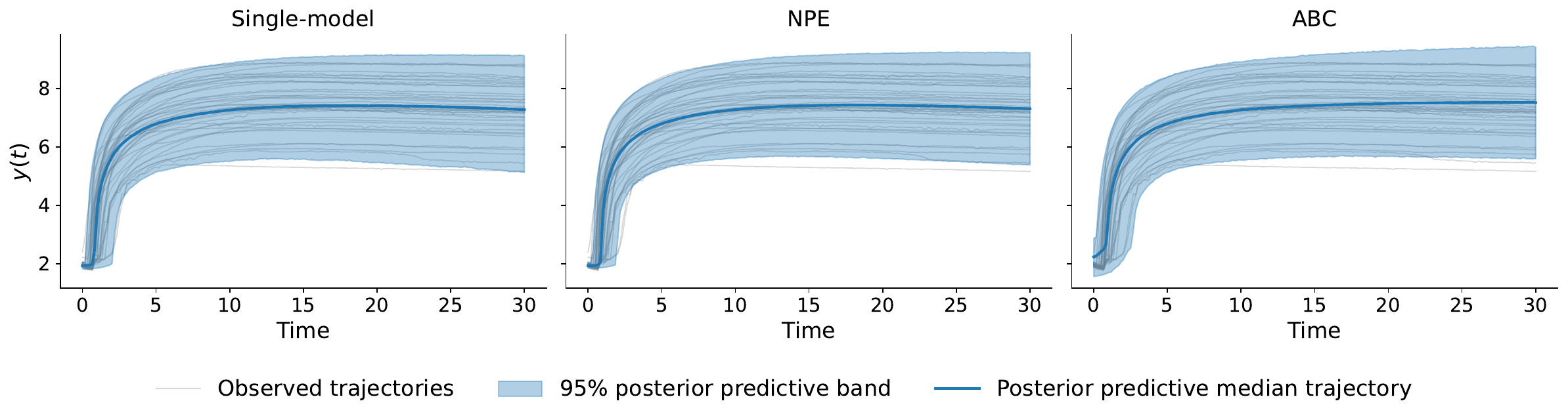}
    \caption{Posterior predictive distributions from different methods in the real-data application.
}
    \label{fig:SDE_real_postpred}
\end{figure}

}

\vspace{-0.3in}
\section{Discusssion}\label{sec:discussion}
\vspace{-0.1in}
Our idea of enforcing statistical structures on score-matching networks
opens several promising avenues for future work. First, although our primary focus has been on Langevin dynamics
and unimodal posteriors, the approach naturally extends to other
gradient-based samplers that leverage score and Hessian information.
Examples include Hamiltonian Monte Carlo \citep{neal2011mcmc},
preconditioned Langevin dynamics \citep{titsias2023optimal},
and Riemann manifold Langevin dynamics \citep{girolami2011riemann}.
Incorporating our regularized score estimators into these samplers
has the potential to further accelerate exploration of complex parameter spaces
and to better exploit inherent low-dimensional structures.

Second, the framework can be extended beyond Langevin-type methods
to other generative models, such as diffusion models.
Diffusion models \citep{songs2021score,song2021maximum}
are fundamentally tied to score matching
and have recently demonstrated superior approximation performance.
Enforcing statistical structures within the diffusion process
may improve both theoretical efficiency and empirical performance,
and we view this as an important future direction.

Lastly, accurate estimation of scores and Hessians provides
a foundation for posterior calibration under model misspecification,
which is nearly unavoidable in real-world applications.
For example, \citet{frazier2025synthetic} proposed calibrating
BSL posteriors using approximate score and Hessian information
derived under Gaussian assumptions on summary statistics.
Our method is expected to improve upon this approach,
since it learns gradient information directly from simulated datasets
rather than relying on an ad hoc Gaussian approximation.

\bibliographystyle{apalike}

{\footnotesize
\setstretch{0.8}
\bibliography{langevin}
}

\clearpage
\appendix
\setstretch{1.25}

\addcontentsline{toc}{section}{Appendix}\part{Appendix}\parttoc
{The appendix is organized as follows. Appendix~\ref{sec:related_work} provides an expanded review of related work on simulation-based inference and discusses approaches for handling large parameter dimensions and sample sizes. Appendix~\ref{sec:implementation} presents further methodological details, including the localization step, boundary conditions, full-data score matching, curvature matching, and alternative implementations of debiased score matching. Appendix~\ref{sec:additional_theory} contains proofs of the main results and additional theoretical analysis of the localization step. Finally, \Cref{sec:additional_implementation} provides implementation details for all examples in \Cref{sec:simulation} and additional numerical results for the stochastic epidemic examples.
}

{
\section{Review of Related Work}\label{sec:related_work}

This section provides a broader discussion of related work, especially the ones most relevant to our method.

\paragraph*{ABC and BSL.} In ABC methods, one simulates $N$ pairs of parameters and datasets
$\{ (\theta^{(k)}, \Xn^{(k)}) \}_{k=1}^N$ from the joint distribution
$p(\theta, \Xn) = \pi(\theta)\,\pn_\theta(\Xn)$, which we refer to as the
ABC \emph{reference table}. The parameter draws $\{\theta^{(k)}\}$ are then
weighted according to the similarity between the simulated dataset
$\Xn^{(k)}$ and the observed dataset $\Xn^\ast$. Much of the ABC literature
focuses on defining effective similarity measures, as these directly
determine the quality of the approximate posterior. Common strategies
include: (1) computing distances between summary statistics, chosen
either through expert knowledge or automated procedures
\citep{fearnhead2011constructing}; and (2) using discrepancy metrics
on empirical distributions, such as the
Kullback-Leibler (KL) divergence
\citep{jiang2018approximate,wang2022approximate}, the Wasserstein
distance \citep{bernton2019approximate}, or the
Maximum Mean Discrepancy (MMD) \citep{park2016k2,niu2023discrepancy,bharti2023optimally}.

BSL \citep{price2018bayesian} takes a parametric approach by assuming that the summary statistics follow a Gaussian distribution and approximating the likelihood accordingly, with the Gaussian mean and covariance estimated from simulated datasets.  Robust and generalized variants have also been developed, for example by \citep{frazier2021robust,frazier2025synthetic}.
This avoids the need to specify a kernel function or impose ad hoc thresholding, but the accuracy of the method depends critically on the validity of the Gaussian assumption for the summary statistics.

\paragraph*{Neural SBI methods.} Thanks to the advances in deep neural networks and generative models, a new class of SBI methods has emerged that replace discrepancy-based approaches with flexible neural approximations to posterior-related quantities learn from simulations. One prominent direction is neural posterior estimation (NPE), which directly learns a conditional density estimator for the posterior distribution. Early work by \citep{papamakarios2016fast} uses mixture density networks, followed by more normalizing flow-based approaches. The second direction is neural likelihood-free inference (NLE) \citep{papamakarios2019sequential}, which learns a conditional density estimator for the likelihood function, and then combines it with the prior for posterior inference. A third direction is neural ratio estimation (NRE), which learns the likelihood-to-evidence ration by contrastive learning \citep{hermans2020likelihood, durkan2020contrastive,miller2022contrastive}. Other approaches including adversarial generative networks \citep{wang2022adversarial}. These neural SBI methods have shown capable of producing good results with much fewer simulations than ABC methods, see for example \citet{lueckmann2021benchmarking,frazier2024statistical}.

\paragraph*{Score-based and diffusion-based SBI.} We would first review the methods utilize score estimation without a diffusion model. \citet{pacchiardi2022score} use score matching to fit neural exponential families for likelihood-free inference. \citet{zeghal2022neural} study neural posterior estimation with differentiable simulators and score-based constraints. \citet{glaser2022maximum} propose energy-based likelihood learning for simulation-based inference. More recent direct score-learning approaches, such as \citet{khoo2025direct}, further illustrate the potential to get a maximum likelihood estimator in likelihood-free settings.

A closely related and rapidly growing literature combines SBI with diffusion-based or conditional score-based generative models. Examples include \citet{simons2023neural} which study neural score estimation with conditional diffusion models, \citet{sharrock2024sequential} which develop sequential diffusion-based SBI, \citet{gloeckler2024all} which propose an all-in-one diffusion framework for SBI, and Nautiyal, Hellander and Singh (2026) which develop conditional diffusion models for simulation-based inference.

\paragraph*{Solutions to large $d_\theta$.}
A recurring challenge in simulation-based inference is scalability with respect to the parameter dimension $d_\theta$. A large literature addresses this problem by adapting the proposal distribution over the course of the algorithm.

Within ABC, sequential and adaptive solutions include SMC-ABC \citep{sisson2007sequential,beaumont2009adaptive,del2012adaptive}, with more recent developments such as guided sequential ABC in \citep{picchini2025guided}, and Thompson sampling \citep{ohagan2024tree}. There is also a literature on high-dimensional ABC refinements, including the Bayes-linear approach \citep{nott2014approximate} and the Gaussian-copula construction \citep{li2017extending}. These methods aim to improve efficiency by concentrating computation in regions of the parameter space that are more likely to be posterior-relevant.

Adaptive ideas also appear prominently in neural SBI. While neural SBI methods have been shown to require much fewer simulations than traditional ABC methods \citep{frazier2024statistical} they still face similar challenges when $d_\theta$ is large. Sequential neural likelihood estimation \citep{papamakarios2019sequential,wang2022adversarial} updates the proposal distribution across rounds to focus training on regions supported by the observation. \citep{wang2024preconditioned} propose a precoditioning step using cheap ABC approximations to improve the efficiency of neural posterior estimation. More broadly, many practical SBI workflows use some form of proposal adaptation, pilot estimation, or sequential refinement in order to improve performance when $d_\theta$ is moderate or large.

\paragraph*{Solutions to large $n$.} In this setting, the cost of simulating full datasets may be substantial, and the complexity of a neural approximation can grow unfavorably with the sample size n if the entire dataset is treated as a single high-dimensional input. A broad class of solutions addresses this problem by compressing or factorizing the information in the sample.

One common approach is to use summaries or learned representations of the sample. In neural SBI, permutation-invariant encoders such as Deep Sets, introduced by \citet{zaheer2017deep}, and related embedding-network constructions discussed by \citet{deistler2025simulation}, are widely used to process i.i.d. data. These architectures can greatly improve computational efficiency and have become standard practical tools. More recent work, such as \citet{luciano2025permutations}, also studies how exploiting permutation structure can improve simulation-based procedures.

A different strategy is to exploit factorization of the likelihood or posterior-related quantity across observations. In neural likelihood or ratio estimation, one may learn the contribution of a single observation and then aggregate across the sample, thereby avoiding the need to process the entire dataset jointly \citep{papamakarios2019sequential,durkan2020contrastive}. Closely related ideas appear in recent score-based SBI methods for tall-data settings, most notably \citet{geffner2023compositional} and \citet{linhart2024diffusion}, where smaller-scale score contributions are composed to perform inference on large datasets.

}

\section{Method Details}\label{sec:implementation}

\subsection{Localization step} \label{sec:preconditioning_details}

We first want to demonstrate poor performance of score-matching networks in low-density region on a simple example.  In  \Cref{fig:low_density}, we consider a simple Binomial
example where $X \mid \theta \sim \text{Bin}(100, \theta)$ with a Beta prior
$\mathcal{B}(5,5)$ on $\theta$. We examine how the score-matching error at a
fixed $\theta^\ast$, defined as
\[
\text{Error}(\theta^\ast) \;=\;
\E_{(\theta, X) \sim \pi(\theta \mid X)\,p(X \mid \theta^\ast)}
\big[\big\| s_\phi(\theta, X) - s^\ast(\theta, X) \big\|^2\big],
\]
varies as a function of the prior density at $\theta^\ast$.

\begin{figure}[!ht]
    \centering
    \includegraphics[width=0.4\textwidth]{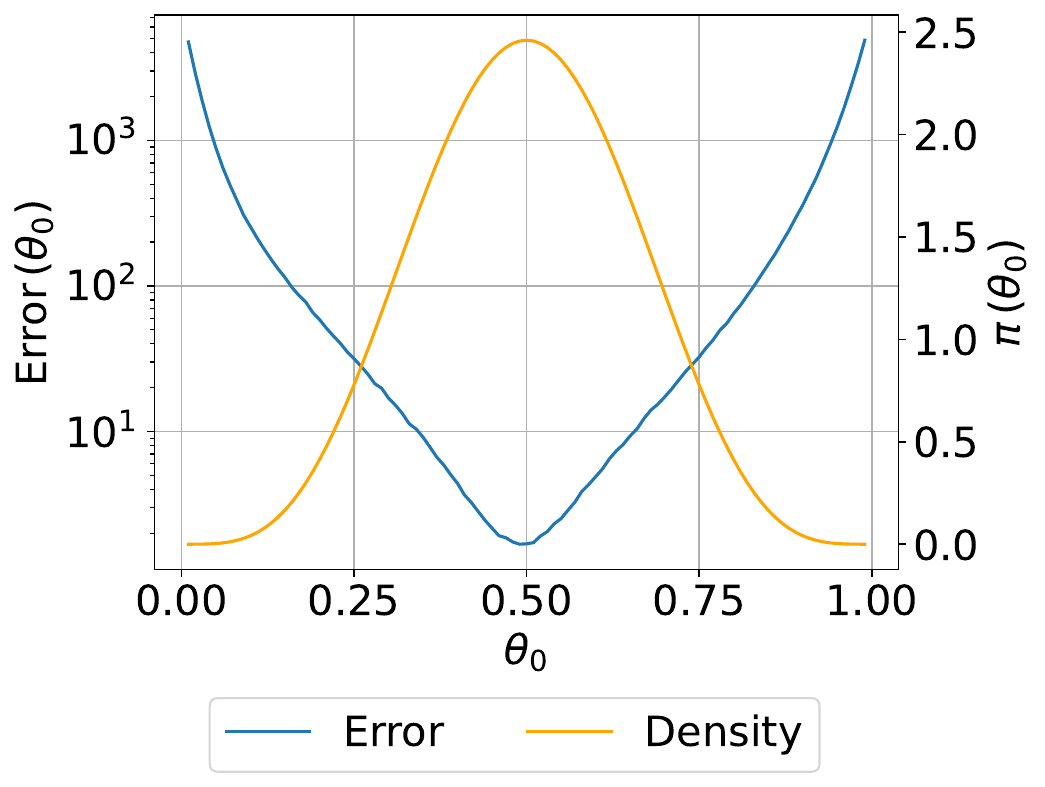}
    \caption{Estimation errors against prior density in the Beta-Binomial example.
    The estimation error increases significantly when the prior density is low.}\label{fig:low_density}
\end{figure}

As shown in \Cref{fig:low_density}, the estimation error of the score-matching
network increases substantially when the prior density at $\theta^\ast$ is low.
This occurs because the network is trained on simulated datasets, and when the
prior density near $\theta^\ast$ is small, few simulations fall in the vicinity
of the observed data $\Xn^\ast$. The resulting scarcity of informative training
examples leads to poor score estimation precisely where accuracy is most
critical.

Now we want to continue on implementation details of the localization step. Recall from \Cref{sec:preconditioning} that our localization method is to solve
\begin{equation}\label{eq:obj_local}
    \widehat \theta^{(b)} =\arg\min_\theta d_\text{SW}(\tau(\theta,\Zm^{(b)}), \Xn^{0}) \text{ for } \, b=1,\ldots, B,
\end{equation}
where
\begin{equation*}
d_{\text{SW}_p}(\mu, \nu) := \int_{S^{p-1}} d_{\text{W}_1}(\mu_\omega, \nu_\omega)\diff \sigma(\omega).
\end{equation*}
where $\omega \in S^{p-1} := \{\omega' \in \R^p : \|\omega'\| \leq 1\}$ is a  projection direction, $\sigma(\cdot)$ is the uniform measure on the unit sphere,
 and $\mu_{\omega}$ and $\nu_{\omega}$ denote the pushforward distributions of  $\mu$ and $\nu$ under the projection $x \mapsto \omega^T x$.

This localization method is applicable to all examples considered in this paper. We do not use it in the M/G/1-queuing model since the parameter dimension is low, and we do not use it in the Stochastic Epidemic Model (5-floor case) because the prior is already informative.

We apply the localization method in the monotonic regression model and the  Stochastic Epidemic Model (10-floor case). In both examples, we use $100$ random directions $\omega_k\in \mathcal{S}^{p-1}, k=1, \ldots, 100$ to approximate the $\text{SW}_1$ distance and obtain $B = 100$ samples. Besides, we set $m=n$ in each example so that the calculation of the $\text{W}_1$ distance between the two projected $1$-dimensional datasets reduces to a sorting problem and further decreases the computational cost. Finally, we solve the optimization problem \eqref{eq:obj_local} using Adam, with gradient calculated by PyTorch's Autograd module.

{
\subsubsection{Connection to adaptive proposal construction}\label{sec:localization_discussion}

Our localization step can be viewed as a form of adaptive proposal construction. Our method relies on the assumption that the data generating process admits a reparameterization that is differentiable in $\theta$, which is common in simulation-based methods such as pseudo-marginal and partical methods \citep{andrieu2009pseudo,nemeth2016particle},  indirect inference \citep{gourieroux1993indirect,mcfadden1989method,pakes1989simulation}, and much of the SBI literature \citep{meeds2015optimization,brehmer2020mining,cranmer2020frontier}. By utilizing the gradient information from the simulator, our localization method can move efficiently towards the relevant region.

This assumption in general is not restrictive. Potential violations include models with discontinuous operations, such as accept-reject steps or discrete branching logic, which can make $\tau$ non-differentiable. In such cases, several remedies are available. One can smooth or relax the
simulator (e.g., via continuous approximations or reparameterizations), or adopt techniques
developed in the particle MCMC literature to handle non-differentiability \citep{nemeth2016particle}.

Additionally, our framework does not critically rely on this assumption in the localization
stage. When differentiability is not available, one can instead use alternative, derivative-
free localization strategies—such as sequential neural methods \citep{papamakarios2019sequential,sharrock2024sequential}, summary-based ABC preconditioning \citep{wang2024preconditioned}, or surrogate-
based optimization—to construct a proposal distribution concentrated in high posterior density
regions. Since the goal of localization is only to provide a coarse approximation (rather than
an exact posterior), the subsequent score matching and Langevin refinement remain effective
as long as this proposal places sufficient mass near the true parameter.}

{
\subsubsection{Comparison with sequential methods}\label{sec:comp_seq_methods}

As discussed above, our localization scheme serves a similar role to sequential NPE/NLE by constructing a data-adaptive proposal distribution that concentrates around the high posterior density region near $\theta^\ast$. Here, we compare two adaptive approaches, sequential NPE (SNPE) and non-amortized NPE with our localized proposal,in terms of both computational efficiency and inferential performance.

We find that SNPE works well in the queuing model (see \Cref{queuing_10exp}), achieving results comparable to non-amortized NPE with our localized proposal while requiring a smaller simulation budget. However, in the monotonic regression example, we find that the non-amortized NPE is more efficient than SNPE. In particular, the first round of SNPE training can be inaccurate in higher-dimensional parameter spaces, producing poor posterior samples that are then used to construct the proposal distribution for subsequent rounds. As a result, the errors from the first round can propagate and even amplify in later rounds, leading to progressively worse training behavior. In contrast, our localization method remains effective because it relies on an optimization-based procedure rather than iterative posterior refinement. In \Cref{fig:sw1_vs_seq}, we compare the localization results with the first round of SNPE under different simulation budgets. The results show that our localization method concentrates more effectively around the true parameter while using fewer simulations.

\begin{figure}[htbp]
    \centering
    \includegraphics[width=0.99\textwidth]{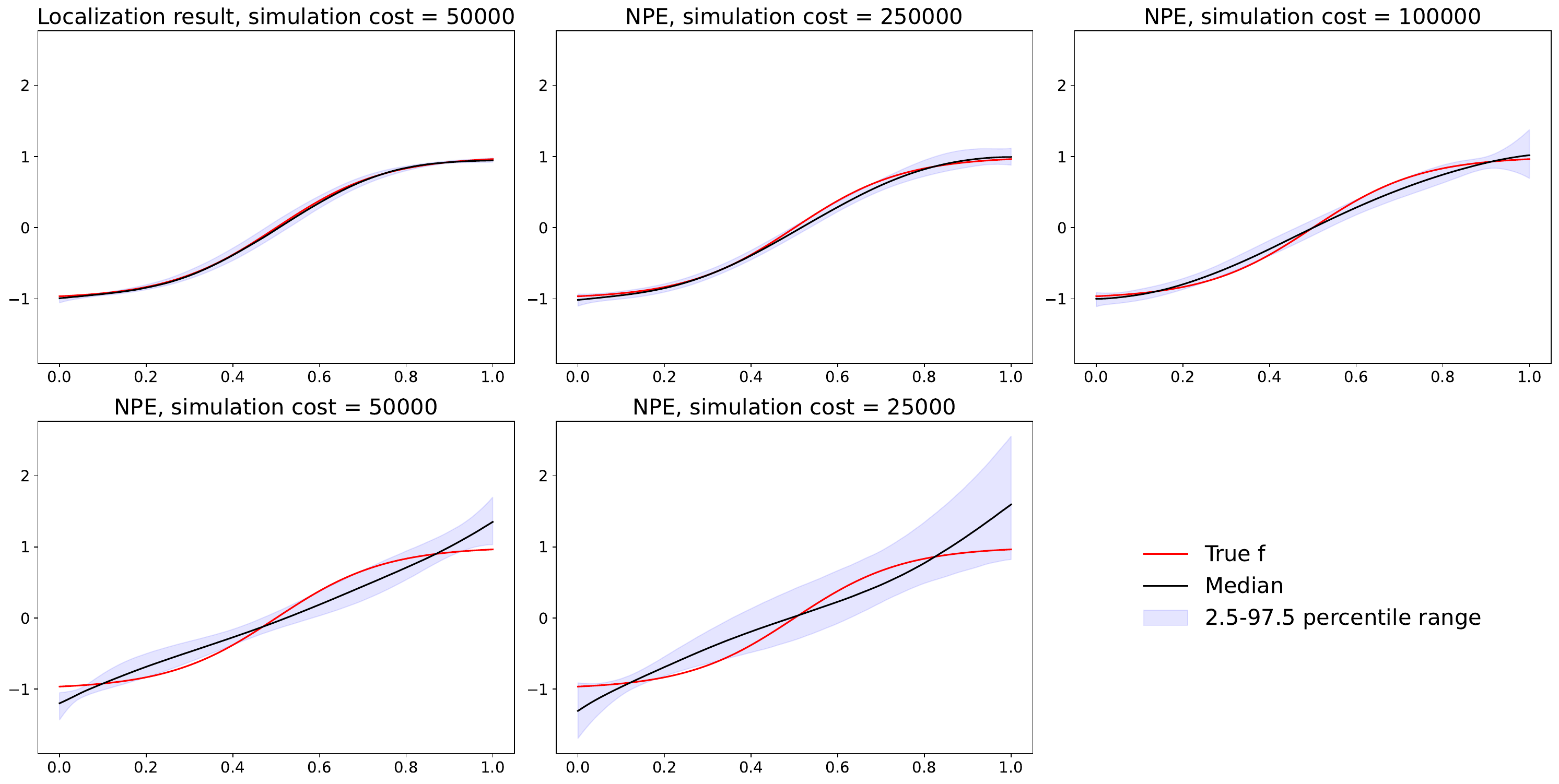}
    \caption{Comparison of the localization method and NPE under varying simulation cost}
    \label{fig:sw1_vs_seq}
\end{figure}

\subsubsection{Other choices of proposal distributions}\label{sec:other_proposal}

Once $\{\widehat\theta^{(b)}\}_{b=1}^B$ are obtained from \Cref{eq:smm_sol}, the proposal distribution can be constructed in several ways. For simplicity, we use a multivariate Gaussian with a diagonal covariance matrix. More flexible choices are possible, such as a multivariate Gaussian with a non-diagonal covariance matrix, or even more complex distributions. However, such flexibility introduces additional parameters and therefore often requires a larger $B$  for reliable estimation, increasing the simulation cost. To investigate this trade-off, we revisit the monotonic regression example, where the posterior exhibits strong correlation, and compare diagonal and non-diagonal Gaussian proposal distributions.

We present the results for both n-model and single-model in \Cref{tab:diag_vs_nondiag,tab:diag_vs_nondiag_single}.  When the covariance matrix can be accurately estimated ($B=100$), the $n$-model benefits considerably from the non-diagonal proposal, achieving substantially improved posterior sampling. This improvement is expected, as the proposal better matches the geometry of the posterior and consequently provides more accurate score estimation in high-density regions. In contrast, the single-model performs noticeably worse under the non-diagonal proposal.

\begin{table}[!ht]
    \centering
\caption{Averaged results of n-model under different proposal distribution across 10 experiments in the monotonic regression example.}
\begin{tabular}{l|cccc}
\toprule
 & KS & W1 ($\times 10^{-2}$) & Cover95 & CI95 Width \\
\midrule
Diagonal, $B=100$ & 0.159 & 0.395 & 0.985 & 0.042 \\
\midrule
Diagonal, $B=30$ & 0.162 & 0.414 & 0.989 & 0.044 \\
\midrule
Non-diagonal, $B=100$ & 0.116 & 0.262 & 0.983 & 0.038 \\
\midrule
Non-diagonal, $B=30$ & 0.165 & 0.392 & 0.987 & 0.041 \\
\bottomrule
\end{tabular}
    \label{tab:diag_vs_nondiag}
\end{table}

\begin{table}[!ht]
    \centering
\caption{Averaged results of single-model under different proposal distribution across 10 experiments in the monotonic regression example.}
\begin{tabular}{l|cccc}
\toprule
 & KS & W1 ($\times 10^{-2}$) & Cover95 & CI95 Width \\
\midrule
Diagonal, $B=100$ & 0.095 & 0.210 & 0.976 & 0.034 \\
\midrule
Diagonal, $B=30$ & 0.099 & 0.218 & 0.976 & 0.034 \\
\midrule
Non-diagonal, $B=100$ & 0.184 & 0.565 & 0.979 & 0.079 \\
\midrule
Non-diagonal, $B=30$ & 0.412 & 3.274 & 0.903 & 0.129 \\
\bottomrule
\end{tabular}
    \label{tab:diag_vs_nondiag_single}
\end{table}

We attribute this behavior to two factors. First, a highly correlated covariance matrix emphasizes a few dominant directions, corresponding to the leading eigenvectors of the covariance matrix $\Sigma$. Consequently, the simulated parameter values concentrate primarily along these directions, leaving relatively few training samples in other directions. Since the score network learns from the geometry of the training data, it estimates the score accurately only when the true score aligns with the dominant directions of $\Sigma$. This issue is much less pronounced for the $n$-model, because the score aggregated over $n$ observations is substantially less variable and is more likely to align with the dominant posterior directions. In contrast, the score for a single observation exhibits much greater directional variability, making learning in the underrepresented directions considerably more difficult. We verify this phenomenon in \Cref{fig:scoredir_covdir} by comparing the score estimation error against the alignment between the true score and the proposal covariance, measured by the Rayleigh quotient $\frac{s^\ast(\theta, X)^T \Sigma s^\ast(\theta, X)}{\lVert s^\ast(\theta, X) \rVert^2}$. The score error increases substantially when the true score is poorly aligned with the dominant directions of $\Sigma$.

\begin{figure}[!ht]
    \centering
    \includegraphics[width=0.7\textwidth]{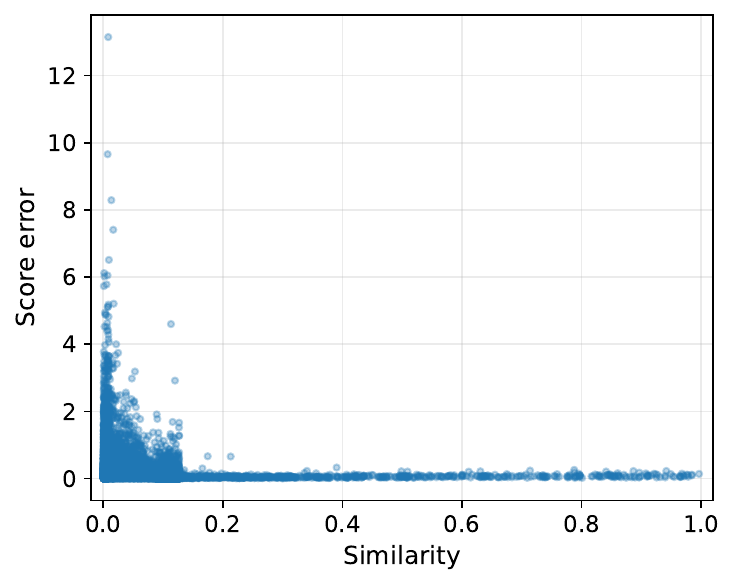}
    \caption{Score error versus alignment between true score and proposal covariance under the non-diagonal proposal.}
    \label{fig:scoredir_covdir}
\end{figure}

Second, the score matching objective in \Cref{thm:score_matching} requires evaluating the proposal score $\nabla_\theta\log p(\theta)$. When the proposal covariance contains strong correlations, this score can have a very large magnitude, making the optimization problem significantly more difficult.

Additionally, even when we train the score network by directly regressing on the true score, which avoids the previous two issues, the sampling result of using a non-diagonal covariance is just slightly better than the digonal version without true score, with the averaged KS and W1 errors equal to $0.089$ and $0.198 \times 10^{-2}$, respectively. This suggests that, for the single-model, a diagonal proposal already captures most of the benefits of localization.

Finally, the performance under the diagonal proposal is more robust to the choice of $B$ than the non-diagonal proposal. Reducing $B$ from $100$ to $30$ has little effect on either the $n$-model or the single-model under the diagonal proposal. In contrast, the performance of both models deteriorates noticeably under the non-diagonal proposal. This degradation suggests that the estimated covariance matrix becomes inaccurate when only a limited number of localization solutions are available, causing the proposal to miss important regions of the posterior. The effect is particularly severe for the single-model, where the debiasing correction becomes ineffective in these missed regions, leading to significant error accumulation, as illustrated in \Cref{fig:benefit_debias}.

Overall, we advocate using a multivariate gaussian with diagonal covariance matrix as a baseline for the proposal distribution due to (1) its simplicity; (2) it requires smaller $B$ and thus lower simulation cost in the localization step; (3) it treats all directions equally, which is important for single data score matching.

}

\subsection{Boundary condition}\label{sec:boundary_sol}

Another challenge of applying the naive implementation in \eqref{eq:score_loss} origins from the boundary condition in \Cref{ass:boundary} required by \Cref{thm:score_matching}. While this condition is essential for the validity of the optimization objective in \eqref{eq:score_loss}, it is often violated in simulation-based models. The support of $\theta$ can be constrained by its prior distribution, such as uniform distribution or non-negative distribution, which are quite common in SBI.  For these cases, the constrained support can be resolved by using the change-of-variable trick.

The more challenging case is where the support of the parameters is constrained by the data, and this cannot be addressed by the change-of-variable trick.  For example, in our queuing model example in \Cref{sec:queuing_simu}, the boundary issues arise because (1) the joint density $p(\theta, \Xn)$ does not approach $0$ as $\theta$ approaches the boundary of the prior, and (2) the support of $\theta_1$ depends on $\Xn$ as $\theta_1\leq \min\{x_{i,j}\}$.

 For the second scenario, we consider two solutions in our project. The first one requires the constrained support to be fully known.  The second one does not require such knowledge but instead introduces a small amount of noise into the simulation process to address the problem.

{\bf Solution 1: Introducing a weight function} Our first solution involves incorporating a non-negative weight function $g(\theta, \Xn): \R^d \times \R^{np} \to [0, +\infty)^d$  into the score function, which was proposed in \citet{yu2019generalized,yu2022generalized}. This weight function ensures that the product of the score function and the weight function satisfied the boundary condition.  Basically we replace the objective in \eqref{eq:obj_cond_score} with the following
\begin{equation}\label{eq:obj_cond_score_wt}
\min_{\phi}\E_{ (\theta, \Xn)\sim  p(\theta)\pn_\theta(\Xn)} \norm{s_\phi(\theta,\Xn)\odot g^\frac{1}{2} (\theta, \Xn)-\nabla_\theta \log \pn_\theta(\Xn) \odot g^\frac{1}{2} (\theta, \Xn)}^2,
\end{equation}
where $\odot$ denotes element-wise multiplication. When the constrained support is fully known, one can tailor such weight function $g(\cdot)$ such that the product  $p(\theta)\pn_\theta(\Xn )s_\phi(\theta, \Xn)\odot g^{1/2}(\theta, \Xn)$ satisfies \Cref{ass:boundary}. By incorporating the weight function $g$, we can replace all conditions on $s_\phi(\theta, \Xn)$ with conditions on $s_\phi(\theta, \Xn)\odot g^{1/2}(\theta, \Xn)$ in \Cref{ass:boundary} and the proof for \Cref{thm:score_matching} can also be amended accordingly. This ensures that the score function can still be accurately estimated even when the original boundary condition in \Cref{ass:boundary} is violated. We provide the details of how such weight function is constructed below.

First we replace the assumptions introduce for deriving the score-matching objective in \eqref{eq:score_loss} with the weighted version. We use $p(\theta)$ to denote the distribution where $\theta$ is drawn, which can be either the prior distribution $\pi(\theta)$ or the proposal distribution $q(\theta)$.

\begin{assumption}\label{g_finite_E}
$\mathbb{E}_{(\theta, \Xn) \sim p(\theta)\pn_\theta(\Xn)}\; \left[\norm{ s_\phi (\theta, \Xn) \odot g^{\frac{1}{2}}(\theta, \Xn) }^2\right]$  is finite and  $\mathbb{E}_{(\theta, \Xn) \sim p(\theta)\pn_\theta(\Xn)}$ $\left[ \norm{ \nabla_\theta \log \pn_\theta(\Xn) \odot g^{\frac{1}{2}}(\theta, \Xn)}^2 \right]$ is also finite.
\end{assumption}

\begin{assumption}\label{g_boundary}
For any $\Xn\in \mX$, we have $p(\theta)\pn_\theta(\Xn) s_\phi(\theta, \Xn) g(\theta, x)_j \rightarrow 0$ for any $\theta$ approaching $\partial \Omega(\Xn)$.
\end{assumption}

Denote $\theta_{-j}=(\theta_1, \ldots, \theta_{j-1},\theta_{j+1}, \theta_{d_\theta})^T$ and  the joint support of $(\theta,\Xn)$ as $\Omega:=\{(\theta, \Xn): p(\theta)\pn_\theta(\Xn)>0\}$. Then we define the marginal support of $(\theta_{-j}, \Xn)$ as
$\Omega_{\theta_{-j},\Xn}:=\{(\theta_j, \Xn): \exists \theta_j \text{ such that } (\theta,\Xn)\in \Omega\}$ and the section of $\theta_j$ at $(\theta_{-j},\Xn)$ as $\Omega_{\theta_j\mid \theta_{-j},\Xn}:=\{\theta_j \in \R: (\theta, \Xn)\in \Omega\}$.

\begin{assumption}\label{secD}
$\forall j \in \{1, ..., d_\theta\}$, fix any $(\theta_{-j}, \Xn) \in \Omega_{\theta_{-j}, \Xn}$, $\partial \Omega_{\theta_j\mid \theta_{-j},\Xn}$ is a countable union of intervals.
\end{assumption}

We also implicitly assume that $p(\theta)$, $s_\phi(\theta, x)$ and $g(\theta, x)$ are continuous and differentiable.

According to \citet{yu2022generalized}, we can set the $j$-the coordinate of the weight function $g_j(\theta, x) = \text{dist}(\theta_j,  \partial \Omega_{\theta_j\mid \theta_{-j},\Xn})$ to satisfy Assumption \ref{g_boundary}, where $\partial\,\Omega$  is the set of all the boundary points of $\Omega$. Furthermore, \cite{yu2019generalized,yu2022generalized} suggest that using a composite function $h \circ \text{dist}(\theta_j, \partial\, \text{Sec}(\mathcal{D};\, x,\, \theta_{-j}))$ would improve the performance, where $h(\cdot)$ is a slowly-increasing function and $h(0) = 0$, e.g. $h(t) = \log (1+t)$. In our implementation, we use simply the scaled $L_2$ distance to weight the all coordinates fairly.

\begin{theorem}[Adopted from Lemma 3.2 of \citet{yu2022generalized}]\label{g_like_score_matching}
Under \Cref{g_finite_E,g_boundary,secD}, we have
\begin{equation}\label{sm_thm_eq}
\begin{split}
&\mathbb{E}_{(\theta, \Xn) \sim p(\theta)\pn_\theta(\Xn)}\; \frac{1}{2} \norm{ s_\phi (\theta, \Xn) \odot g^{\frac{1}{2}}(\theta, \Xn) - \nabla_\theta \log \pn_\theta(\Xn) \odot g^{\frac{1}{2}}(\theta, \Xn) }^2 \\
= &\mathbb{E}_{(\theta, \Xn) \sim p(\theta)\pn_\theta(\Xn)}\; \Bigg[ \frac{1}{2} \norm{s_\phi (\theta, \Xn) \odot g^{\frac{1}{2}}(\theta, \Xn)}^2 + (s_\phi(\theta, \Xn) \odot g(\theta, \Xn))^T \nabla_\theta \log p(\theta) + \\
&\sum_{j=1}^{d_\theta} \left(\frac{\partial s_{\phi,j}(\theta, \Xn)}{\partial \theta_j} g_j(\theta, \Xn) + s_{\phi,j}(\theta, \Xn) \frac{\partial g_j(\theta, \Xn)}{\partial \theta_j} \right) \Bigg] + const,
\end{split}
\end{equation}
\end{theorem}

\noindent The proof here is very similar to \citet{yu2022generalized}. We extend the results from unconditional scores under the general domain to conditional scores under the general domain.
\begin{proof}
\begin{align*}
&\mathbb{E}_{(\theta, x) \sim p(\theta)p(x\mid \theta)}\; \frac{1}{2} \lVert s_\phi (\theta, x) \odot g^{\frac{1}{2}}(\theta, x) - \nabla_\theta \log p(x \mid \theta) \odot g^{\frac{1}{2}}(\theta, x) \rVert^2 \\
= & \frac{1}{2} \mathbb{E} \left[ \lVert s_\phi(\theta, x) \odot g^{\frac{1}{2}}(\theta, x) \rVert^2 \right] - \mathbb{E} \left[ (s_\phi (\theta, x) \odot g^{\frac{1}{2}}(\theta, x))^T (\nabla_\theta \log p(x \mid \theta) \odot g^{\frac{1}{2}}(\theta, x)) \right] \\
&\quad + \frac{1}{2}\mathbb{E} \left[ \lVert \nabla_\theta \log p(x \mid \theta) \odot g^{\frac{1}{2}}(\theta, x) \rVert^2 \right],
\end{align*}
where the first and third terms are finite under Assumption \ref{g_finite_E}, and the second term is finite due to Cauchy-Schwartz inequality. The third term is a constant in $\phi$, and the second term does not involve the unknown true score, so we only need to address the second term.

\begin{align}
&-\mathbb{E}_{(\theta,\Xn) \sim p(\theta)\pn_\theta(\Xn)}\;[(s_\phi (\theta, \Xn) \odot g^{\frac{1}{2}}(\theta, \Xn))^T (\nabla_\theta \log \pn_\theta(\Xn) \odot g^{\frac{1}{2}}(\theta, \Xn))] \notag\\
=& -\iint p(\theta) \pn_\theta(\Xn)\sum_{j=1}^{d_\theta} s_{\phi,j}(\theta, \Xn) g_j(\theta, \Xn) \frac{\partial \log \pn_\theta(\Xn)}{\partial \theta_j} \;\diff\theta \diff\Xn \notag\\
=& -\int_{\Omega_{\theta_{-j},\Xn}} \diff(\theta_{-j}, \Xn) \sum_{j=1}^{d_\theta} \int_{\partial\Omega_{\theta_j\mid \theta_{-j}, \Xn}} p(\theta) s_{\phi,j} (\theta, \Xn) g_j(\theta, \Xn) \frac{\partial\pn_\theta(\Xn)}{\partial \theta_j} \;\diff\theta_j \label{deri1}\\
\intertext{\hfill (by Fubini's Thm, and Assumption \ref{g_finite_E})} \notag
\end{align}
For simplicity, assume for now that $\partial\Omega_{\theta_j\mid \theta_{-j}, \Xn}$ is a single interval for each $j\in\{1, ..., d_\theta\}$, and denote it as $(a_j, b_j)$, then
\begin{align}
&\int_{\partial\Omega_{\theta_j\mid \theta_{-j}, \Xn}} p(\theta) s_{\phi,j} (\theta, \Xn) g_j(\theta, \Xn) \frac{\partial \pn_\theta(\Xn)}{\partial \theta_j} \;\diff\theta_j \notag\\
=& p(\theta) s_{\phi,j} (\theta, \Xn) g_j(\theta, \Xn) \pn_\theta(\Xn) \left|^{\theta_j\nearrow b_j}_{\theta_j\searrow a_j} \right. - \int_{\Omega_{\theta_j\mid \theta_{-j}, \Xn}} \frac{\partial p(\theta) s_{\phi,j} (\theta, \Xn) g_j(\theta, \Xn)}{\partial \theta_j} \pn_\theta(\Xn) \;\diff\theta_j \notag\\
\intertext{\hfill (by the fundamental law of calculus)}
=& - \int_{\Omega_{\theta_j\mid \theta_{-j}, \Xn}} \frac{\partial p(\theta) s_{\phi,j} (\theta, \Xn) g_j(\theta, \Xn)}{\partial \theta_j} \pn_\theta(\Xn) \;\diff\theta_j\text{\qquad\qquad\qquad\qquad\, (by Assumption \ref{g_boundary})} \notag\\
=& - \int_{\Omega_{\theta_j\mid \theta_{-j}, \Xn}} \{\frac{\partial p(\theta)}{\partial \theta_j}s_{\phi,j} (\theta, \Xn) g_j(\theta, \Xn) + p(\theta) \frac{\partial s_{\phi,j} (\theta, \Xn) g_j(\theta, \Xn)}{\partial \theta_j} \} \pn_\theta(\Xn) \;\diff\theta_j \notag\\
=& - \int_{\Omega_{\theta_j\mid \theta_{-j}, \Xn}} \{\frac{\partial \log p(\theta)}{\partial \theta_j}s_{\phi,j} (\theta, \Xn) g_j(\theta, \Xn) + \frac{\partial s_{\phi,j} (\theta, \Xn) g_j(\theta, \Xn)}{\partial \theta_j} \} p(\theta) \pn_\theta(\Xn) \;\diff\theta_j \label{int_trick}
\end{align}
It is worth mentioning that although $\frac{\partial g(\theta, \Xn)}{\partial \theta_j}$ is discontinuous at the middle of the interval under our distance-based definition of $g(\theta, \Xn)$, the second line is still valid, as $g(\theta, \Xn)$ is continuous. Besides, it is easy to see that \eqref{int_trick} still holds when $\Omega_{\theta_j\mid \theta_{-j}, \Xn}$ is not an interval but a countable union of intervals. Therefore, we can plug the result of \eqref{int_trick} into \eqref{deri1} and apply Fubini's Theorem again, then the proof is completed.
\end{proof}

 {\bf Solution 2: Smoothing the boundary by adding random noise} Our second solution draws inspiration from  denoising score-matching methods \citep{lu2022maximum}.
 This approach is helpful when designing a complex weighting function is impractical or when the dependency of support is unclear.

Essentially we would revise the data generating process from $\Pn_\theta$ by applying some Gaussian smoothing. The new process $\wt P^{(n)}_{\theta,\sigma_\varepsilon}$ is defined as
\[
\theta\sim \pi(\theta), \, \Xn^\theta \sim \Pn_\theta, \, \wt \Xn^{\theta} := \Xn^\theta + {\boldsymbol\varepsilon}_n \sim \wt P^{(n)}_{\theta,\sigma_\varepsilon}\]
where ${\boldsymbol\varepsilon}_n \iid N(0, \sigma_\varepsilon^2I_{np})$. By introducing this noise, the support of $\wt \Xn^{\theta}$ is $\R^{np}$ and unconstrained, resolving the boundary condition issue. Furthermore, since the noise is independent of $\theta$, the score function of $\wt p^{(n)}_{\theta,\sigma_\varepsilon}(\wt X^{(n)}_\theta)$ can be expressed as
\[
\nabla_\theta \log \wt p^{(n)}_{\theta,\sigma_\varepsilon}(\wt X^{(n)}) = \E_{\varepsilon^{(n)}} \Big[\nabla_\theta \log \pn_\theta(\Xn) \mid \Xn+\varepsilon^{(n)}=\wt X^{(n)} \Big]
\]
Naturally when  $\sigma_\varepsilon\to 0$, we have
\[
\nabla_\theta\log \wt p^{(n)}_\theta(\wt X^{(n)}) \to  \nabla_\theta\log  \pn_\theta(X^{(n)}).
\]
This is similar to the denoising score-matching method \citep{lu2022maximum}, where the noise level $\sigma_\varepsilon$ is gradually reduced to zero. While \citet{lu2022maximum} uses the noise as a simulated-annealing strategy to   avoid local optima, we use it to resolve the boundary condition issue.  In our implementation, we set the noise level according to the variation in the datasets.

We apply the two solutions to our simulations on the M/G/1-queuing model, and we provide more discussions how to implement the two solutions in \Cref{sec:queuing_details}.

\subsection{Full data score matching}\label{sec:version1}

In this version, we focus on matching the full-data score across all $n$
observations, which requires generating the reference table
$\mD = \{(\theta^{(k)}, \Xn^{(k)})\}_{k=1}^N \iid q(\theta)\,\pn_\theta(\Xn)$ using the
localized proposal distribution $q$ for training. We first still focus on the scenario of i.i.d. datasets and present the theoretical analysis similar to \Cref{thm:post_convergence_iid}. Lastly we conclude with a discussion on how this can be generalized to non i.i.d. data setting.

The additive structure allows
us to simplify the score network to $s_\phi(\theta, X)$, so that the estimated
full-data score becomes $\sum_{i=1}^n s_\phi(\theta, X_i)$. To enforce the
curvature structure, we add a curvature-matching penalty to regularize the
score network and replace the objective in \eqref{eq:score_loss} with the
following:
{\footnotesize\begin{equation}
    \min_\phi  \  \E_{q(\theta)}  \bigg[ \E_{\pn_\theta}  \Big[\ \underbrace{\frac{1}{2n} \Big \|\sum_{i=1}^n \big(s_\phi(\theta,X_i)- s^\ast(\theta,X_i)\big)\Big\|^2}_\text{score-matching loss on $\Xn$} \Big] + \lambda \,\underbrace{\Big\|\E_{ p_\theta}\big[s_\phi(\theta,X)s_\phi(\theta,X)^T + \nabla_\theta s_\phi(\theta,X)\big]\Big\|_F^2}_\text{curvature-matching loss}\bigg],\label{eq:score_loss_struct}
\end{equation}}
where $\lambda>0$ is a hyperparameter that controls the strength of the curvature regularization, and $\norm{\cdot}_F$ denotes the Frobenius norm. Note that we introduce the scaling $1/n$ in the score-matching loss to balance the contribution of the two terms in the final loss. In practice, we approximate the expectation in \eqref{eq:score_loss_struct} by
the empirical average over the reference table $\mD$ as
\begin{equation}\label{eq:score_loss_regularized_empirical}
\begin{split}
\min_\phi &\frac{1}{N}\sum_{k=1}^N \Bigg\{ \frac{1}{n}\bigg[\frac{1}{2} \norm{\sum_{i=1}^n s_\phi(\theta^{(k)}, X_i^{(k)})}^2 + \big(\sum_{i=1}^n s_\phi(\theta^{(k)}, X_i^{(k)})\big)^T \nabla_\theta \log \pi(\theta)\big|_{\theta=\theta^{(k)}} \\
&+\sum_{j=1}^{d_\theta}\sum_{i=1}^n \frac{\partial s_{\phi,j}(\theta, X_i^{(k)})}{\partial \theta_{j}}\big|_{\theta=\theta^{(k)}} \bigg]   + \lambda \,\underbrace{\Big\|\frac{1}{n}\sum_{i=1}^n \big[s_\phi(\theta^{(k)},X^{(k)}_i)s_\phi(\theta^{(k)},X^{(k)}_i)^T + \nabla_\theta s_\phi(\theta^{(k)},X^{(k)}_i)\big]\Big\|_F^2}_\text{curvature-matching loss}\Bigg\}
\end{split}
\end{equation}

An algorithmic overview of the method is provided in  \Cref{alg:regularized_langevin}.

\begin{algorithm}[ht]
    \caption{Langevin Monte Carlo with regularized score matching}\label{alg:regularized_langevin}
    \begin{algorithmic}
            \STATE {\bf Input:} Prior distribution $\pi(\theta)$, observed dataset $\Xn^\ast$, number of particles $N$, number of Langevin steps $K$, step size $\tau_n$, score network $s_\phi(\theta, X)$, initial value $\theta^{(0)}$.
        \end{algorithmic}
        \hrule
        \begin{algorithmic}
        \STATE {\bf 1. Localization:} Construct a proposal distribution $q(\theta)$ using \eqref{eq:smm_sol}.
            \STATE {\bf 2. Reference Table:} Generate $\mD = \{(\theta^{(k)}, \Xn^{(k)})\}_{k=1}^N\iid q(\theta)\, \pn_\theta(\Xn)$.
            \STATE {\bf 3. Network Training:} Train $s_\phi(\theta, X)$ on $\mD$ using loss in \eqref{eq:score_loss_regularized_empirical} and obtain $\widehat \phi$.
            \STATE {\bf 4. Langevin Sampling:} {For} {$k = 1$ to $K$}
                \STATE \hspace{2em} $\theta^{(k)} \gets \theta^{(k-1)} + \tau_n\Big(\sum_{i=1}^n s_{\widehat\phi}\big(\theta^{(k-1)}, X^\ast_{i}\big)+ \nabla_\theta \log \pi(\theta^{(k)}) \Big) + \sqrt{2\tau_n}\,U_k$, \quad $U_k \iid \mathcal{N}(0, I_{d_\theta})$.
            \end{algorithmic}
            \hrule
            \begin{algorithmic}
            \STATE {\bf Return} $\{\theta^{(k)}\}_{k=1}^K$ as approximated posterior samples
    \end{algorithmic}
    \end{algorithm}

In the actual implementation, we randomly partition the data into a training set (50\%) and a validation set (50\%). We first initialize the neural network with only the score-matching loss without penalty on curvature as in \eqref{eq:score_loss_regularized_empirical}. Next,
we continue training our neural networks with the penalized loss \eqref{eq:score_loss_regularized_empirical} and use the total score-matching loss on $n$ data as in \eqref{eq:score_loss_regularized_empirical} evaluated on the validation set to select the optimal $\lambda$. The score-matching loss and the curvature-matching loss are evaluated on the same reference table during the training process.

Now we continue on the theoretical analysis of \Cref{alg:regularized_langevin}. Similar to \Cref{ass:uniform_sm_err_single}, we introduce the assumption on uniform estimation errors.

\begin{assumption}[Uniform estimation error]\label{ass:uniform_sm_err} Under the same set $\mA_{n,1}$ defined in \Cref{ass:uniform_sm_err_single}, we assume the score-matching error in this neighborhood is uniformly bounded as
    \[
    \varepsilon_{N,n,1}^2:=\sup_{\theta\in \mA_{n,1}}\E_{\Xn\sim \pn_\theta} \norm{\frac{1}{\sqrt n} \widehat s(\theta, \Xn)- \frac{1}{\sqrt n}s^\ast(\theta, \Xn)}_2^2,
    \]
and the curvature-matching error is also uniformly bounded as
\[
\varepsilon_{N,n,2}^2:=\sup_{\theta\in \mA_{n,1}}  \norm{\E_{X\sim p_\theta}[\nabla_\theta \widehat s(\theta, X)+\widehat s(\theta, X)\widehat s(\theta, X)^T]}_F^2.
\]
\end{assumption}

Under the simialr argument to \Cref{ass:uniform_sm_err_single}, we localize the uniform estimation error to the set $\mA_{n,1}$. We consider the $1/\sqrt n$ scaling in the score matching error since we focus on the non-degenerate transformed variable $\sqrt n(\theta-\theta^*)$ in our analysis (similar to \Cref{ass:log_sobolev}). This also matches with our scaling in \eqref{eq:score_loss_regularized_empirical} to balance the contribution of the score loss and the curvature loss in training. The score-matching error $\epsilon_{N,n,1}$ here depends on the complexity of the true score $s^*(\theta, \Xn)$ and the size $(N,n)$ of the reference table $\mD$. The curvature error $\epsilon_{N,n,2}$ is mainly determined by the Monte Carlo approximation of expectations, which is decaying at rate $1/\sqrt n$ in this case.

\begin{theorem}[Posterior approximation error under full data score matching]\label{thm:post_convergence_n_obs} Under \Cref{ass:lipschitz_score,ass:log_sobolev,ass:true_converge,ass:bdd_derivatives,ass:uniform_sm_err}, and assume that $I(\theta^\ast)<\infty$. If the step size $\tau_n$ and initial distribution of the Langevin Monte Carlo satisfy
\[
\tau_n = O\Big( \frac{1}{d_\theta C_{LSI}\lambda_L^2n}\Big) \quad \mbox{and}\quad d_{\chi^2}\big(\widehat \pi^0_n(\cdot\mid \Xn^\ast)|| \pi_n(\cdot\mid \Xn^\ast)\big)\leq \eta_\chi^2,
\]
then we have
\begin{align*}
& \E_{\Xn^\ast\sim \Pn_{\theta^\ast}} \Big[d_{TV}^2\big(\widehat\pi_n(\cdot\mid \Xn^\ast), \pi(\cdot\mid \Xn^\ast)\big)\Big]\\
&\qquad\lesssim \ \underbrace{\exp\Big(-\frac{n K\tau_n}{5C_\text{LSI}}\Big)\, \eta_\chi^2}_{\text{burn-in error}} \ + \  \underbrace{d_\theta C_\text{LSI} \lambda_L^2 n\tau_n}_\text{discretization error} \ +\ \underbrace{ \varepsilon_n (Kn \tau_n +\eta_\chi C_\text{LSI})}_\text{score error}
\end{align*}
where $\varepsilon_n^2= (\log n)^2 \varepsilon_{N,n,1}^2 +  (\log n)^2 \varepsilon_{N,n,2}^2 + (\log n)^3/n$
\end{theorem}

We provide the proof in \Cref{sec:post_convergence_n_obs_pf}. The error decomposition here is similar to \Cref{thm:post_convergence_iid} except the score error. The score error now only depends on the score-matching error $\varepsilon_{N,n,1}$, which is defined differently from $\wt\varepsilon_{N,1}$, the curvature-matching error $\epsilon_{N,n,2}$.
To ensure a diminishing error $\varepsilon_n=o(1)$ as $n\to\infty$, we need both $\varepsilon_{N,n,1}$ and $\varepsilon_{N,n,2}$ to converge at least at the rate of $1/(\log n)$.  This suggests that we should set $\lambda=\mO(1)$ to balance the two errors in \eqref{eq:score_loss_struct}. The Monte Carlo error $\varepsilon_{N,n,2}$ is decaying at the rate $1/\sqrt n$. For the score-matching error $\varepsilon_{N,n,1}$, see our discussion in \Cref{rmk:sm_err} regarding how to choose $N$.

The regularization idea here can be naturally extended to the general
non i.i.d.~data setting by introducing a computationally more costly
curvature-matching regularization term involving the full-data score, given by
$\big\|\E_{\Pn_\theta}\!\big[s_{\phi}(\theta, \Xn)s_{\phi}(\theta, \Xn)^T
+ \nabla_\theta s_{\phi}(\theta, \Xn)\big]\big\|_F^2$.
To evaluate this penalty, we construct a separate reference table in which
multiple independent $\Xn^{(k)}$ are simulated for each $\theta_{(k)}$, so that the
expectation with respect to $\Pn_\theta$ can be approximated by an empirical
average. Although this procedure increases simulation cost, our theoretical
analysis in \Cref{thm:post_convergence_n_obs} demonstrates that
curvature-matching is essential: it ensures that the estimated score remains
accurate when $\theta$ deviates slightly from the true parameter $\theta^\ast$,
which is critical for the stability of subsequent Langevin sampling.

\subsection{Full data score estimation via single data score matching}\label{sec:regularization_details}

We first provide an algorithm view of the method in \Cref{sec:regular} in \Cref{alg:langevin_single_obs}.

In the algorithm, we generate two reference tables $\mD^S = \{(\theta^{(k)}, X^{(k)})\}_{k=1}^N \iid q(\theta)\,p_\theta(\cdot)$ and $\mD^R = \{(\wt{\theta}^{(l)}, \wt{\X}_{m_R}^{(l)})\}_{l=1}^{N_R} \iid q(\theta)\,p^{(m_R)}_\theta(\cdot)$, where $q$ is the localized proposal distribution, and each $\wt{\X}_{m_R}^{(l)} = (\wt{X}_1^{(l)},\dots, \wt{X}_{m_R}^{(l)})^T$ is a dataset of $m_R$ observations.  Here we use a slightly different notation for samples in $\mD^R$ from our main text in order to differentiate the samples in different tables.
$\mD^S$ is used for evaluating the score matching loss, and $\mD^R$ is used for evaluating the curvature loss and the mean-regression loss. Concretely, the regularized score matching is conducted with the empirical mean of \eqref{eq:score_loss_struct_single}:
\begin{equation}\label{eq:emp_regu_score_loss}
\begin{split}
\min_\phi\Big\{& \frac{1}{N}\sum_{k=1}^N \Big[ \frac{1}{2} \norm{s_\phi(\theta^{(k)}, X^{(k)})}^2 + s_\phi (\theta^{(k)}, X^{(k)})^T \nabla_\theta \log \pi(\theta)\big|_{\theta=\theta^{(k)}} +\sum_{j=1}^{d_\theta} \frac{\partial s_{\phi,j}(\theta, X^{(k)})}{\partial \theta_{j}}\big|_{\theta=\theta^{(k)}}  \Big] \\
&+ \lambda_1 \frac{1}{N_R}\sum_{l=1}^{N_R} \Big\lVert \frac{1}{m_R}\sum_{i=1}^{m_R}\big[s_\phi(\wt{\theta}^{(l)}, \wt{X}_i^{(l)}) s_\phi(\wt{\theta}^{(l)}, \wt{X}_i^{(l)})^T + \nabla_\theta s_\phi(\wt{\theta}^{(l)}, \wt{X}_i^{(l)})\big] \Big\rVert_F^2  \Big\}
\end{split}
\end{equation}
and the mean regression is conducted with the empirical mean of \eqref{eq:score_demean}:
\begin{equation}\label{eq:emp_mean_reg_loss}
\begin{split}
\min_\psi& \frac{1}{N_R} \sum_{l=1}^{N_R} \bigg[ \Big\lVert h_\psi(\wt{\theta}^{(l)}) - \frac{1}{m_R}\sum_{i=1}^{m_R} s_{\hat{\phi}}(\wt{\theta}^{(l)}, \wt{X}_i^{(l)}) \Big\rVert^2 + \lambda_2 \Big\lVert h_\psi(\wt{\theta}^{(l)})h_\psi(\wt{\theta}^{(l)})^T - \nabla_\theta h_\psi(\wt{\theta}^{(l)}) \\
&- \Big[\frac{1}{m_R}\sum_{i=1}^{m_R} s_{\hat{\phi}}(\wt{\theta}^{(l)}, \wt{X}_i^{(l)})\Big]h_\psi(\wt{\theta}^{(l)})^T - h_\psi(\wt{\theta}^{(l)}) \Big[\frac{1}{m_R}\sum_{i=1}^{m_R} s_{\hat{\phi}}(\wt{\theta}^{(l)}, \wt{X}_i^{(l)})^T\Big] \Big\rVert^2 \bigg]
\end{split}
\end{equation}
In implementation, we randomly partition each of $\mathcal{D}^S$ and $\mathcal{D}^R$ into a training set (50\%) and a validation set (50\%). We first initialize the neural network using only the score-matching loss without penalty on curvature as in \eqref{eq:emp_regu_score_loss}. Then,
we continue training our neural networks with the penalized loss \eqref{eq:emp_regu_score_loss} and use the score-matching loss in \eqref{eq:emp_regu_score_loss} evaluated on the validation set to select the optimal $\lambda_1$. Next, we fix $s_{\hat{\phi}}$ and use $h_\psi$ to estimate its mean using \eqref{eq:emp_mean_reg_loss}. Still, we first initialize the network $h_\psi$ using only the regression loss without penalty on curvature as in \eqref{eq:emp_mean_reg_loss}, and then continue training $h_\psi$ with the penalized loss \eqref{eq:emp_mean_reg_loss}, where we use the regression loss in \eqref{eq:emp_mean_reg_loss} on the validation set to select the optimal $\lambda_2$.

\begin{algorithm}[ht]
    \caption{Langevin Monte Carlo with debiased score matching}\label{alg:langevin_single_obs}
    \begin{algorithmic}
            \STATE {\bf Input:} Prior distribution $\pi(\theta)$, observed dataset $\Xn^\ast$, number of particles $N$, number of Langevin steps $K$, step size $\tau_n$, networks $s_\phi(\theta, X)$ and $h_\psi(\theta)$, initial value $\theta^{(0)}$.
        \end{algorithmic}
        \hrule
        \begin{algorithmic}
        \STATE {\bf 1. Localization:} Construct a proposal distribution $q(\theta)$ using \eqref{eq:smm_sol}.
            \STATE {\bf 2. Reference Table:} Generate $\mD^S = \{(\theta^{(k)}, X^{(k)})\}_{k=1}^N\iid q(\theta)p_\theta(\cdot)$ and $\mD^{R} = \{\theta^{(l)}, \X^{(l)}_{m_R}\}_{l=1}^{N_R} \iid q(\theta)p^{(m_R)}_\theta(\cdot)$.

            \STATE {\bf 3. Network Training:} Train $s_\phi(\theta, X)$ on $\mD^S$ and $\mD^R$ using loss in \eqref{eq:emp_regu_score_loss} and obtain $\widehat \phi$.
            \STATE {\bf 4. Mean Regression:} Estimate the mean of $s_{\widehat\phi}(\theta, X)$ on $\mD^R$ using \eqref{eq:emp_mean_reg_loss} and obtain $\widehat\psi$.
            \STATE {\bf 5. Debiasing:}  $\wt s(\theta, X)= s_{\widehat\phi}(\theta, X)- h_{\widehat\psi}(\theta)$.
            \STATE {\bf 6. Langevin Sampling:}  {For} {$k = 1$ to $K$}
                \STATE \hspace{2em} $\theta^{(k)} \gets \theta^{(k-1)} + \tau_n\Big(\sum_{i=1}^n \tilde s\big(\theta^{(k-1)}, X^*_{i}\big)+ \nabla_\theta \log \pi(\theta^{(k)}) \Big) + \sqrt{2\tau_n}\,U_k$, \quad $U_k \iid \mathcal{N}(0, I_{d_\theta})$.
            \end{algorithmic}
            \hrule
            \begin{algorithmic}
            \STATE {\bf Return} $\{\theta^{(k)}\}_{k=1}^K$ as approximated posterior samples
    \end{algorithmic}
    \end{algorithm}

\begin{remark}[Weakly dependent data] \Cref{alg:langevin_single_obs} can also be generalized to weakly dependent
settings. For example, many time-series models, such as $MA(1)$ or the Lotka--Volterra
model, have the Markov property that the current state $X_i$ depends only on the last
state $X_{i-1}$ (or a small number of lags). In such cases, the full data likelihood can
still be factorized into conditional terms as
$\pn_\theta(\Xn) = \prod_{i=1}^{n-1} p(X_i \mid X_{i-1}, \theta)$ (with $X_0$
as the initial state). The resulting score function continues to satisfy the
three statistical properties:
1.~\emph{additive structure:}
   $s^\ast(\theta, \Xn) = \sum_{i=1}^{n-1} s^\ast(\theta, X_{i-1},X_{i})$.
2.~\emph{curvature structure:}
   $\E_{p(\cdot \mid X_{i-1}, \theta)}\!\left[s^\ast(\theta, X_{i-1}, X_{i})
   s^\ast(\theta, X_{i-1}, X_{i})^T
   + \nabla_\theta s^\ast(\theta, X_{i-1}, X_{i})\right] = 0$.
3.~\emph{mean-zero structure:}
   $\E_{p(\,\cdot\, \mid X_{i-1}=x, \theta)}[s^\ast(\theta, x, X_{i})] = 0$ for each $x$.
Thus, the two-step debiased score-matching procedure can still be applied after
accounting for the dependency structure by replacing the individual-level score
function $s(\theta, x)$ with a blockwise score function $s(\theta, x, x')$ for
approximating $\log p(X_i = x' \mid X_{i-1} = x, \theta)$. This modification
retains the benefits of structural regularization while reducing simulation costs.
\end{remark}

{
\subsection{Benefits of curvature matching}\label{sec:curvature_benefit}
Here we provide the details of the example shown in \Cref{fig:relative_error_heatmap} in \Cref{sec:curvature_structure}. We consider $n$ observations $X_i \iid \text{Categorical}(\theta_1, \theta_2, 1-\theta_1-\theta_2)$ for $i =1, \ldots, n$, where $\theta = (\theta_1, \theta_2)$ satisfies $\theta_1\geq 0, \theta_2\geq 0$ and $\theta_1+\theta_2 \leq 1$. We use a homogeneous Dirichlet prior $\text{Dir}(3)$ for $\theta$.

We compare two models trained with and without the curvature penalty, under the same simulation budget, same network complexity and same debiasing. We are interested in the score estimation error $\E_{\Xn^\ast \sim P_{\theta^\ast}^n} \normsm{\hat{s}(\theta, \Xn^\ast) - s^\ast(\theta, \Xn^\ast)}$ when $\theta$ is in a neighborhood of $\theta^\ast$. We also evaluate the curvature estimation error $\normsm{\hat I(\theta)- I^*(\theta)}$. In particular, we use the relative error normalized by the magnitude of the true score and curvature, defined as
\[
\frac{\E_{\Xn^\ast \sim P_{\theta^\ast}^n} \normsm{\hat{s}(\theta, \Xn^\ast) - s^\ast(\theta, \Xn^\ast)}^2}{\E_{\Xn^\ast \sim P_{\theta^\ast}^n} \normsm{s^\ast(\theta, \Xn^\ast)}^2} \text{ and } \frac{\normsm{\hat I(\theta)- I^*(\theta)}^2}{\normsm{I^*(\theta)}^2},
\]
to account for the scale of the score and curvature at different $\theta$.

\Cref{fig:relative_error_heatmap} in the main text shows the relative score error across the 2-dimensional parameter space, overlaid with posterior contours. We provide two additional plots in \Cref{fig:benefit_curv} to further illustrate the benefits of curvature matching. \Cref{fig:error_score_and_curv} shows maximum relative errors at different values of $\normsm{\theta - \theta^\ast}$, demonstrating that the curvature penalty leads to smaller curvature errors, which in turn leads to smaller score errors for $\theta$ near $\theta^\ast$. This is because the curvature structure characterizes how the score changes when $\theta$ deviates from $\theta^\ast$, and thus a better estimation of curvature leads to a more accurate estimation of the score in a neighborhood of $\theta^\ast$. Finally, \Cref{fig:marginal_density} shows that reduced score error in the posterior region leads to more accurate LMC posterior samples.

\begin{figure}[!ht]
    \centering
    \begin{subfigure}[b]{0.8\textwidth}
        \centering
        \includegraphics[width=\textwidth]{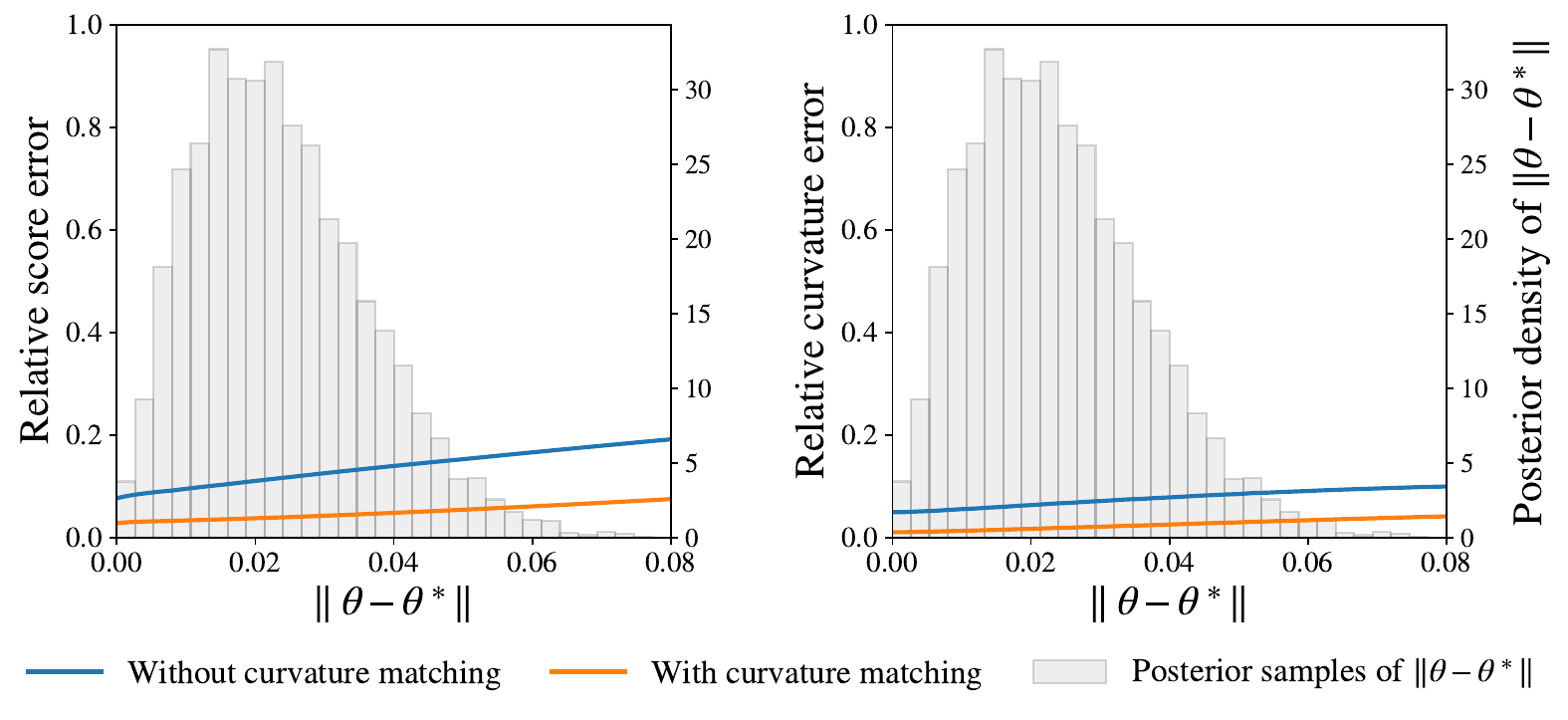}
        \caption{Relative score error and relative curvature error v.s. $\lVert \theta - \theta^\ast \rVert$}
        \label{fig:error_score_and_curv}
    \end{subfigure}
    \hfill
    \begin{subfigure}[b]{0.8\textwidth}
        \centering
        \includegraphics[width=\textwidth]{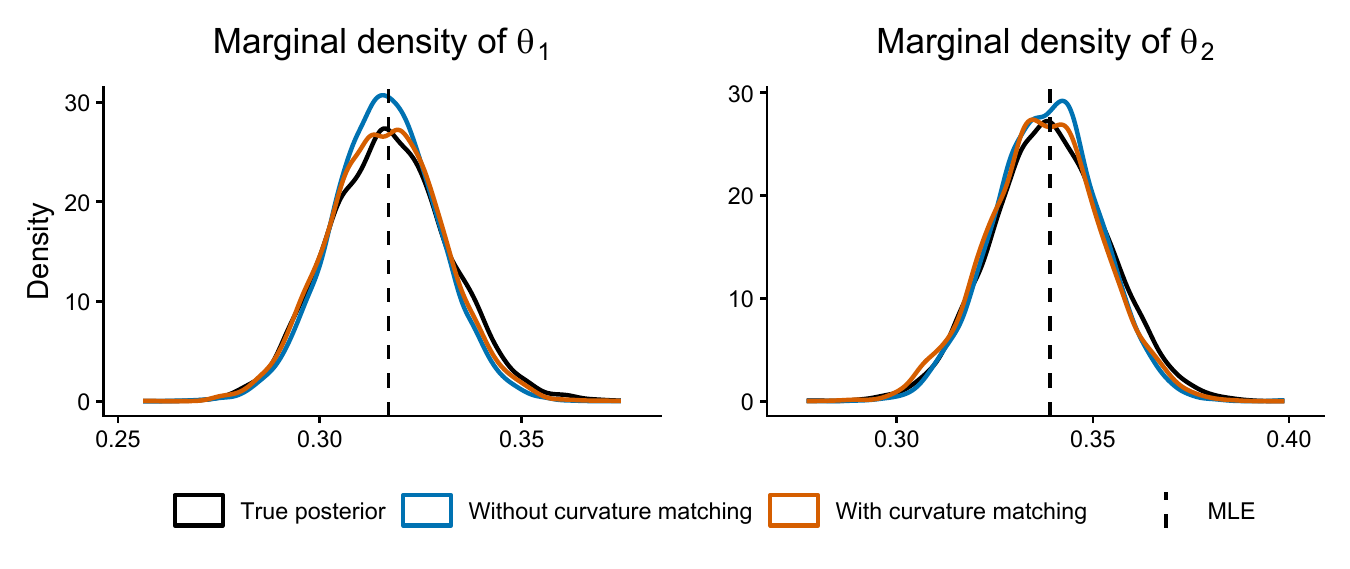}
        \caption{Marginal density of the posterior sampled using the two models}
        \label{fig:marginal_density}
    \end{subfigure}
    \caption{Benefits of the curvature penalty}
    \label{fig:benefit_curv}
\end{figure}
}

\subsection{Implementations of debiased score matching}\label{sec:alter_ways_demean}
In this subsection, we first introduce the debiased estimator used in \Cref{alg:langevin_single_obs}.

\begin{equation}\label{eq:score_demean}
 \begin{split}
 \widehat \psi =& \arg\min_{\psi} \mathbb{E}_{q(\theta)} \bigg[\normbig{h_\psi(\theta)- \E_{p_\theta} \big[s_{\widehat\phi} (\theta, X)\big] }^2 \\
 &+\lambda_2  \big\lVert h_\psi(\theta)h_\psi(\theta)^T - \nabla_\theta h_\psi(\theta) - \E_{p_\theta} \big[s_{\widehat\phi} (\theta, X)\big] h_\psi(\theta)^T - h_\psi(\theta) \E_{p_\theta} \big[s_{\widehat\phi} (\theta, X)^T\big]  \big\rVert_F^2 \bigg]
 \end{split}
 \end{equation}
 where $\lambda_2$ is again a hyperparameter that controls the strength of the curvature penalty. Let $\widehat \psi$ denote the solution to the above problem. The second curvature penalty ensures that the final debiased score, defined as $\wt s(\theta, X) = s_{\widehat\phi}(\theta, X) - h_{\widehat\psi}(\theta)$, continues to satisfy the curvature structure, which is verified below.

Using the triangle inequality, we rewrite the curvature-matching on $\wt s(\theta,X)$ as
\begin{align}
&\Big\lVert \E_{p_\theta}\Big[\big(s_{\widehat\phi} (\theta, X) - h_{\psi}(\theta)\big) \big(s_{\widehat\phi} (\theta, X) - h_{\psi}(\theta)\big)^T + \nabla_\theta\big(s_{\widehat\phi} (\theta, X) - h_{\psi}(\theta)\big)\Big] \Big\rVert_F \notag \\
\leq& \lVert \E_{p_\theta}\big[s_{\widehat\phi} (\theta, X)s_{\widehat\phi} (\theta, X)^T + \nabla_\theta s_{\widehat\phi} (\theta, X)\big] \rVert_F \label{eq:deb_est_curv_part1} \\
&+ \big\lVert h_\psi(\theta)h_\psi(\theta)^T - \nabla_\theta h_\psi(\theta) - \E_{p_\theta} \big[s_{\widehat\phi} (\theta, X)\big] h_\psi(\theta)^T - h_\psi(\theta) \E_{p_\theta} \big[s_{\widehat\phi} (\theta, X)^T\big]  \big\rVert_F \label{eq:deb_est_curv_part2}
\end{align}
where \eqref{eq:deb_est_curv_part1} is controlled during the score matching \eqref{eq:score_loss_struct_single}, and \eqref{eq:deb_est_curv_part2} is incorporated into the mean regression objective \eqref{eq:score_demean}. Thus, we ensure that with penalty in \eqref{eq:score_demean}, the debiased network has smaller bias while preserving the curvature structure. It is worth mentioning that $h_\psi$ can be any regression model not limited to neural networks, but we find empirically that a neural network with the same hidden layer structure as $s_\phi$ performs the best.

Next, we introduce two alternative approaches other than \Cref{alg:langevin_single_obs}.

{\bf Alternative 1:} If we further decompose
\eqref{eq:deb_est_curv_part2}, we have
\begin{align*}
\eqref{eq:deb_est_curv_part2} \leq& \big\lVert -h_\psi(\theta)h_\psi(\theta)^T - \nabla_\theta h_\psi(\theta)\big\rVert_F \\
&+ \big\lVert2h_\psi(\theta)h_\psi(\theta)^T - \E_{p_\theta} \big[s_{\widehat\phi} (\theta, X)\big] h_\psi(\theta)^T - h_\psi(\theta) \E_{p_\theta} \big[s_{\widehat\phi} (\theta, X)^T\big]  \big\rVert_F \\
=& \underbrace{\big\lVert h_\psi(\theta)h_\psi(\theta)^T + \nabla_\theta h_\psi(\theta)\big\rVert_F}_{\text{A1.1}} + 2 \underbrace{\big\lVert h_\psi(\theta) - \E_{p_\theta} \big[s_{\widehat\phi} (\theta, X)^T\big] \big\rVert_2}_{A1.2} \cdot \underbrace{\big\lVert h_\psi(\theta) \big\rVert_2}_{A1.3}
\end{align*}
In this decomposition, A1.2 is the regression error, and A1.3 can be bounded by A1.2 and the bias of the score network $\big\lVert \E_{p_\theta} \big[s_{\widehat\phi} (\theta, X)^T\big] \big\rVert_2 \le \sqrt{\E_{p_\theta}\big\lVert s_{\widehat\phi} (\theta, X)^T - s^\ast(\theta, X) \big\rVert_2^2}$. Since both A1.2 and A1.3 can be well controlled, it suffices to control A1.1, and we have the following alternative objective to replace \eqref{eq:score_demean} as
\begin{equation*}
\begin{split}
\widehat \psi =& \arg\min_{\psi} \mathbb{E}_{q(\theta)} \bigg[\normbig{h_\psi(\theta)- \E_{p_\theta} s_{\widehat\phi} (\theta, X)}^2 + \lambda_2  \big\lVert h_\psi(\theta)h_\psi(\theta)^T + \nabla_\theta h_\psi(\theta)\big\rVert_F^2 \bigg]
\end{split}
\end{equation*}
This objective here has a simpler form than \eqref{eq:score_demean}, although their computational cost is nearly the same.

{\bf Alternative 2:} This alternative utilizes the idea of projected gradient descent. We now write $\wt{s}_{\phi,\psi}(\theta, X) = s_{\phi}(\theta, X) - h_{\psi}(\theta)$.

We want to optimize $\wt{s}_{\phi,\psi}(\theta, X)$ by minimizing the regularized score matching loss with curvature penalty, while subject to the mean-zero constraint.
Similar to projected gradient descent methods, we alternatively minimize the regularized score loss and project the score model onto the mean-zero model family, where the projection is again enforced by mean regression.  Essentially, we iterate between the following two steps until convergence.
\begin{enumerate}
    \item Minimize the regularized score loss:
   {\small \begin{equation*}
    \min_{\phi,\psi}  \E_{q(\theta)}  \bigg[ \E_{p_\theta}  \Big[
    \|\wt{s}_{\phi,\psi}(\theta,X) - s^\ast(\theta,X)\|^2 \Big]
    + \lambda_1 \,\Big\| \E_{ p_\theta}\!\big[\wt{s}_{\phi,\psi}(\theta,X)\wt{s}_{\phi,\psi}(\theta,X)^T
    + \nabla_\theta \wt{s}_{\phi,\psi}(\theta,X)\big]\Big\|_F^2 \bigg],
    \end{equation*}}
    \item Projection:
    \[
    \min_{\psi} \mathbb{E}_{q(\theta)} \bigg[\normbig{h_\psi(\theta)- \E_{p_\theta} s_{\widehat\phi} (\theta, X)}^2 \bigg]
    \]
\end{enumerate}
When this procedure converges, we will get a score model within the mean-zero model family that minimizes the curvature regularized score loss.

We find \Cref{alg:langevin_single_obs} and the two alternatives have similar empirical performance in our examples, but Alternative 2 incurs higher computational costs.  We present \Cref{alg:langevin_single_obs} in the main text because its regularization term is more natural and straightforward.

\section{Proofs and Additional Theoretical Results}\label{sec:additional_theory}

\subsection{Proof of \Cref{thm:score_matching}}\label{pf:score_matching}

We restate  the proof from \citet{hyvarinen2005estimation,hyvarinen2007some} here to keep our content self-contained and to better motivate the discussions in \Cref{sec:boundary_sol}.

We first formally state the assumptions required for \Cref{thm:score_matching}.

For any $\Xn$ in its marginal support, we define the
section of $\theta$ as
$\Omega_{\theta \mid \Xn} := \big\{\theta \in \Theta : \pi(\theta)\,\pn_\theta(\Xn) > 0\big\}$,
and denote its boundary by $\partial \Omega_{\theta \mid \Xn}$.

\begin{assumption}[Boundary Condition]\label{ass:boundary} For any $\Xn \in \mX$ and score network parameter $\phi$, it holds that
$\pi(\theta)\,\pn_\theta(\Xn)\, s_\phi(\theta,\Xn) \to 0$
as $\theta$ approaches $\partial \Omega_{\theta \mid \Xn}$.
\end{assumption}

\begin{remark}[Boundary Condition]
Many simulation-based models violate the boundary condition required in \Cref{ass:boundary}. For example, for the benchmark M/G/1-queuing model in \Cref{sec:queuing_simu}, it violates the boundary condition because (1) the uniform prior has non-diminishing density near the boundary, and (2) the support of $\theta_1$ depends on the data. We introduce two solutions to this issue in \Cref{sec:boundary_sol} and illustrate the details of the treatments on the queuing model in \Cref{sec:queuing_details}.
\end{remark}

\begin{assumption}[Finite Moments]\label{ass:finite_moment}
     For any $\phi$, $\E_{(\theta, \Xn)\sim \pi(\theta)\,\pn_\theta(\Xn)} \big[\|\nabla_\theta\log \pn_\theta(\Xn)\|^2\big]$ and $\E_{(\theta, \Xn)\sim \pi(\theta)\,\pn_\theta(\Xn)} \big[\|s_\phi(\theta, \Xn)\|^2\big]$ are both finite.
\end{assumption}

\proof Now we are stating the complete proof. For notational simplicity, we denote the data as $X$ and we write the training distribution of $\theta$ as $p(\theta)$, which can be either the prior distribution $\pi(\theta)$ or any proposal distribution $q(\theta)$. We first rewrite the score-matching objective as
\begin{align*}
&\E_{(\theta,X)\sim p(\theta)p(X\mid \theta)} \norm{s_\phi(\theta, X) - \nabla_\theta \log p(X\mid \theta)}^2 \\
=& \E_{p(\theta,X)} \norm{s_\phi(\theta, X)}^2 + \E_{p(\theta,X)}\norm{\nabla_\theta \log p(X\mid\theta)}^2 - 2\E_{p(\theta,X)} \Big[ s_\phi(\theta, X)^T\nabla_\theta \log p( X\mid\theta)\Big].
\end{align*}

Here the first two terms are finite uder \Cref{ass:finite_moment} and the last term is also finite due to the Cauchy-Schwarz inequality. Additionally, the second term $\E_{p(\theta,X)} \norm{\nabla_\theta \log p_\theta(X)}^2$ is a constant in $\phi$ can thus can be ignored in the optimization program. The first term does not depend on unknown quantity $p(X\mid\theta)$, so we only need to address the last term.

Denote the joint support of $(\theta, X)$ as $\Omega:=\{(\theta, X)\in \Theta\times\mX: p(\theta)p(X\mid\theta)>0\}$. We denote $\theta_{-j} = (\theta_1, \ldots, \theta_{j-1}, \theta_{j+1}, \theta_{d_\theta})^T$  and the marginal support of $(\theta_{-j},X)$ as $\Omega_{(\theta_{-j}, X)} :=\{ (\theta_{-j}, X): (\theta, X)\in \Omega \text{ for some } \theta_j\}$. We denote the boundary segments orthogonal to the $j$-th axis at $(\theta_{-j},X)$ as $\text{Sec}(\Omega; \theta_{-j},X):=\{\theta_j \in \R: (\theta, X) \in \Omega\}$.

\begin{align*}
&\E_{p(\theta,X)} \Big[ s_\phi(\theta, X)^T\nabla_\theta \log p( X\mid\theta)\Big]\\
=&\int_{\mX} \diff X \int_{\Omega(X)}  p(\theta )p(X\mid \theta) s_{\phi}(\theta, X)^T\nabla_\theta \log p( X\mid\theta) \diff \theta \\
=& \int_{\mX} \diff X \int_{\Omega(X)} p(\theta) \sum_{j=1}^{d_\theta} s_{\phi,j}(\theta, X) \nabla_{\theta_j} p(X\mid \theta) \diff\theta \\
= &\sum_{j=1}^{d_\theta} \int_{\Omega_{(\theta_j, X)}}\diff X \diff \theta_{-j} \int_{\text{Sec}(\Omega; \theta_{-j},X)} p(\theta)s_{\phi,j} (\theta, X)\frac{\partial p(X\mid \theta)}{\partial \theta_j}\diff\theta_j
\end{align*}

For each coordinate $j$ and the inside integral, assuming $\text{Sec}(\Omega; X, \theta_{-j})$ is an interval and denote it as $(a_j,b_j)$, we have
\begin{align*}
&\int_{\text{Sec}(\Omega; \theta_{-j},X)} p(\theta)s_{\phi,j} (\theta, X)\frac{\partial p(X\mid \theta)}{\partial \theta_j}\diff\theta_j \\
=& p(\theta)s_{\phi,j}(\theta, X) p(X\mid \theta) \big|_{a_j}^{b_j}- \int_{\text{Sec}(\Omega; \theta_{-j},X)} \frac{\partial p(\theta) s_{\phi, j}(\theta, X)}{\partial\theta_j}p(X\mid \theta)\diff\theta_j \\
=& -\int_{\text{Sec}(\Omega; \theta_{-j},X)} \Big[\frac{\partial p(\theta)}{\partial\theta_j} s_{\phi, j}(\theta, X)+p(\theta)\frac{\partial  s_{\phi, j}(\theta, X)}{\partial\theta_j}\Big]p(X\mid \theta)\diff\theta_j \qquad \text{(by \Cref{ass:boundary})} \\
=& - \int_{\text{Sec}(\Omega; \theta_{-j},X)} \Big[\frac{\partial \log p(\theta)}{\partial \theta_j}  s_{\phi, j}(\theta, X)+\frac{\partial  s_{\phi, j}(\theta, X)}{\partial\theta_j}\Big]p(\theta)p(X\mid \theta)\diff\theta_j
\end{align*}
This concludes our proof.

\subsection{Convergence analysis on localization scheme}\label{sec:rate_loc}
Recall from \Cref{sec:preconditioning} that we assumed the simulation process for
$\Xn^\theta$ can be represented as a deterministic map $\tau(\theta, \Zm)$ applied
to latent variables $\Zm = (Z_1,\ldots,Z_m)$ drawn i.i.d.\ from a known distribution
$P_Z$. For generality, here we allow the simulated datasets used in the localization
step to have size $m$, which may differ from the size $n$ of the observed dataset
$\Xn^\ast$.

Next, we list our assumptions. Our first assumption ensures that closeness of
parametric distributions implies closeness of the corresponding parameters.

\begin{assumption}[Lipschitz Continuity]\label{ass:lipschitz} Assume that there exists some constant $C>0$, such that for any $\theta_1, \theta_2 \in \Theta$, we have $\norm{\theta_1-\theta_2}\leq C\cdot d_\text{SW}(P_{\theta_1}, P_{\theta_2})$.
\end{assumption}
\Cref{ass:lipschitz} is equivalently stating that the inverse mapping
$P_\theta \mapsto \theta$ must be Lipschitz continuous. This condition holds for
many parametric models under standard regularity assumptions. For example,
in a location family with $\tau(\theta,\Zm)=\theta+\Zm$, we have
$d_\text{SW}(P_{\theta_1}, P_{\theta_2}) = c_p\norm{\theta_1-\theta_2}$, where $c_p$ is the expected norm of the unit projection vector in $\R^p$.
For  exponential families with density
$p_\theta(X) = h(X)\exp\!\big(\langle \theta, T(X) \rangle - A(\theta)\big)$,
Assumption~\ref{ass:lipschitz} holds provided that $T(X)$ is injective and smooth
and that $A(\theta)$ is strongly convex.

Next, we need a condition ensuring that empirical distributions are close to
their population counterparts. Denote the two empirical distributions by
$\PP_n = \tfrac{1}{n}\sum_{i=1}^n \delta_{X_{0,i}}$ and
$\Q_m^\theta = \tfrac{1}{m}\sum_{j=1}^m \delta_{\tau(\theta, Z_j)}$,
where $\delta_x$ denotes the Dirac measure at $x$.

\begin{assumption}[Uniform Convergence]\label{ass:unif_convergence} Assume for $\Zm \sim P_Z^{(m)}$, we have
    $\sup_{\theta\in \Theta} d_\text{SW}\big(\Q^\theta_m, P_\theta\big) = O_p(m^{-\frac{1}{2}})$.
\end{assumption}

In many SBI applications, the prior on $\theta$ is chosen to be uniform, which
implies a compact parameter space $\Theta$. In this setting,
\Cref{ass:unif_convergence} generally holds when the simulator
$\tau(\cdot,\cdot)$ is jointly Lipschitz in $(\theta, Z)$ and the induced data distributions are sub-Gaussian. A similar result regarding $\E[d_{SW}(\Q_{m}^\theta, P_\theta)]$ was studied in \citet{nietert2022statistical}.

\begin{theorem}[Convergence rate in localization scheme]\label{thm:precond_convergence} Under \Cref{ass:lipschitz,ass:unif_convergence}, for any $\widehat\theta_{m,n} \in \arg\min_\theta d_\text{SW}\big(\tau(\theta, \Zm), \Xn^\ast\big)$, we have $\|\widehat\theta_{m,n}-\theta^\ast\|=O_p(n^{-\frac{1}{2}}+m^{-\frac{1}{2}})$.
\end{theorem}

\noindent A proof of this result is provided in \Cref{sec:precond_convergence_pf}.
\Cref{thm:precond_convergence} shows that setting $m = \mathcal{O}(n)$ allows
the localization step to restrict the search to a neighborhood of $\theta^\ast$
with radius $\mathcal{O}(n^{-1/2})$.

While these results are stated under SWD, a similar
$\sqrt{n}$ convergence rate can also be shown for the
max-sliced Wasserstein distance \citep{deshpande2019max} and Maximum Mean
Discrepancy (MMD) \citep{gretton2012kernel} under similar assumptions. However, using the 1-Wasserstein distance yields a slower rate of
$\mathcal{O}_p(n^{-1/p} + m^{-1/p})$, due to the curse of dimensionality in the
convergence rate $m^{-1/p}$ of the empirical distribution to its population counterpart \citep{boissard2011simple}. An additional advantage of SWD is its computational scalability with sample sizes
$m$ and $n$, as well as its robustness to increasing dimensionality.

\subsection{Conditions for \Cref{ass:unif_convergence}}\label{sec:unif_convergence_cond}
\begin{lemma}\label{lem:unif_convergence}Assume $\Theta \subset \R^{d_\theta}$ is compact and the simulator $\tau(\cdot, \cdot)$ is jointly Lipschitz in both $\theta$ and $Z$ such that
    \[
   \norm{\tau(\theta_1, Z)- \tau(\theta_2,Z)}\leq L(Z)\norm{\theta_1-\theta_2}, \quad L_*:=\E(L(Z))<\infty.
    \]
    Additionally, we assume $X=\tau(\theta, Z)$ is subgaussian for any $\theta \in \Theta$, then we have
    \[
        \sup_{\theta\in \Theta} d_\text{SW}\big(\Q^\theta_m, P_\theta\big) = O_p(m^{-\frac{1}{2}}).
        \]
\end{lemma}

\begin{proof}
We denote the 1-Wasserstein distance as $d_{W1}$ and recall the relationship between the 1-Wasserstein distance and the sliced Wasserstein distance as
\[
d_{SW}(\Q^{\theta}_m,P_\theta) = \int_{\mathcal{S}^{p-1}} d_{W1}(\Q^{\theta}_{m,\omega}, P_{\theta,\omega}) \diff\sigma(\omega)
\]
where $\omega\in \mathcal{S}^{p-1}:=\{\omega' \in \R^{p-1}:\norm{\omega'}\leq 1\}$ is a projection direction, $\sigma(\cdot)$ is the uniform measure on the unit sphere, and $\Q^{\theta}_{m,\omega}$ and $P_{\theta,\omega}$ are the projections of $\Q^{\theta}_m$ and $P_\theta$ onto the direction $\omega$ by $x\mapsto \omega^Tx$.

First, since $\tau(\theta, Z)$ is subgaussian, its projected variable $\omega^T\tau(\theta,Z)$ is also subgaussian.  For any fixed $(\theta, u)$, from \citet[Theorem 2]{fournier2015rate} (plugging in $p=d=1$, and condition (1) is satisfied with $\alpha=2$), we have
\begin{equation}\label{eq:W1_univariate_concentrate}
P(d_{W1}(\Q^{\theta}_{m,\omega}, P_{\theta,\omega})>t)\leq c_1 \exp(-c_2{mt^2})
\end{equation}
for some constant $c_1, c_2>0$. This leads to $\E(d_{W1}(\Q^{\theta}_{m,\omega}, P_{\theta,\omega}))\leq C_1 m^{-1/2}$ with some constant $C_1>0$.

Since all W1 distances are non-negative, we have
\[
d_{SW}(\Q^{\theta}_m,P_\theta)\leq \sup_{\omega\in \mathcal{S}^{p-1}} d_{W1}(\Q^{\theta}_{m,\omega}, P_{\theta,\omega}).
\]

Since $\Theta\subset \R^{d_\theta}$ is compact, we refer its euclidean radius as $R$. We can cover $\Theta$ with $N_\varepsilon\leq (R/\varepsilon)^{d_\theta}$ balls of radius $\varepsilon$ such that for any $\theta\in \Theta$, there exists a ball $B_{d_\theta}(\theta^i, \varepsilon)$ such that $\theta\in B_{d_\theta}(\theta^i, \varepsilon) =\{\theta':\norm{\theta'-\theta^i}\leq\varepsilon\}$. Similarly, since $\omega\in \mathcal{S}^{p-1}$, we can cover $\mathcal{S}^{p-1}$ with  $M_{\gamma}\leq (1/\gamma)^{p-1}$ balls of radius $\Gamma$ such that for any $\omega \in \mathcal{S}^{p-1}$, there exists a ball $B_{p-1}(\omega_j,\gamma)=\{\omega': \norm{\omega'-\omega_i}\leq \gamma\}$ such that $\omega\in B_p(\omega_j,\gamma)$.

Let $\Delta_m(\theta, \omega):= d_{W1}(\Q^{\theta}_{m,\omega}, P_{\theta,\omega})$. For any $\theta \in B_{d_\theta}(\theta^i, \varepsilon)$ and any $\omega \in B_{p-1}(\omega_j, \gamma)$, we have
\begin{align*}
    \abs{\Delta_m(\theta, \omega)-\Delta_m(\theta^i, \omega_j)}&\leq \abs{\Delta_m(\theta, \omega)-\Delta_m(\theta, \omega_j)}+\abs{\Delta_m(\theta, \omega_j)-\Delta_m(\theta^i,\omega_j)}\\
    & \leq \abs{\Delta_m(\theta, \omega)-\Delta_m(\theta, \omega_j)}+ L_*\varepsilon
\end{align*}

Using the four-point form of the triangle inequality, we have
\begin{align*}
    \abs{\Delta_m(\theta, \omega)-\Delta_m(\theta, \omega_j)} & \leq W_1(\Q^{\theta}_{m,\omega}, \Q^{\theta}_{m,\omega_j}) + W_1(P_{\theta,\omega}, P_{\theta,\omega_j}) \\
    & \leq \frac{1}{m}\sum_{i=1}^m\norm{\tau(\theta, Z_i)}\norm{\omega-\omega_j} + \E_{X\sim P_\theta}\norm{X}\norm{\omega-\omega_j} \leq 2C_X\gamma.
\end{align*}

Thus, we can rewrite the supremum  as
\begin{align*}
\sup_{\theta\in \Theta}\sup_{\omega\in \mathcal{S}^{p-1}} \Delta_m(\theta, \omega)  &\leq \max_{i}\max_j \Delta_m(\theta^i, \omega_j) + 2C_X\gamma+L_*\varepsilon.
\end{align*}

Combing the above with the inequality in \eqref{eq:W1_univariate_concentrate}, we have the union bound as
\begin{equation}\label{eq:DKW_union}
    P(\max_{i,j} \Delta_m(\theta^i, \omega_j))\leq c_1 N_\varepsilon M_\gamma \exp(-c_2mt^2).
\end{equation}

Setting $t=\kappa m^{-1/2}\sqrt{\log m}, \varepsilon =t/(3L_*), \gamma =t/(3C_X)$, we have
\begin{align*}
    N_\varepsilon M_\gamma \exp(-2mt^2)&\leq (3RL_*)^{d_\theta}(3C_x)^{p-1} t^{-(d_\theta+p-1)}\exp(-c_2mt^2)\\
    &\leq C_2 m^{(d_\theta+p-1)/2}(\log m)^{-\frac{d_\theta+p-1}{2}}m^{-c_2\kappa^2}
\end{align*}
for some constant $C_2\geq (3RL_*)^{d_\theta}(3C_x)^{p-1}$. For fixed $d_\theta, p$, we can choose $\kappa$ large enough such that $c_2\kappa^2 = \beta+ (d_\theta+p-1)/2$ with $\beta>0$, then
\[
    P(\sup_{\theta\in \Theta}\sup_{\omega\in \mathcal{S}^{p-1}} \Delta_m(\theta, \omega)>t)\leq C_2 m^{-\beta}  (\log m)^{-(d_\theta+p-1)/2}
\]
Thus we have $\sup_{\theta\in \Theta}\sup_{\omega\in \mathcal{S}^{p-1}} d_{W1}(\Q^{\theta}_{m,\omega}, P_{\theta,\omega})= O_p(m^{-1/2}\sqrt{\log m})$.

We can further refine the bound by using the generic-chaining bound \citep{dudley1967sizes}, as
\begin{align*}
\E \big[\sup_{\theta,\omega}\Delta_m(\theta, \omega)  \big] &\leq  \frac{C_3}{\sqrt m} \int_{0}^1 \sqrt{\log N_{\Theta\times\mathcal{S}^{p-1}}(\varepsilon)}\diff\varepsilon
\end{align*}
where $N_{\Theta\times\mathcal{S}^{p-1}}(\varepsilon)$ is the covering number of the joint space of $\Theta \times\mathcal{S}^{p-1}$ with balls of radius $\varepsilon$, and we can bound it as $N_{\Theta\times\mathcal{S}^{p-1}}(\varepsilon)\leq (R/\varepsilon)^{d_\theta}(1/\varepsilon)^{p-1}\leq C_4\varepsilon^{-(d_\theta+p-1)}$. Plugging this number into the inequality above, we have
\[
\E \big[\sup_{\theta,\omega}\Delta_m(\theta, \omega)  \big] \leq \frac{C_3'\sqrt{d_\theta+p-1}}{\sqrt m}
\]
where $C_3'$ is again a constant depending on $C_3$ and $C_4$.

Using McDiarmid's inequality on the subgaussian variables and the fact that $\Delta_m$ is Lipschitz in Z, we have
\[
P\left(\abs{\sup_{\theta, \omega}\Delta_m(\theta, \omega) - E\sup_{\theta, \omega}\Delta_m(\theta, \omega)} >t\right)\leq 2 \exp(-C_5 t^2/(m\sigma^2))
\]
where $\sigma:=\sup_\theta \E \norm{{\tau(\theta, Z)}}_{\psi_2}$ is from the subgaussian assumption, and $C_5$ is another constant.

Taking $t=\sigma m^{-1/2}$, we have
\[
\sup_{\theta,\omega}\Delta_m(\theta, \omega) =O_p(m^{-1/2}).
\]

\end{proof}

\subsection{Regularity Assumpions for \Cref{thm:post_convergence_iid}}\label{sec:regularity_conditions_thm}

Here we list the regularity assumptions required for the convergence of the approximated posterior distribution in the i.i.d.\ setting in \Cref{thm:post_convergence_iid}.

Denote the maximum likelihood estimator (MLE) by $\widehat{\theta}^{\text{\rm MLE}}_n := \arg\max_\theta \pn_\theta(\Xn^\ast)$. The next assumption requires that both the true posterior $\pi_n$
and the MLE $\widehat{\theta}^{\text{\rm MLE}}_n$ concentrate in a neighborhood of $\theta^\ast$, a condition that is standard in the literature on Bayesian and frequentist large-sample theory for parametric models (see, e.g., \cite{ghosal2000convergence,spokoiny2012parametric}).

\begin{assumption}[Concentration of the true posterior and MLE]\label{ass:true_converge} There exists some constant $C_1>0$, such that for every $t>\sqrt{\frac{\log n}{n}}$, we have
\begin{align*}
    \E_{\Pn_{\theta^\ast}} \Pi_n\big( \norm{\theta-\theta^\ast} >t \mid \Xn^\ast \big) \ &\leq\   \exp(-{C_1 n t^2}), \\
    \Pn_{\theta^\ast} \big( \|\widehat{\theta}^{\text{\rm MLE}}_n-\theta^\ast\| >t \mid \Xn^\ast \big) \ &\leq \  \exp(-{C_1 n t^2}).
\end{align*}
\end{assumption}

We now introduce two assumptions that are essential for the convergence of LMC under the true score and for controlling this discretization error.

\begin{assumption}[True Score Lipschitz continuity]\label{ass:lipschitz_score}The true likelihood score $s^\ast$ is uniformly $\lambda_L$-Lipschitz in $\theta$ over $\R^{d_\theta}$. That is, for every $x\in \mX $ and every $\theta_1, \theta_2\in \Theta$, we have
    $
    \norm{s^\ast(\theta_1, x) -s^\ast(\theta_2, x)}\leq \lambda_L\norm{\theta_1-\theta_2}.
    $
\end{assumption}
\Cref{ass:lipschitz_score} is essential for guaranteeing that the drift
$s^\ast(\theta, X_i)$ grows at most linearly, ensuring that the
Langevin diffusion admits a unique stationary measure
\citep{chewi2024analysis,lee2022convergence}. Our proof for controlling the score error also requires this assumption in order to carry out a perturbation analysis.

\begin{assumption}[Log-Sobolev inequality]\label{ass:log_sobolev} The posterior distribution of $\sqrt n(\theta-\theta^\ast)$ satisfies a log-Sobolev inequality with constant $C_\text{LSI}$, i.e., for each function $f\in C_0^\infty(\R^{d_\theta})$,
    we have $\text{Ent}(f^2)\,\leq \,2C_\text{LSI}\, \E_{\alpha := \sqrt n (\theta- \theta^\ast)}[\norm{\nabla_\alpha f}^2]$,
    where the entropy is defined as $\text{Ent}(g)=\E[g\log g]-\E[g]\log\E[g]$.
\end{assumption}
\noindent Here we impose the condition on the transformed variable $\alpha:=\sqrt n(\theta-\theta^*)$, since \Cref{ass:true_converge} suggests that the distribution of $\alpha$ is non-degenerate while $\theta$ concentrates to $\theta^*$.
\Cref{ass:log_sobolev} is a commonly used assumption to ensure the convergence
of LMC \citep{chewi2024analysis,lee2022convergence}. In
the Bayesian setting, this assumption is mild when $n$ is large, since the
posterior is approximately Gaussian by the Bernstein-von Mises (BvM) theorem,
and the log-Sobolev inequality is then immediately satisfied
\citep{nickl2022polynomial,tang2024computational}.

Next we make some regularity assumptions on the true score and also on the estimated score. Similar regularity conditions on score functions are standard in the literature to guarantee the asymptotic normality of the MLE and the posterior; see, for example, \cite{ghosh2003bayesian,vanderVaart1998}.

\begin{assumption}\label{ass:bdd_derivatives}
There exist some constant $C_5 > 0$ and $\delta>0$, such that
\[\E_{X\sim P_{\theta^\ast}}\big[\norm{s^\ast(\theta^\ast, X)}^2\big], \E_{X\sim P_{\theta^\ast}}\big[\norm{\widehat{s}(\theta^\ast, X)}^2\big], \E_{X\sim P_{\theta^\ast}}\big[ \lVert \nabla_\theta s^{\ast}(\theta^\ast, X)\rVert_F^2\big],  \; \E_{X\sim P_{\theta^\ast}} \big[\lVert \nabla_\theta \widehat{s}(\theta^\ast, X)\rVert_F^2\big] \]
are all finite and bounded by $C_5^2$, and for each $x$,
\[
\sup_{\theta:\,\|\theta-\theta^\ast\|\leq \delta} \bigg\{\sum_{j=1}^{d_\theta} \big\lVert \nabla_\theta^2 s_j^\ast(\theta, x) \big\rVert_F^2 \bigg\}\leq M(x),\ \ \mbox{and}\ \ \sup_{\theta:\,\|\theta-\theta^\ast\|\leq \delta}\bigg\{\sum_{j=1}^{d_\theta}\big\lVert \nabla_\theta^2 \widehat{s}_j(\theta, x) \big\rVert_F^2 \bigg\} \leq M(x).
\]
where $s_j^\ast(\theta, X)$ denotes the j-th coordinate of $s^\ast(\theta, X)$, and the function $M(\cdot)$ satisfies $\E_{X\sim P_{\theta^\ast}}[M(X)]\leq C_5^2$. Additionally, we also assume $\norm{\widehat s(\theta, \Xn^\ast)}_2\leq C_3n\,(1+\lVert \theta - \theta^\ast\rVert_2)$ and $\sup_{\theta \in \mA_{n,1}}\norm{\widehat s(\theta, \Xn^\ast)}_2\leq C_3\sqrt{n\log n}$ holds with probability at least $1-n^{-1}$. Here, the set $\mA_{n,1}$ is defined in Assumption~\ref{ass:uniform_sm_err_single}.
\end{assumption}
For the last part of \Cref{ass:bdd_derivatives},
note that the Bernstein--von Mises theorem suggests that
the true score $s^\ast(\theta, \Xn^\ast)$ is of order $O_p(\sqrt{n \log n})$
within the neighborhood $\mA_{n,1}$.
This motivates us to assume that the estimated score
$\widehat s(\theta, \Xn^\ast)$ satisfies a similar bound.
Otherwise, one can always clip $\widehat s(\theta, \Xn^\ast)$ during the sampling process, which corresponds to a projection operator that never increases
the score matching error.

\subsection{Proof of \Cref{thm:precond_convergence}}\label{sec:precond_convergence_pf}

We refer the latent variables corresponding to the observed data $\Xn^\ast$ as $\Zn^\ast$, such that $\Xn^\ast= \tau(\theta^\ast, \Zn^\ast)$, then we can write
\begin{equation}\label{eq:theta_matching}
    \arg\min_\theta d_\text{SW}\big(\tau(\theta, \Zm), \Xn^\ast\big)= \arg\min_\theta d_\text{SW}\big(\tau(\theta, \Zm), \tau(\theta^\ast, \Zn^0)\big).
\end{equation}

Although the solution to \eqref{eq:theta_matching} might not be unique, we can show that for any solution $\widehat\theta_{m,n} \in \Big\{ \theta: \arg\min_\theta d_\text{SW}\big(\tau(\theta, \Zm), \tau(\theta^\ast, \Zn^0)\big)\Big\}$, using the triangle inequality, we have
\begin{align*}
d_\text{SW}\big(\tau(\widehat\theta_{m,n}, \Zm), \tau(\theta^\ast, \Zn^0)\big) &\leq d_\text{SW}\big(\tau(\theta^\ast, \Zm),\tau(\theta^\ast, \Zn^0)\big)\\
&\leq d_\text{SW}(P_{\theta^\ast}, \Q_m^{\theta^\ast}) + d_\text{SW}(P_{\theta^\ast}, \PP_n) = O_p(n^{-\frac{1}{2}}+m^{-\frac{1}{2}}).
\end{align*}

Furthermore, we show the distance between the two distributions $d(P_{\theta^\ast}, P_{\widehat\theta_{m,n}})$ is bounded by the distance between the two datasets $d(\tau(\widehat\theta, \Zm), \tau(\theta^\ast, \Zn^0))$.
\begin{align*}
    d_\text{SW}(P_{\widehat\theta_{m,n}}, P_{\theta^\ast}) &\leq  d_\text{SW}(P_{\widehat\theta_{m,n}},\Q^{\widehat\theta_{m,n}}_m)+ d_\text{SW}(P_{\theta^\ast}, \PP_n) +d_\text{SW}\big(\Q^{\widehat\theta_{m,n}}_m, \PP_n\big)\\
    &\leq \sup_{\theta\in \Theta} d_\text{SW}(P_\theta, \Q^\theta_m)+d_\text{SW}(P_{\theta^\ast}, \PP_n) +d_\text{SW}\big(\tau(\widehat\theta_{m,n}, \Zm), \tau(\theta^\ast, \Zn^0)\big)\\
    &= O_p(m^{-\frac{1}{2}})+O_p(n^{-\frac{1}{2}})+O_p(m^{-\frac{1}{2}}+n^{-\frac{1}{2}})= O_p(m^{-\frac{1}{2}}+n^{-\frac{1}{2}}).
\end{align*}
Combing the inequality above and \Cref{ass:lipschitz}, we
\[
\norm{\theta_1-\theta_2} = O_p(m^{-\frac{1}{2}}+n^{-\frac{1}{2}}).
\]

\subsection{Proof of \Cref{thm:post_convergence_n_obs}}\label{sec:post_convergence_n_obs_pf}

The proof consists of two parts. We first show how we control the discretization error, which provides theoretical guidance on how we should choose the step size $\tau_n$ and the initial distribution. Later we focus on analyzing the score error.

 Here we introduce the local variable $\alpha= \sqrt{n} (\theta-\theta^\ast)$, then its scores satisfy $s^\ast_\alpha(\alpha,x)= \frac{1}{\sqrt n}s^\ast(\theta, x)$ and $s_{\widehat\phi,\alpha}(\alpha, x) =\frac{1}{\sqrt n}s_{\widehat\phi} (\theta, x)$. It is easier to work with $\alpha$ since it has constant independent of $n$. We refer all transformed densities and functions under $\alpha$ with subscript $\alpha$, such as $\tau_\alpha, \pi_\alpha$. We  rewrite the Langevin Monte Carlo update as
\[
\alpha^{(k)} = \alpha^{(k-1)} +\tau_\alpha  \big( s_{\widehat\phi, \alpha}(\alpha^{(k-1)}, \Xn^\ast)+ \nabla_\alpha\log \pi_\alpha(\alpha) \big)+\sqrt{2\tau_\alpha} U_k
\]
with $\tau_\alpha= n\tau_n$. For notational simplicity, we write $\widehat s:=s_{\widehat\phi}$ and $\widehat s_\alpha:= s_{\widehat\phi, \alpha}$.

{\bf Part 1: The discretization error and burn-in error.} Since total variation distance is invariant under any bijective transformation, we can rewrite the discretiztion error as
\begin{align*}
    d_\text{TV}(\pi^{K\tau_n} (\theta\mid \Xn^\ast), \pi_n(\theta\mid\Xn^\ast))=d_\text{TV}(\pi^{K\tau_\alpha}(\alpha\mid\Xn^\ast), \pi_n(\alpha\mid \Xn^\ast))
\end{align*}

Using results from \citet[Lemmas D.1 and H.4]{ding2024nonlinear}, we have, at the terminal time $T=K\tau_\alpha$,
\begin{align*}
d^2_{TV}\big(\pi^{T}(\alpha\mid \Xn^\ast), \pi_n(\alpha\mid \Xn^\ast) \big) & \leq \frac{1}{4}d_{\chi^2}\big(\pi^T(\alpha\mid \Xn^\ast),\pi_n(\alpha\mid \Xn^\ast)\big)\\
&\leq \frac{1}{4}\exp\big(-\frac{T}{5C_\text{LSI}}\big)\eta_\chi^2 + 35 d_\theta C_\text{LSI} \lambda_L^2 \tau_\alpha
\end{align*}
where the step size $\tau_\alpha$  and the initial distribution $\pi^0(\cdot\mid \Xn^\ast )$ satisfy
\[
400 d_\theta C_\text{LSI}\lambda_L^2 \tau_\alpha\leq 1, \qquad d_{\chi^2}(\pi^0(\alpha\mid \Xn^\ast),\pi_n(\cdot\mid \Xn^\ast) )\leq \eta_\chi^2.
\]

This suggests that the step size should be $\tau_n =\tau_\alpha/n = O(n^{-1})$.

{\bf Part 2: The score error.}

 The score error is induced by using the inexact score $s_{\widehat\phi}$. We first break down the score error as summation of score-matching errors at each Langevin update, by applying the Girsanov theorem \citep{chen2023sampling}
\begin{align}
&d_\text{TV}^2\big(\widehat \pi^{K\tau_n}(\theta\mid \Xn^\ast), \pi^{K\tau_n}(\theta\mid \Xn^\ast)\big)\\
 &= d_\text{TV}^2\big(\pi^{K\tau_\alpha}(\alpha\mid \Xn^\ast),  \widehat \pi^{K\tau_\alpha}(\alpha \mid \Xn^\ast)\big) \nonumber\\
 &\leq \frac{1}{2}d_{KL} (\pi^{K\tau_\alpha}(\alpha\mid \Xn^\ast) ||  \widehat \pi^{K\tau_\alpha}(\alpha \mid \Xn^\ast)) \quad \text{(Pinsker's inequality)} \nonumber\\
 &\leq  \frac{1}{8} \sum_{k=0}^{K-1} \tau_\alpha \E_{\alpha\sim \pi^{k\tau_\alpha}\big(\alpha\mid \Xn^\ast\big)}\norm{\widehat s_\alpha(\alpha, \Xn^\ast)- s^\ast_\alpha(\alpha, \Xn^\ast)}^2.\label{eq:girsanov}
\end{align}

Using the Cauchy-Schwarz inequality, we further bound each term in the  summation in \eqref{eq:girsanov} as
\begin{align*}
&\E_{\alpha\sim \pi^{k\tau_\alpha}(\alpha\mid \Xn^\ast)}\norm{\widehat s_\alpha(\alpha, \Xn^\ast)- s^\ast_\alpha(\alpha, \Xn^\ast)}^2  \\
=& \E_{\alpha\sim \pi_n(\alpha\mid \Xn^\ast)} \Big[ \norm{\widehat s_\alpha(\alpha, \Xn^\ast)- s^\ast_\alpha(\alpha, \Xn^\ast)}^2 \frac{\pi^{k\tau_\alpha}(\alpha\mid \Xn^\ast)}{\pi_n(\alpha\mid \Xn^\ast)} \Big]\\
= & \sqrt{\E_{\alpha\sim \pi_n(\alpha\mid \Xn^\ast)} \norm{\widehat s_\alpha(\alpha, \Xn^\ast)- s^\ast_\alpha(\alpha, \Xn^\ast)}^4}\sqrt{\E_{\alpha\sim \pi_n(\alpha\mid \Xn^\ast)}\Big[ \frac{\pi^{k\tau_\alpha}(\alpha\mid \Xn^\ast)}{\pi_n(\alpha\mid \Xn^\ast)}\Big]^2}
\end{align*}

Following Lemma D.1 from \citet{ding2024nonlinear}, we can bound the second term as
\begin{align*}
\sqrt{\E_{\alpha\sim \pi_n(\alpha\mid \Xn^\ast)}\Big[ \frac{\pi^{k\tau_\alpha}(\alpha\mid \Xn^\ast)}{\pi_n(\alpha\mid \Xn^\ast)}\Big]^2} &= \sqrt{d_{\chi^2}\big({\pi^{k\tau_\alpha}(\alpha\mid \Xn^\ast)}||{\pi_n(\alpha\mid \Xn^\ast)}\big) +1} \\
& \leq \sqrt{\exp(-\frac{k\tau_\alpha}{5C_\text{LSI}})\eta_\chi^2+2}.
\end{align*}

For the first term, its expectation w.r.t. $\Xn^\ast$ can be bounded as
\begin{align*}
&\E_{\Xn^\ast \sim \Pn_{\theta^\ast}} \sqrt{\E_{\alpha\sim \pi_n(\alpha\mid \Xn^\ast)} \norm{\widehat s_\alpha(\alpha, \Xn^\ast)- s^\ast_\alpha(\alpha, \Xn^\ast)}^4} \\
=& \E_{\Xn^\ast \sim \Pn_{\theta^\ast}} \sqrt{\E_{\theta\sim \pi_n(\theta\mid \Xn^\ast)} \normbigg{\frac{1}{\sqrt n}\widehat s(\theta, \Xn^\ast)- \frac{1}{\sqrt n}s^\ast(\theta, \Xn^\ast)}^4} \\
\leq& \sqrt{\E_{(\theta, \Xn^\ast)\sim \pi_n(\theta\mid \Xn^\ast)\pn_{\theta^\ast}(\Xn^\ast)} \normbigg{\frac{1}{\sqrt n}\widehat s(\theta, \Xn^\ast)- \frac{1}{\sqrt n}s^\ast(\theta, \Xn^\ast)}^4} \\
\lesssim& \sqrt{(\log n)^2 \varepsilon_{N,n,1}^2 + (\log n)^2 \varepsilon_{N,n,2}^2 + \frac{(\log n)^3}{n}}
\end{align*}
where the first step is due to the scale transform between $\alpha$ and $\theta$, the second step is by Jensen's inequality, and the last step is by \Cref{lem:score_err_fourth_moment}.

This leads to the final bound on the score error as
\begin{align*}
    & \E_{\Xn^\ast \sim \Pn_{\theta^\ast}} \Big[d_\text{TV}^2\big(\widehat \pi^{K\tau_n}(\theta\mid \Xn^\ast), \pi^{K\tau_n}(\theta\mid \Xn^\ast)\big) \Big]\\
    \lesssim  &\sqrt{(\log n)^2 \varepsilon_{N,n,1}^2 + (\log n)^2 \varepsilon_{N,n,2}^2 + \frac{(\log n)^3}{n}} \times \tau_\alpha \sum_{k=0}^{K-1}\sqrt{\exp(-\frac{k\tau_\alpha}{5C_\text{LSI}})\eta_\chi^2+2}\\
    &\lesssim \sqrt{(\log n)^2 \varepsilon_{N,n,1}^2 + (\log n)^2 \varepsilon_{N,n,2}^2 + \frac{(\log n)^3}{n}}  (K\tau_\alpha+ \eta_\chi C_\text{LSI}).
\end{align*}
where the bound of the summation in the last step is because $\sqrt{\exp(-\frac{k\tau_\alpha}{5C_\text{LSI}})\eta_\chi^2+2} \leq 2\exp(-\frac{k\tau_\alpha}{10C_\text{LSI}})\eta_\chi + 2\sqrt{2}$ for each $k$, and $\sum_{k=0}^{K-1}\exp(-\frac{k\tau_\alpha}{10C_\text{LSI}}) \leq \frac{20C_\text{LSI}}{3\tau_\alpha}$ under the chosen $\tau_\alpha$. Finally, adding the two sources of error concludes the proof of \Cref{thm:post_convergence_n_obs}.

\begin{lemma}\label{lem:score_err_fourth_moment}
Under the assumptions in \Cref{thm:post_convergence_n_obs}, we have
\begin{align*}
\E_{(\theta, \Xn^\ast)\sim \pi_n(\theta\mid \Xn^\ast)\pn_{\theta^\ast}(\Xn^\ast)} \norm{\frac{1}{\sqrt n}\widehat s(\theta, \Xn^\ast)- \frac{1}{\sqrt n}s^\ast(\theta, \Xn^\ast)}^4
\lesssim  (\log n)^2 \varepsilon_{N,n,1}^2 + (\log n)^2 \varepsilon_{N,n,2}^2 + \frac{(\log n)^3}{n}
\end{align*}
\end{lemma}
\begin{proof}
For notational simplicity, we use $\E_{\Pi_n \cdot \Pn_{\theta^\ast}}$ to denote $\E_{(\theta, \Xn^\ast)\sim \pi_n(\theta\mid \Xn^\ast)\pn_{\theta^\ast}(\Xn^\ast)}$ in this proof. Recall that $\mA_{n,1}:= \{\theta: \norm{\sqrt n(\theta- \theta^\ast)}_2 \leq C_0\sqrt{\log n}\}$ in \Cref{ass:uniform_sm_err}. We define an event $\mA_{n,2} := \Big\{\X_n^\ast: \sqrt n \Big\lVert \widehat{\theta}^{\text{\rm MLE}}_n - \theta^\ast \Big\rVert\leq C_0\sqrt{\log n}\Big\}$ and let $\mA_{n,3}:=\mA_{n,1} \cap \mA_{n,2}$. Now, we split the integral into two parts
\begin{align*}
&\E_{\Pi_n \cdot \Pn_{\theta^\ast}} \normbigg{\frac{1}{\sqrt n}\widehat s(\theta, \Xn^\ast)- \frac{1}{\sqrt n}s^\ast(\theta, \Xn^\ast)}^4 \\
=& \E_{\Pi_n \cdot \Pn_{\theta^\ast}} \bigg[\normbigg{\frac{1}{\sqrt n}\widehat s(\theta, \Xn^\ast)- \frac{1}{\sqrt n}s^\ast(\theta, \Xn^\ast)}^4 \, \I_{\mA_{n,3}}\bigg] &\text{(I)}\\
&+ \E_{\Pi_n \cdot \Pn_{\theta^\ast}} \bigg[\normbigg{\frac{1}{\sqrt n}\widehat s(\theta, \Xn^\ast)- \frac{1}{\sqrt n}s^\ast(\theta, \Xn^\ast)}^4\, \I_{\mA_{n,3}^C}\bigg] &\text{(II)}
\end{align*}
For term (I), we have $\normbig{\frac{1}{\sqrt n}\widehat s(\theta, \Xn^\ast)}\leq C_3 \sqrt{\log n}$ by \Cref{ass:uniform_sm_err} and $\big\lVert \frac{1}{\sqrt n} s^\ast(\theta, \Xn^\ast) \big\rVert \leq \lambda_L\big\lVert \sqrt n(\theta - \MLE) \big\rVert \leq 2\lambda_LC_0 \sqrt{\log n}$ on $\mA_{n,3}(\Xn^\ast)$ by \Cref{lem:true_sco_bd} and triangle inequality. Thus, we can bound the fourth moment by the product of the second moment and the sup norm as
\begin{align*}
  \text{(I)}\lesssim & \log n \E_{\Pi_n \cdot \Pn_{\theta^\ast}}\Big[\normbig{\frac{1}{\sqrt n}\widehat s(\theta, \Xn^\ast)- \frac{1}{\sqrt n}s^\ast(\theta, \Xn^\ast)}^2 \, \I_{\mA_{n,3}}\Big] \\
  \lesssim & \log n \E_{\Pn_{\theta^\ast}(\Xn^\ast)}\Big[ \normbig{\frac{1}{\sqrt n}\widehat s(\theta^\ast, \Xn^\ast)- \frac{1}{\sqrt n}s^\ast(\theta^\ast, \Xn^\ast)}^2 \Big] \qquad \text{(I.1)} \\
    &+ \log n \E_{\Pi_n \cdot \Pn_{\theta^\ast}}\Big[\normbig{\frac{1}{\sqrt n}(\widehat s-s^\ast)(\theta, \Xn^\ast)- \frac{1}{\sqrt n}(\widehat s-s^\ast)(\theta^\ast, \Xn^\ast)}^2 \I_{\mA_{n,3}}\Big]  \;\text{(I.2)}
\end{align*}
where the last step is due to triangle inequality. Term (I.1) is bounded by the uniform score-matching error $\varepsilon_{N,n,1}^2$. Term (I.2) can be simplified by applying Taylor's expansion for $\widehat s - s^\ast$ around $\theta^\ast$ as
\begin{equation}\label{eq:taylor_sdiff}
\begin{split}
\frac{1}{\sqrt n} (\widehat{s} - s^\ast)(\theta, \Xn^\ast)
=& \frac{1}{\sqrt n} (\widehat{s} - s^\ast)(\theta^\ast, \Xn^\ast) +  \frac{\nabla_\theta(\widehat{s} - s^\ast)(\theta^\ast, \Xn^\ast)}{n} \sqrt{n}(\theta - \theta^\ast) \\
&+ \frac{1}{\sqrt{n}}\cdot\frac{T(\theta^\ast, \theta, \Xn^\ast)}{2n} [\sqrt{n}(\theta - \theta^\ast),\sqrt{n}(\theta - \theta^\ast)],
\end{split}
\end{equation}
where $T(\theta^\ast, \theta, \Xn^\ast)$ is obtained by applying a second-order Taylor expansions to each coordinate of $(\hat s-s^*)$. Specifically, let $\nabla_{\theta,j} (\hat s-s^*):\R^{d_\theta\times d_\theta\times d_\theta} \mapsto \R^{d_\theta\times d_\theta}$ be the gradient with respect to $\theta$ on the $j$-th coordinate of $(\hat s-s^\ast)$, then
\[T(\theta^\ast, \theta, \Xn^\ast) = [\nabla_{\theta,1}^2 (\widehat{s} - s^\ast)(\theta_1', \Xn^\ast), \dots, \nabla_{\theta,{d_\theta}}^2 (\widehat{s} - s^\ast)(\theta_{d_\theta}', \Xn^\ast)] \in \R^{d_\theta \times d_\theta \times d_\theta}\] is a tensor collecting the corresponding Hessian matrices evaluated at $\theta_j' =c_j\theta+(1-c_j)\theta^\ast$ for some $c_j\in [0,1], j = 1,\cdots,d_\theta$.
For $z\in \R^d$ and tensor $A=[A_1, \dots, A_p] \in \R^{p\times d\times d}$ with $A_1,\dots,A_p \in \R^{d\times d}$, $A[z, z] \in \R^p$ is defined as
\[
A[z, z] := (z^T A_1 z, \dots, z^TA_pz)^T
\]
Next, we use \eqref{eq:taylor_sdiff} to show that the term in (I.2) is close to $\big(I(\theta^\ast)-\widehat I(\theta^\ast)\big)\sqrt n (\theta-\theta^\ast)$ in $L^2$, where $I(\theta) = \mathbb{E}_{X\sim P_{\theta}}[-\nabla_\theta s^\ast(\theta, X)]$ and $\widehat{I}(\theta) = \mathbb{E}_{X\sim P_{\theta}}[-\nabla_\theta \widehat{s}(\theta, X)]$.
{\small
\begin{equation}\label{eq:L2_close}
\begin{split}
& \E_{\Pi_n \cdot \Pn_{\theta^\ast}} \Bigg[ \bigg\lVert \frac{1}{\sqrt n} (\widehat{s} - s^\ast)(\theta, \Xn^\ast) - \frac{1}{\sqrt n} (\widehat{s} - s^\ast)(\theta^\ast, \Xn^\ast) - [I(\theta^\ast) - \widehat{I}(\theta^\ast)]\sqrt{n}(\theta - \theta^\ast) \bigg\rVert^2 \I_{\mA_{n,3}}\Bigg] \\
=& \E_{\Pi_n \cdot \Pn_{\theta^\ast}} \Bigg[\bigg\lVert  \bigg[\frac{\nabla_\theta(\widehat{s} - s^\ast)(\theta^\ast, \Xn^\ast)}{n} - I(\theta^\ast) + \widehat{I}(\theta^\ast)\bigg]\sqrt{n}(\theta - \theta^\ast) \\
&+ \frac{1}{\sqrt{n}}\cdot\frac{T(\theta^\ast, \theta, \Xn^\ast)}{2n} [\sqrt{n}(\theta - \theta^\ast), \sqrt{n}(\theta - \theta^\ast)] \bigg\rVert^2 \I_{\mA_{n,3}}\Bigg] \\
\lesssim& \log n \E_{\Pn_{\theta^\ast}} \Bigg[ \bigg\lVert \frac{\nabla_\theta(\widehat{s} - s^\ast)(\theta^\ast, \Xn^\ast)}{n} - I(\theta^\ast) + \widehat{I}(\theta^\ast)\bigg\rVert_2^2 \Bigg] \quad\text{(III)} \\
&+ 2\E_{\Pi_n \cdot \Pn_{\theta^\ast}} \Bigg[ \bigg\lVert\frac{1}{\sqrt{n}}\cdot\frac{T(\theta^\ast, \theta, \Xn^\ast)}{2n} [\sqrt{n}(\theta - \theta^\ast), \sqrt{n}(\theta - \theta^\ast)] \bigg\rVert^2  \I_{\mA_{n,3}} \Bigg] \quad\text{(IV)} \\
&\lesssim \frac{(\log n)^2}{n}
\end{split}
\end{equation}}
where the first step uses the expansion \eqref{eq:taylor_sdiff}, the second step is by triangle inequality and the fact that $\lVert \sqrt{n} (\theta - \theta^\ast) \rVert \le C_0\sqrt{\log n}$ on $\I_{\mA_{n,3}}$, and the last step is because:
\begin{align*}
\text{(III)} &\lesssim \log n \Bigg\{ \E_{\Pn_{\theta^\ast}} \Bigg[ \bigg\lVert \frac{\nabla_\theta\widehat{s}(\theta^\ast, \Xn^\ast)}{n} + \widehat{I}(\theta^\ast)\bigg\rVert_F^2 \Bigg] + \E_{\Pn_{\theta^\ast}} \Bigg[ \bigg\lVert \frac{\nabla_\theta s^\ast(\theta^\ast, \Xn^\ast)}{n} + I(\theta^\ast) \bigg\rVert_F^2 \Bigg] \Bigg\}\\
&\lesssim \frac{\log n}{n}\{ \E_{P_{\theta^\ast}} \lVert \nabla_\theta \widehat{s}(\theta^\ast, \Xn^{\ast})\rVert_F^2 + \E_{P_{\theta^\ast}} \lVert \nabla_\theta s^{\ast}(\theta^\ast, \Xn^{\ast})\rVert_F^2 \} \\
&\lesssim \frac{\log n}{n}  \qquad\qquad \text{(by \Cref{ass:bdd_derivatives})} \\
\text{(IV)} &\leq\frac{1}{2n} \E_{\Pi_n \cdot \Pn_{\theta^\ast}} \bigg[ \sum_{j=1}^{d_\theta} \lVert \sqrt{n}(\theta - \theta^\ast)  \rVert^4 \cdot \bigg\lVert \frac{\nabla_{\theta, j}^2 (\widehat{s} - s^\ast)(\theta_j', \Xn^{\ast})}{n} \bigg\rVert_F^2 \I_{\mA_{n,3}} \bigg] \\
&\lesssim \frac{(\log n)^2}{n} \qquad\qquad \text{(by \Cref{ass:bdd_derivatives})}
\end{align*}

Now, with \eqref{eq:L2_close} and triangle inequality, we can rewrite (I.2) as
\begin{align*}
    \text{(I.2)} &\lesssim \log n \bigg[\frac{(\log n)^2}{n} + \E\big[\big\lVert [I(\theta^\ast) - \widehat{I}(\theta^\ast)] \sqrt{n}(\theta - \theta^\ast) \big\rVert^2 \I_{\mA_{n,3}} \big] \bigg] \\
    &\lesssim (\log n)^2 \varepsilon_{N,n,1}^2 + (\log n)^2 \varepsilon_{N,n,2}^2 +  \frac{(\log n)^3}{n}
\end{align*}
where the second step is by \Cref{lem:curv_err_full} and the fact that $\lVert \sqrt{n}(\theta - \theta^\ast) \rVert \leq C_0\sqrt{\log n}$ on $\I_{\mA_{n,3}}$. Therefore, (I) is bounded by
\begin{equation*}
    \text{(I)} \lesssim \text{(I.1)} + \text{(I.2)} \lesssim  (\log n)^2 \varepsilon_{N,n,1}^2 +  (\log n)^2 \varepsilon_{N,n,2}^2 + \frac{(\log n)^3}{n} ,
\end{equation*}

Next, we continue the proof on the second part (II) on the complement set $\mA_{n,3}^c$.
Then, we can bound (II) as
\begin{align*}
    \text{(II)} =& \E_{\Pi_n \cdot \Pn_{\theta^\ast}}\bigg[\bigg\lVert\frac{1}{\sqrt n}\widehat s(\theta, \Xn^\ast)-\frac{1}{\sqrt n}s^\ast(\theta, \Xn^\ast)\bigg\rVert^4 \I_{\mA_{n,3}^C} \bigg] \\
    \leq& \sqrt{\E_{\Pi_n \cdot \Pn_{\theta^\ast}}\bigg[\bigg\lVert\frac{1}{\sqrt n}\widehat s(\theta, \Xn^\ast)-\frac{1}{\sqrt n}s^\ast(\theta, \Xn^\ast)\bigg\rVert^8 \bigg] \E_{\Pi_n \cdot \Pn_{\theta^\ast}}\big[ \I_{\mA_{n,3}^C} \big]} \\
    \lesssim& \sqrt{\E_{\Pi_n \cdot \Pn_{\theta^\ast}}\Big[n^4 +  \lVert \sqrt{n}(\theta - \theta^\ast) \rVert^8 +  \big\lVert \sqrt{n}(\theta - \MLE) \big\rVert^8 \Big] \E_{\Pi_n \cdot \Pn_{\theta^\ast}}\big[ \I_{\mA_{n,3}^C} \big]} \\
    \lesssim& n^{-(\frac{C_1C_0^2}{2} - 2)}
\end{align*}
where the second step is by Cauchy–Schwarz inequality, the third step is by \Cref{ass:bdd_derivatives}, \Cref{lem:true_sco_bd} and triangle inequality, and the last step is by \Cref{lem:tail_prob_A2} and \ref{lem:8th_mmt}.

Finally, adding (I) and (II) together, we have the fourth moment of the score error as .
\begin{align*}
&\E_{(\theta, \Xn^\ast)\sim \pi_n(\theta\mid \Xn^\ast)\pn_{\theta^\ast}(\Xn^\ast)} \normbigg{\frac{1}{\sqrt n}\widehat s(\theta, \Xn^\ast)- \frac{1}{\sqrt n}s^\ast(\theta, \Xn^\ast)}^4 \\
=& \text{(I)} + \text{(II)} \\
\lesssim&  (\log n)^2 \varepsilon_{N,n,1}^2 + (\log n)^2 \varepsilon_{N,n,2}^2 + \frac{(\log n)^3}{n} + n^{-(\frac{C_1C_0^2}{2} - 2)} \\
\lesssim& (\log n)^2 \varepsilon_{N,n,1}^2 + (\log n)^2 \varepsilon_{N,n,2}^2 + \frac{(\log n)^3}{n}
\end{align*}
where the last step is because $C_0 \ge \sqrt{\frac{6}{C_1}}$.
\end{proof}

\subsection{Proof of \Cref{thm:post_convergence_iid}}\label{sec:post_convergence_iid_pf}

The proof is similar to \Cref{sec:post_convergence_n_obs_pf}. The major difference is that we have a different way to control the total score-matching error, which was shown in \Cref{lem:score_err_fourth_moment}. Below we provide an equivalent of \Cref{lem:score_err_fourth_moment} for the case when we are matching the score on  a single observation.

\begin{lemma}Under the assumptions in \Cref{thm:post_convergence_iid}, we have
\begin{align*}
&\E_{(\theta, \Xn^\ast)\sim \pi_n(\theta\mid \Xn^\ast)p_{\theta^\ast}^{(n)}(\Xn^\ast)} \norm{\frac{1}{\sqrt n}\widehat s(\theta, \Xn^\ast)- \frac{1}{\sqrt n}s^\ast(\theta, \Xn^\ast)}^4 \\
\lesssim& (\log n)^2 \widetilde{\varepsilon}_{N,1}^2 + (\log n)^2 \widetilde{\varepsilon}_{N_R,m_R,2}^2 + n \log n \widetilde{\varepsilon}_{N_R,m_R,3}^2 + \frac{(\log n)^3}{n}
\end{align*}
\end{lemma}

\begin{proof}
With the same definition of $\mA_{n,3}(\Xn^\ast)$ as in \Cref{lem:score_err_fourth_moment}, we can do the same decomposition and obtain
\[
\E_{(\theta, \Xn^\ast)\sim \pi_n(\theta\mid \Xn^\ast)p_{\theta^\ast}^{(n)}(\Xn^\ast)} \norm{\frac{1}{\sqrt n}\widehat s(\theta, \Xn^\ast)- \frac{1}{\sqrt n}s^\ast(\theta, \Xn^\ast)}^4 \lesssim \text{(I.1)} + \text{(I.2)} + \text{(II)},
\]
where (I.1), (I.2) and (II) are defined the same as in \Cref{lem:score_err_fourth_moment}.

For (I.1), we have
\begin{align*}
\text{(I.1)} &= \log n \E_{\Pn_{\theta^\ast}(\Xn^\ast)}\Big[ \normbig{\frac{1}{\sqrt n}\widehat s(\theta^\ast, \Xn^\ast)- \frac{1}{\sqrt n}s^\ast(\theta^\ast, \Xn^\ast)}^2 \Big] \\
&= \log n \Big\{ \E_{P_{\theta^\ast}} \big[ \| \widehat s(\theta^\ast, X^\ast)-s^\ast(\theta^\ast, X^\ast) \|^2 \big] + (n-1) \| \E_{P_{\theta^\ast}}[\widehat s(\theta^\ast, X^\ast)] \|^2 \Big\} \\
&\lesssim \log n \big(\widetilde{\varepsilon}_{N,1}^2 + n \widetilde{\varepsilon}_{N_R, m_R,3}^2 \big)
\end{align*}
where the second equality is because $X_i^\ast,\;i=1,\dots,n$ are i.i.d., and the last step is by the uniform error bound assumptions in \Cref{thm:post_convergence_iid}.

For (I.2), we can bound it using the same way as in \Cref{lem:score_err_fourth_moment}. The only difference is that the upper bound for the error of the estimated fisher information matrix is
\[
\big\lVert I(\theta^\ast) - \widehat{I}(\theta^\ast) \big\rVert_2 \le 2C_5\widetilde{\varepsilon}_{N,1} + \widetilde{\varepsilon}_{N_R,m_R,2}
\]
by \Cref{lem:curv_err_single}, and the bound for (I.2) becomes
\[
\text{(I.2)} \lesssim (\log n)^2 \widetilde{\varepsilon}_{N,n,1}^2 + (\log n)^2 \widetilde{\varepsilon}_{N,n,2}^2 +  \frac{(\log n)^3}{n}
\]

For (II), it is the tail expectation and has the same bound as in \Cref{lem:score_err_fourth_moment}, i.e.
\[
\text{(II)} \lesssim n^{-(\frac{C_1C_0^2}{2} - 2)}
\]

Finally, adding (I.1), (I.2) and (II) together, we obtain
\begin{align*}
&\E_{(\theta, \Xn^\ast)\sim \pi_n(\theta\mid \Xn^\ast)p_{\theta^\ast}^{(n)}(\Xn^\ast)} \norm{\frac{1}{\sqrt n}\widehat s(\theta, \Xn^\ast)- \frac{1}{\sqrt n}s^\ast(\theta, \Xn^\ast)}^4 \\
\lesssim& (\log n)^2 \widetilde{\varepsilon}_{N,n,1}^2 + (\log n)^2 \widetilde{\varepsilon}_{N,n,2}^2 + n \log n \widetilde{\varepsilon}_{N,n,3}^2 + \frac{(\log n)^3}{n}
\end{align*}

Note that if we do not have the debiasing step, then we would bound (I.1) using the score error alone as
\begin{align*}
\text{(I.1)} &= \log n \E_{\Pn_{\theta^\ast}(\Xn^\ast)}\Big[ \normbig{\frac{1}{\sqrt n}\widehat s(\theta^\ast, \Xn^\ast)- \frac{1}{\sqrt n}s^\ast(\theta^\ast, \Xn^\ast)}^2 \Big] \\
&\leq n\log n \E_{P_{\theta^\ast}} \big[ \| \widehat s(\theta^\ast, X^\ast)-s^\ast(\theta^\ast, X^\ast) \|^2 \big] \\
&\leq n\log n \widetilde{\varepsilon}_{N,1}^2
\end{align*}
and the final bound becomes
\begin{align*}
\E_{(\theta, \Xn^\ast)\sim \pi_n(\theta\mid \Xn^\ast)p_{\theta^\ast}^{(n)}(\Xn^\ast)} \norm{\frac{1}{\sqrt n}\widehat s(\theta, \Xn^\ast)- \frac{1}{\sqrt n}s^\ast(\theta, \Xn^\ast)}^4
\lesssim  n\log n \widetilde{\varepsilon}_{N,n,1}^2 +  (\log n)^2 \widetilde{\varepsilon}_{N,n,2}^2  + \frac{(\log n)^3}{n}
\end{align*}
That is, now we need to control $\widetilde{\varepsilon}_{N,n,1}^2$ under $O(\frac{1}{n\log n})$, instead of controlling $\widetilde{\varepsilon}_{N,n,3}^2$ under $O(\frac{1}{n\log n})$ with the debiasing step.
\end{proof}

\subsection{Proof of \Cref{lem:debias_error}}\label{sec:debias_error_pf}

For simplicity, we write $\wh{s}(\theta,X)=s_{\hat\phi}(\theta,X)$ and $\hat h(\theta)=h_{\psi}(\theta)$.

Since $h(\theta) \equiv 0$ is feasible under \eqref{eq:score_demean} and $\widehat{h}$ is a minimizer, we have
\begin{equation}\label{res_gis0}
\mathbb{E}_{\theta\sim q(\theta)} \lVert \widehat{h}(\theta) - \mathbb{E}_{X\sim P_\theta} \widehat{s}(\theta, X) \rVert^2 \le \mathbb{E}_{\theta\sim q(\theta)} \lVert \mathbb{E}_{X\sim P_\theta} \widehat{s}(\theta, X) \rVert^2
\end{equation}
For the debiasd score, using the fact that $\mathbb{E}_{X\sim P_\theta} [s^\ast(\theta, X)] = \bm{0}$, we have
\begin{align*}
& \mathbb{E}_{\theta\sim q(\theta)}\Big\{\mathbb{E}_{X\sim P_\theta} \lVert \widehat{s}(\theta, X) - \widehat{h}(\theta) - s^\ast(\theta, X) \rVert^2 \Big\} \\
=& \mathbb{E}_{\theta\sim q(\theta)}\Big\{\mathbb{E}_{X\sim P_\theta} \lVert \widehat{s}(\theta, X) - \mathbb{E}_{X\sim P_\theta}[\widehat{s}(\theta, X)] - s^\ast(\theta, X) + \mathbb{E}_{X\sim P_\theta}[\widehat{s}(\theta, X)] - \widehat{h}(\theta) \rVert^2 \Big\} \\
=& \mathbb{E}_{\theta\sim q(\theta)}\Big\{\mathbb{E}_{X\sim P_\theta} \lVert \widehat{s}(\theta, X)  - \mathbb{E}_{X\sim P_\theta}[\widehat{s}(\theta, X) ] - s^\ast(\theta, X) \rVert^2 + \lVert \mathbb{E}_{X\sim P_\theta}[\widehat{s}(\theta, X)] - \widehat{h}(\theta) \rVert^2 \Big\}.
\end{align*}

Similarly for the score-matching error of $\hat s(\theta,X)$, we have
\begin{align*}
& \mathbb{E}_{\theta\sim q(\theta)}\Big\{\mathbb{E}_{X\sim P_\theta} \lVert \widehat{s}(\theta, X) - s^\ast(\theta, X) \rVert^2 \Big\} \\
=& \mathbb{E}_{\theta\sim q(\theta)}\Big\{\mathbb{E}_{X\sim P_\theta} \lVert \widehat{s}(\theta, X)  - \mathbb{E}_{X\sim P_\theta}[\widehat{s}(\theta, X) ] - s^\ast(\theta, X) \rVert^2 + \lVert \mathbb{E}_{X\sim P_\theta}[\widehat{s}(\theta, X)] \rVert^2 \Big\}
\end{align*}
Then, the result is proved because of \eqref{res_gis0}.

\subsection{Auxiliary Lemmas}\label{pf:aux_lemmas}
The following Lemma shows the scaled score $\frac{1}{\sqrt n} s^\ast(\theta, \Xn^\ast)$ is upper bounded by $\sqrt n(\theta - \MLE)$, which agrees with the score given by the limit distribution in the BvM theorem.
\begin{lemma}\label{lem:true_sco_bd}
Under \Cref{ass:lipschitz_score}, we have
\[
\bigg\lVert \frac{1}{\sqrt n} s^\ast(\theta, \Xn^\ast) \bigg\rVert \leq \lambda_L\lVert \sqrt n(\theta - \MLE) \rVert
\]
\begin{proof}
Since $s^\ast(\MLE, \Xn^\ast) = 0$, we have
\begin{equation*}
\lVert s^\ast(\theta, \Xn^\ast) \rVert = \lVert s^\ast(\theta, \Xn^\ast) - s^\ast(\MLE, \Xn^\ast) \rVert
\leq n\lambda_L\lVert \theta - \MLE \rVert
\end{equation*}
where the last step is because $s^\ast(\cdot, \Xn^\ast)$ is $n\lambda_L$-Lipschitz, by \Cref{ass:lipschitz_score}.
\end{proof}
\end{lemma}

\begin{lemma}\label{lem:tail_prob_A2}
For any $C_0 \ge 1$, denote $\mA_{n,3}= \Big\{\theta: \sqrt n \lVert \theta-\theta^\ast \rVert\leq C_0\sqrt{\log n}\Big\} \bigcap \Big\{\Xn^\ast: \sqrt n \big\lVert \MLE - \theta^\ast \big\rVert\leq C_0\sqrt{\log n}\Big\}$, then
\[
\E_{(\theta, \Xn^\ast)\sim \pi_n(\theta\mid \Xn^\ast)\pn_{\theta^\ast}(\Xn^\ast)}\big[\I_{\mA_{n,3}^C}\big] \leq 2n^{-C_1C_0^2}
\]
\end{lemma}
\begin{proof}
\begin{align*}
&\E_{(\theta, \Xn^\ast)\sim \pi_n(\theta\mid \Xn^\ast)\pn_{\theta^\ast}(\Xn^\ast)}\big[\I_{\mA_{n,3}^C}\big] \\
\le& \E_{\Pn_{\theta^\ast}}\Pi_n\Big[\sqrt n \lVert \theta-\theta^\ast \rVert > C_0\sqrt{\log n} \bigm| \Xn^\ast\Big] + \Pn_{\theta^\ast}\Big[\sqrt n \big\lVert \MLE - \theta^\ast \big\rVert > C_0\sqrt{\log n}\Big] \\
\le& 2n^{-C_1C_0^2}
\end{align*}
where the last step is by \Cref{ass:true_converge}.
\end{proof}

\begin{lemma}\label{lem:8th_mmt}
Under \Cref{ass:true_converge}, we have
\begin{align*}
\E_{(\theta, \Xn^\ast)\sim \pi_n(\theta\mid \Xn^\ast)\pn_{\theta^\ast}(\Xn^\ast)}\big[\big\lVert\sqrt{n}(\theta - \theta^\ast)\big\rVert^8\big] &\lesssim  \log^4 n \\
\E_{(\theta, \Xn^\ast)\sim \pi_n(\theta\mid \Xn^\ast)\pn_{\theta^\ast}(\Xn^\ast)}\big[\big\lVert\sqrt{n}(\theta - \MLE)\big\rVert^8\big] &\lesssim  \log^4 n
\end{align*}
\begin{proof}
The first line in the result is because
\begin{equation}\label{eq:lem_8th_mmt}
\begin{split}
& \E_{(\theta, \Xn^\ast)\sim \pi_n(\theta\mid \Xn^\ast)\pn_{\theta^\ast}(\Xn^\ast)}\big[\big\lVert\sqrt{n}(\theta - \theta^\ast)\big\rVert^8\big] \\
=& n^4 \int_0^{+\infty} 8t^7 P_{(\theta, \Xn^\ast)\sim \pi_n(\theta\mid \Xn^\ast)\pn_{\theta^\ast}(\Xn^\ast)}\big(\lVert \theta - \theta^\ast \rVert > t \big)\;dt \\
=& n^4 \int_0^{\sqrt{\frac{\log n}{n}}} 8t^7 \E_{\Pn_{\theta^\ast}}\big[\Pi_n\big(\lVert \theta - \theta^\ast \rVert > t \bigm| \Xn^\ast \big)\big]\;dt \\
&+ n^4 \int_{\sqrt{\frac{\log n}{n}}}^{+\infty} 8t^7 \E_{\Pn_{\theta^\ast}}\big[\Pi_n\big(\lVert \theta - \theta^\ast \rVert > t \bigm| \Xn^\ast \big)\big]\;dt \\
\leq& \log^4 n + n^4 \int_{\sqrt{\frac{\log n}{n}}}^{+\infty} 8t^7 \exp(-C_1nt^2)\;dt \qquad\text{(by \Cref{ass:true_converge})} \\
\leq& \log^4 n + \frac{4n^{-C_1}}{C_1^4}(6+6C_1\log n + 3C_1^2 \log^2 n + C_1^3 \log^3 n)
\end{split}
\end{equation}
We next show the second line in the result. By triangle inequality, we have
\begin{align*}
&\E_{(\theta, \Xn^\ast)\sim \pi_n(\theta\mid \Xn^\ast)\pn_{\theta^\ast}(\Xn^\ast)}\big[\big\lVert\sqrt{n}(\theta - \MLE)\big\rVert^8\big]\\
\lesssim& \E_{(\theta, \Xn^\ast)\sim \pi_n(\theta\mid \Xn^\ast)\pn_{\theta^\ast}(\Xn^\ast)}\big[\big\lVert\sqrt{n}(\theta - \theta^\ast)\big\rVert^8\big] + \E_{\Pn_{\theta^\ast}}\big[\big\lVert\sqrt{n}(\MLE - \theta^\ast)\big\rVert^8\big]
\end{align*}
Using \Cref{ass:true_converge} and similar calculations in \eqref{eq:lem_8th_mmt}, we can also get
\[
\E_{\Pn_{\theta^\ast}}\big[\big\lVert\sqrt{n}(\MLE - \theta^\ast)\big\rVert^8\big] \leq \log^4 n + \frac{4n^{-C_1}}{C_1^4}(6+6C_1\log n + 3C_1^2 \log^2 n + C_1^3 \log^3 n)
\]
and the result follows.
\end{proof}
\end{lemma}

The following two Lemmas bound the difference between the true and estimated Fisher information matrices in the single data and full data score matching.

\begin{lemma}\label{lem:curv_err_single}
Under the error assumptions in \Cref{thm:post_convergence_iid},
\[
\big\lVert I(\theta^\ast) - \widehat{I}(\theta^\ast) \big\rVert_2 \le 2C_5\widetilde{\varepsilon}_{N,1} + \widetilde{\varepsilon}_{N_R,m_R,2}
\]
\end{lemma}
\begin{proof}
We introduce $s^\ast(\theta^\ast, X)s^\ast(\theta^\ast, X)^T$ and $\widehat{s}(\theta^\ast, X) \widehat{s}(\theta^\ast, X)^T$ as intermediate terms, and use the triangle inequality.
\begin{align*}
& \big\lVert I(\theta^\ast) - \widehat{I}(\theta^\ast) \big\rVert_2 \\
\le& \big\lVert \mathbb{E}_{X\sim P_{\theta^\ast}}\big[ s^\ast(\theta^\ast, X)s^\ast(\theta^\ast, X)^T - \widehat{s}(\theta^\ast, X) \widehat{s}(\theta^\ast, X)^T \big] \big\rVert_2 \\
&+ \big\lVert \mathbb{E}_{X\sim P_{\theta^\ast}} \big[ \widehat{s}(\theta^\ast, X) \widehat{s}(\theta^\ast, X)^T + \nabla \widehat{s}(\theta^\ast, X) \big] \big\rVert_2
 \\
\le&  \Big\lVert \mathbb{E}_{X\sim P_{\theta^\ast}} \Big[s^\ast(\theta^\ast, X) \big(s^\ast(\theta^\ast, X) - \widehat{s}(\theta^\ast, X) \big)^T \Big] \Big\rVert_2 \\
&+ \big\lVert \mathbb{E}_{X\sim P_{\theta^\ast}} \big[\big(s^\ast(\theta^\ast, X) - \widehat{s}(\theta^\ast, X) \big) \widehat{s}(\theta^\ast, X)^T \big] \big\rVert_2 + \tilde{\varepsilon}_{2} \\
\le& \mathbb{E}_{X\sim P_{\theta^\ast}} \bigg[ \norm{ s^\ast(\theta^\ast, X) \big(s^\ast(\theta^\ast, X) - \widehat{s}(\theta^\ast, X) \big)^T  }_2 \bigg] \\
&+ \mathbb{E}_{X\sim P_{\theta^\ast}} \Big[\Big\lVert \big(s^\ast(\theta^\ast, X) - \widehat{s}(\theta^\ast, X) \big) \widehat{s}(\theta^\ast, X)^T \Big\rVert_2 \Big] + \tilde{\varepsilon}_{2}  &\text{(by Jensen's inequality)} \\
=&  \mathbb{E}_{X\sim P_{\theta^\ast}} \big[ \norm{ s^\ast(\theta^\ast, X)}_2 \cdot \norm{ s^\ast(\theta^\ast, X) - \widehat{s}(\theta^\ast, X) }_2 \big] + \\
&\qquad \mathbb{E}_{X\sim P_{\theta^\ast}} \big[\norm{ s^\ast(\theta^\ast, X) - \widehat{s}(\theta^\ast, X) }_2 \cdot \norm{ \widehat{s}(\theta^\ast, X) }_2 \big] + \tilde{\varepsilon}_{2}    \\
\le& 2C_5\tilde{\varepsilon}_{N,1} + \tilde{\varepsilon}_{N_R, m_R,2},
\end{align*}

where the last step is by Cauchy-Schwarz inequality, \Cref{ass:bdd_derivatives} and \Cref{ass:uniform_sm_err}.
\end{proof}

\begin{lemma}\label{lem:curv_err_full}
Under \Cref{ass:uniform_sm_err},
\[
\big\lVert I(\theta^\ast) - \widehat{I}(\theta^\ast) \big\rVert_2 \le 2C_5\varepsilon_{N,n,1} + \varepsilon_{N,n,2}
\]
\end{lemma}
\begin{proof}
The proof is the same as in \Cref{lem:curv_err_single}, except that we need to upper bound the score error on a single observation in the last step, which is shown below
\begin{align*}
\varepsilon_{N,n,1}^2 &\ge \frac{1}{n} \E_{\Xn \sim P_{\theta^\ast}^{(n)}} \bigg[ \bigg\lVert\sum_{i=1}^n [\widehat{s}(\theta, X_i)- s^\ast(\theta, X_i)] \bigg\rVert^2 \bigg] \\
&= \frac{1}{n} \Big\{ n \E_{X \sim P_{\theta^\ast}}\norm{\widehat{s}(\theta, X) - s^\ast(\theta, X)}^2 +  n(n-1)\norm{\E_{X\sim P_{\theta^\ast}}[\widehat{s}(\theta, X) - s^\ast(\theta, X)]}^2 \Big\} &\text{(since i.i.d.)} \\
&\ge \E_{X\sim P_{\theta^\ast}} \big[ \norm{\widehat{s}(\theta, X) - s^\ast(\theta, X)}^2 \big]
\end{align*}
\end{proof}

\section{Details of the Empirical Studies and Additional Experimental Results
}\label{sec:additional_implementation}
In this section, we provide further details on the empirical studies in \Cref{sec:simulation}, together with additional experimental results that complement those reported in the main paper.

For each empirical study, we provide (1) the data generating process when it is not described in the main paper and (2) implementation details for all methods, including simulation costs and training and sampling configurations.

The additional experimental results include the following.  In \Cref{sec:solve_bdr_queuing}, we use the queuing model example to demonstrate how to address violations of the boundary condition required for score matching. In \Cref{sec:mono_reg_comp}, we examine more carefully how different parts of the score-matching procedure affect estimation accuracy by comparing variants of the single-model and n-model in terms of three types of score error.  For the mRNA transfection example, we present the BSL results for the simulation study in \Cref{sec:SDE_BSL}, as well as the NLE results and a posterior predictive assessment of the bimodal posterior obtained by our single-model for the real-data study in \Cref{sec:more_res_SDErealdata}. We present an additional empirical study on partially observed stochastic epidemic models in \Cref{sec:sto_epi_details}.

\subsection{Details of the queuing model example}\label{sec:queuing_details}

In this subsection, we first provide a more detailed discussion on how we resolve the constrained support problem outlined in \Cref{sec:boundary_sol}. Next we provide all details on implementing our methods and the compared methods in this example.

\subsubsection{Solving the boundary condition}\label{sec:solve_bdr_queuing}
As we discuss in \Cref{sec:boundary_sol}, the score-matching objective in \eqref{eq:score_loss} cannot be directly applied to the queuing model since the boundary condition in \Cref{ass:boundary} is violated. In this example, we consider both solutions mentioned in \Cref{sec:boundary_sol}.

{\bf Solution 1:} In this example, we have full knowledge of the support of $\Xn$ as $\{\Xn: \pn_\theta(\Xn) > 0\} = [\theta_1, +\infty)^{\otimes np}$, which means that the support of the posterior is
\begin{align*}
&\{(\theta_1, \theta_2 - \theta_1, \theta_3): \pi_n(\theta\mid \Xn) > 0\} \\
=& \{(\theta_1, \theta_2 - \theta_1, \theta_3): \pi(\theta) > 0\} \bigcap \{(\theta_1, \theta_2 - \theta_1, \theta_3): \pn_\theta(\Xn) > 0\} \\
=& \Big[0, \min\big(10, \min_{i,j}\{x_{ij}\}\big)\Big] \times [0, 10] \times [0, 0.5]
\end{align*}
Therefore, we use the weight function $g(\theta_1, \theta_2 - \theta_1, \theta_3) = \Big(\text{dist}(\theta_1, \Big[0, \min\big(10, \min_{i,j}\{x_{ij}\}\big)\Big]),$ $ \text{dist}(\theta_2-\theta_1, [0, 10]), \text{dist}(\theta_3, [0, 0.5]) \Big)$ and train the score network using \eqref{sm_thm_eq}.  For this example, we just use Euclidean distance.

Note that we do not apply the weight function when we train the network matching on single data score, since the mean-zero property of the likelihood score no longer holds, making the debiasing step inapplicable and the score error on $\Xn$ could accumulate. This can potentially be amended by including the weight function in the debiasing step as well.

{\bf Solution 2:} Similar to denoising score matching, there is a trade-off in choosing the noise level $\sigma_\varepsilon$. The posterior distribution based on the noised model will be far from the true posterior if $\sigma_\varepsilon$ is too large, while the training of the score network can suffer from numerical instability if $\sigma_\varepsilon$ is too small, as we show in \Cref{fig:queuing_noise}. As a heuristic approach, we recommend choosing $\sigma_\varepsilon$ based on the variance of $\Xn^\ast$. In the queuing example, the standard deviation of all the dimensions of $\Xn^\ast$ is around $5$, and we find $\sigma_\varepsilon$ values from around $5\text{-}
10\%$ of that works well. Finally, we use $\sigma_\varepsilon = 0.25$  and train the debiased score network in  \Cref{alg:langevin_single_obs}.

Moreover, it is worth noting that Solution 2 is ineffective when we train the network using the score matching loss on $n$ data as in \Cref{alg:regularized_langevin}. Since the support boundary of $\pn_\theta(\Xn)$ is much sharper, adding a small amount of Gaussian noise does not help too much, compared to the case when we target at $p_\theta(X)$. The resulting noisy score still has a drastic change near the original boundary, which makes the training problematic.

\begin{figure}[htbp]
    \centering
    \includegraphics[width=0.99\textwidth]{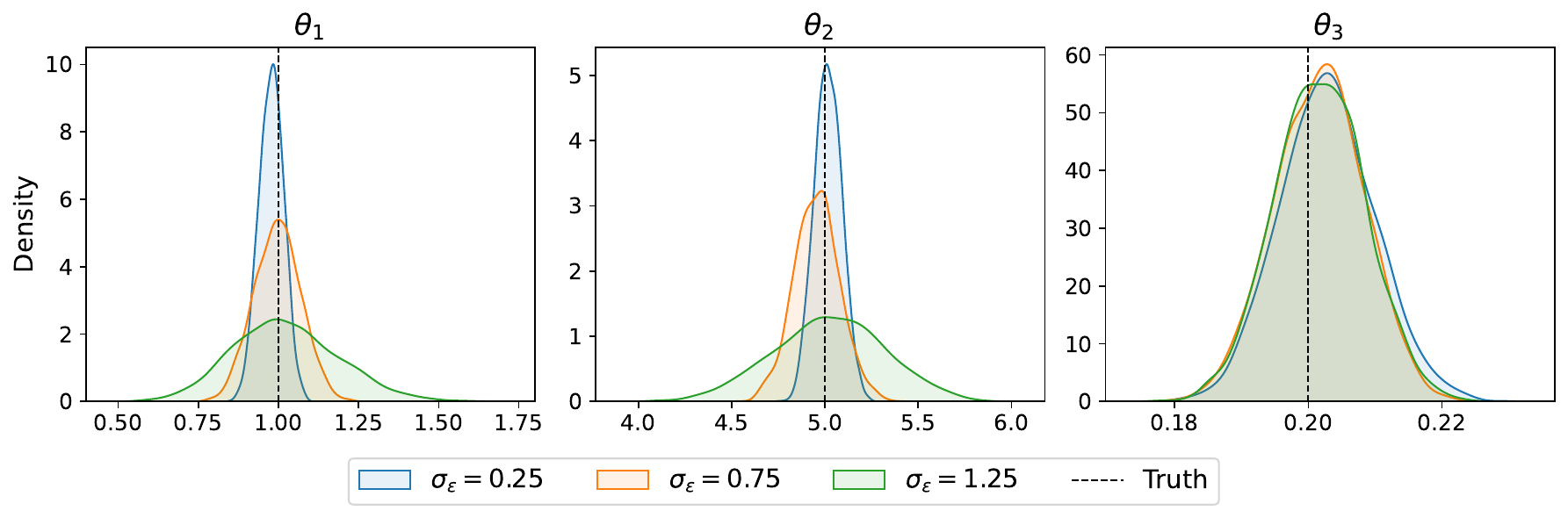}
    \caption{Results of \Cref{alg:langevin_single_obs} on the queuing model under different noise levels}
    \label{fig:queuing_noise}
\end{figure}

As a result, we recommend using Solution 1 when the data has strong dependency structure or the weight function can provide a lot of information into the sampling procedure, as we observe for $\theta_1$ of the queuing model. Solution 2 is preferred and straightforward when the data has i.i.d. structure.

\subsubsection{Implementation details}\label{sec:que_settings}

{\bf single-model training details.}
In order to address the boundary issue, we add noise to all $X$ with $\sigma_\varepsilon = 0.25$. We have $N = 2\times 10^5$ for the reference table $\mathcal{D}^S$, and $(N_R, m_R) = (1\times 10^4, 5\times10^2)$ for the reference table $\mathcal{D}^R$. We use an ELU neural network with $3$ hidden layer and $64$ units in each hidden layer. The network is trained with batch size $500$ and learning rate gradually decreasing from $1\times 10^{-3}$ to $1\times 10^{-5}$, for $200$ epochs or till convergence, and another $50$ epochs after adding the curvature regularization. The mean regression is implemented using an ELU neural network with $3$ hidden layer and $64$ units in the hidden layer. The network is trained with batch size $10$ and learning rate gradually decreasing from $1\times 10^{-3}$ to $1\times 10^{-5}$, for $500$ epochs or until convergence, and another $10$ epochs after adding the curvature regularization. The curvature penalty parameters are chosen as $\lambda_1=1\times 10^{-3}$ in the score matching part and $\lambda_2=1\times 10^{-4}$ in the mean regression part. For sampling, we inject $3$ different sets of noise to $\Xn^\ast$ and obtain $1\;000$ samples using each set of noisy observed data, and aggregate them to get $3\;000$ posterior samples.

{\bf n-model training details.} We have reference table size $N = 2\times 10^4$. We use a Tanh neural network with $2$ hidden layer and $64$ units in the hidden layer. The network is trained with batch size $500$ and learning rate $1\times 10^{-3}$, for $3000$ epochs or till convergence. For sampling, we obtain $8\;000$ samples from $10$ independent Langevin Markov chains.

{\bf Details of NLE and NPE.} We use the ``sbi" Python package from \citep{tejero-cantero2020sbi} to implement both NLE and NPE, and we choose the flow network as a Masked Autoregressive Flow network \citep{papamakarios2017masked} with $5$ autoregressive transforms, each of which has $2$ hidden layers of $50$ units each. For NLE, we have reference table size $N = 1\times 10^7$ . The network is trained with batch size $200$ and learning rate $5\times 10^{-4}$, for $30$ epochs. We draw $2\;000$ posterior samples from $4$ Markov Chains using the trained NLE model. For NPE, we have a permutation invariant embedding network $f_\phi(\Xn) = f_{\phi_1}(\sum_{i=1}^n f_{\phi_2}(X_i))$ before the flow model, where $f_{\phi_2}$ is a fully connected neural network with $3$ hidden layers of $128$ units each and output dimension $128$, and $f_{\phi_1}$ is a fully connected neural network with $2$ hidden layers of $100$ units each and output dimension $20$. We have reference table size $N = 2\times 10^4$. The network is trained with batch size $200$ and learning rate $5\times 10^{-4}$ until convergence. We draw $10\;000$ posterior samples from the trained NPE model.

{\bf Details of other models.}  For BSL, we use the sample mean and element-wise standard deviation of $\Xn$ as  summary statistics.
For the ABC method, we generate reference table of size $N=20\,000$ and keep $200$ samples that have the smallest $\text{W}_1$ distance. For the BSL method, we obtain $20\,000$ samples from $10$ independent Markov chains, where each chain is drawn using Metropolis–Hastings algorithm, with length $3\;000$ ($1\,000$ burn-in's), and we use $100$ simulations at each iteration to estimate the mean and covariance of the synthetic Gaussian likelihood.

{\bf Simulation cost comparison.} The simulation cost for all the methods is listed in \Cref{simu_budget_que},
where one unit is the cost of generating $n$ observations. For the n-model and ABC method, the cost is the size $N$ of the reference table. For BSL, its cost is the product of the number of chains, the length per chain, and the number of simulations at each iteration within each chain. For the single-model, its cost is $N/n + N_Rm_R/n$.

\begin{table}[!ht]
    \centering
\caption{Simulation cost in the queuing example}
\begin{tabular}{cccccc}
\toprule
single-model & n-model & ABC & BSL & NLE & NPE \\
\midrule
 $1.04\times 10^4$ & $2\times 10^4$ & $2\times 10^4$ & $3\times10^6$ & $2\times 10^4$ & $2\times 10^4$\\
\bottomrule
\end{tabular}
    \label{simu_budget_que}
\end{table}

\subsection{Details of the monotonic regression example}\label{sec:mono_reg_details}

\subsubsection{Data generating process and the true posterior}
Following \citet{mckay2011variable}, we approximate the tanh function by a Bernstein polynomial function of degree $M$
\begin{equation}\label{BP_func}
    B_M(x) = \sum_{k=0}^M\beta_k \binom{M}{k} x^k (1-x)^{M - k},
\end{equation}
where the coefficients $\beta = (\beta_0, ..., \beta_M)$ are subject to the constraints that $\beta_{k-1} \le \beta_k$ for all $k = 1,..., M$, in order to ensure monotonicity of $B_M(\cdot)$. For convenience of computation, the following reparameterization is employed:
\begin{equation*}
\theta_0 = \beta_0, \quad \theta_k = \beta_k - \beta_{k-1},\quad k = 1,..., M,
\end{equation*}
and the final approximation model with parameter $\theta = (\theta_0, ...,\theta_M)$ can be written as
\begin{equation}\label{model_BP}
y_i = \sum_{k=0}^M \theta_k b_M(x_i, k) + \varepsilon_i, \quad \varepsilon_i \overset{i.i.d.}{\sim} N(0, \sigma^2) \quad  \text{ s.t. } \theta_k \ge 0,\, k = 1,...,M
\end{equation}
where $b_M(\cdot, k)$ has a known form and can be derived from \eqref{BP_func}.

We generate $n=1000$ i.i.d. data $\{(x_i, y_i): i = 1,..., n\}$ and set polynomial order $M = 10$, which provides models flexible enough to approximate the $\text{tanh}(\cdot)$ function. The prior is uniform on $[-5, 5] \times [0, 1]^M$.  Under this prior, the posterior is a multivariate normal distribution with mean $(\bm{D}^\top \bm{D})^{-1} \bm{D}^\top \bm{y}$ and covariance $\sigma^2 (\bm{D}^\top \bm{D})^{-1}$ truncated at the prior domain, where $\bm{D}$ denotes the design matrix corresponding to \eqref{model_BP}.

Although the posterior has a closed form, it is challenging to directly sample from this truncated normal distribution, because the features of the design matrix $\bm{D}$ are highly correlated, making the covaraince matrix ill-conditioned, and samples drawn from the corresponding untruncated normal distribution barely fall into the domain. Therefore, we follow \citep{mckay2011variable} and use a Gibbs sampling algorithm to draw samples from the true posterior, and treat these samples as the ground truth.

\subsubsection{Implementation details}

For all optimizations in this example, we use Adam.

{\bf Localization.} In this example, the generator $\tau(\theta, \Zn)$ is straightforward: each $Z_i$, $i=1, \dots, n$, consists of a set of i.i.d. $\text{Uniform}(0, 1)$ random variables for generating ${x_i}$'s and a set of i.i.d. $\mathcal{N}(0, 1)$ random variables for generating $\varepsilon_i$'s.
We set Adam with learning rate gradually decreasing from $10^{-1}$ to $10^{-3}$. Each run converges around $500$ iterations, so the simulation cost to obtain $100$ samples is around $5\times 10^4$ (in unit of $\Xn$).

{\bf single-model training details.} We have $N = 2\times 10^6$ for the reference table $\mathcal{D}^S$, and $(N_R, m_R) = (1\times 10^5, 1\times 10^3)$ for the reference table $\mathcal{D}^R$. We use an ELU neural network with $3$ hidden layers and $64$ units in each hidden layer. The network is trained with batch size $1\,000$ and learning rate gradually decreasing from $10^{-3}$ to $10^{-5}$, for $100$ epochs or till convergence, and another $50$ epochs after adding the curvature regularization. The mean regression is implemented using an ELU neural network with $3$ hidden layers and $64$ units in each hidden layer. The network is trained with batch size $256$ and learning rate gradually decreasing from $10^{-3}$ to $10^{-5}$, for $300$ epochs or until convergence, and another $100$ epochs after adding the curvature regularization. The curvature penalty parameter is chosen as $10^{-3}$ in both the score matching and the mean regression parts. For our method, we obtain $10\,000$ samples from $1\,000$ independent Langevin Markov chains, and  use an annealing schedule to mitigate the effect of strong correlation, with the tempering parameter increases from $0.1$, $0.2 \ldots$, to $1$. It is worth mentioning that since the estimated score by the neural network is naturally vectorized, drawing multiple independent Langevin Markov chains is computationally efficient.

{\bf n-model training details.} We have reference table size $N = 2\times 10^5$. We use an ELU neural network with $3$ hidden layers and $64$ units in each hidden layer. The network is trained with batch size $200$ and learning rate gradually decreasing from $10^{-3}$ to $10^{-5}$, for $300$ epochs or till convergence, and another $50$ epochs after adding the curvature regularization. The curvature penalty parameter is chosen as $1$. The configuration of the n-model-5x is basically the same, except that we increase the training data to $N = 5\times 10^5$ and decrease the maximum number of epochs by $5$ times. The sampling schedule for both n-model and n-model-5x is the same as single-model.

{\bf Details of NLE and NPE.} We use the ``sbi" Python package from \citep{tejero-cantero2020sbi} to implement both NLE and NPE, and we choose the flow network as a Masked Autoregressive Flow network \citep{papamakarios2017masked} with $5$ autoregressive transforms, each of which has $2$ hidden layers of $50$ units each. For NLE, we have reference table size $N = 2 \times 10^8$ . The network is trained with batch size $200$ and learning rate $5\times 10^{-4}$, for $2$ epochs. We draw $1\,000$ posterior samples from $4$ Markov Chains using the trained NLE model. For NPE, we have a permutation invariant embedding network $f_\phi(\Xn) = f_{\phi_1}(\sum_{i=1}^n f_{\phi_2}(X_i))$ before the flow model, where both $f_{\phi_1}$ and $f_{\phi_2}$ are fully connected neural networks with $3$ hidden layers of $128$ units each and output dimension $128$. We have reference table size $N = 2 \times 10^5$. The network is trained with batch size $200$ and learning rate $5\times 10^{-4}$ for $300$ epochs or until convergence. We draw $10\,000$ samples from the estimated posterior and samples are reweighted due to the use of proposal distribution.

{\bf Details of other models.}
For the Gibbs posterior, we obtain $100\;000$ samples from $10$ independent runs as the ground truth. For the ABC method, we generate a reference table of size $N=1\times10^6$ and keep $1\,000$ samples that have the smallest $\text{W}_1$ distance, and we reweight the samples by $\frac{\pi(\theta)}{q(\theta)}$. For the BSL method,
we choose the summary statistics to be a part of the sufficient statistics. Specifically, the sufficient statistics in this example is $(\bm{D}\bm{D}^T, \bm{D}^T \bm{y})$, where $\bm{D} = \bm{D}(\bm{x})$ is the design matrix containing the polynomial features of $\bm{x}$. We choose only the second part $\bm{D}^T \bm{y}$ as the summary statistics, because the first part $\bm{D}\bm{D}^T$ does not depend on $\theta$, concentrating around its mean, and has much higher dimension than the second part.
We obtain $10\;000$ samples from $10$ independent Markov chains, where each chain is drawn using Metropolis–Hastings algorithm, with length $1\;200$ ($200$ burn-in's), and we use $100$ simulations at each iteration to estimate the mean and covariance of the synthetic Gaussian likelihood.

{\bf Simulation cost comparison.} The simulation cost for all the methods is listed in \Cref{simu_budget_reg}, where one unit is the cost of generating $n$ observations. For the n-model and ABC method, the cost is the size $N$ of the reference table. For BSL, its cost is the product of the number of chains, the length per chain, and the number of simulations at each iteration within each chain. For the single-model, its cost is $N/n + N_Rm_R/n$.

\begin{table}[htbp]
    \centering
\caption{Simulation cost in the monotonic regression example}
\resizebox{\textwidth}{!}{\begin{tabular}{cccccccc}
\toprule
Localization & single-model & n-model & n-model-5x & ABC-W1 & BSL & NPE & NLE \\
\midrule
$5 \times 10^4$ & $1.02\times 10^5$  & $2\times 10^5$ & $1\times 10^6$ & $1\times 10^6$ & $1.2\times 10^6$ & $2\times 10^5$ & $2\times 10^5$ \\
\bottomrule
\end{tabular}
}
    \label{simu_budget_reg}
\end{table}

\subsubsection{Comparison between single-model and n-model}\label{sec:mono_reg_comp}
Since the true score is available in this example, we take a closer look at how different components of the score-matching procedure affect estimator accuracy. We evaluate three types of score-matching losses: (1) on a single observation $(\theta, X)\sim q(\theta)p_\theta(X)$  (loss-1), (2) on $n$ observations $(\theta, \Xn)\sim q(\theta)\pn_\theta(\Xn)$ (loss-$n$), and (3) on the posterior draws $\theta \sim\pi_n(\theta\mid \Xn^\ast)$ (loss-$n$-posterior). We consider 7 different models all trained from the same proposal distribution $q(\theta)$. There are 4 variants of the score matching on single observation: (1) without debiasing or curvature (1-model), (2) with only curvature (1-model-C), (3) with only debiasing (1-model-D) and (4) with both curvature and debiasing as in \Cref{alg:langevin_single_obs}
(1-model-DC). For the model trained on matching $n$ observations, we consider 3 variants: (1) without curvature (n-model), (2) with curvature as in \Cref{alg:regularized_langevin} (n-model-C), and (3) with curvature and trained on reference table of size $5N$ (n-model-C5). We report the three losses for all 7 models in 10 experiments  in \Cref{score_loss_plt}.

\begin{figure}[!ht]
    \centering
    \includegraphics[width=0.8\textwidth]{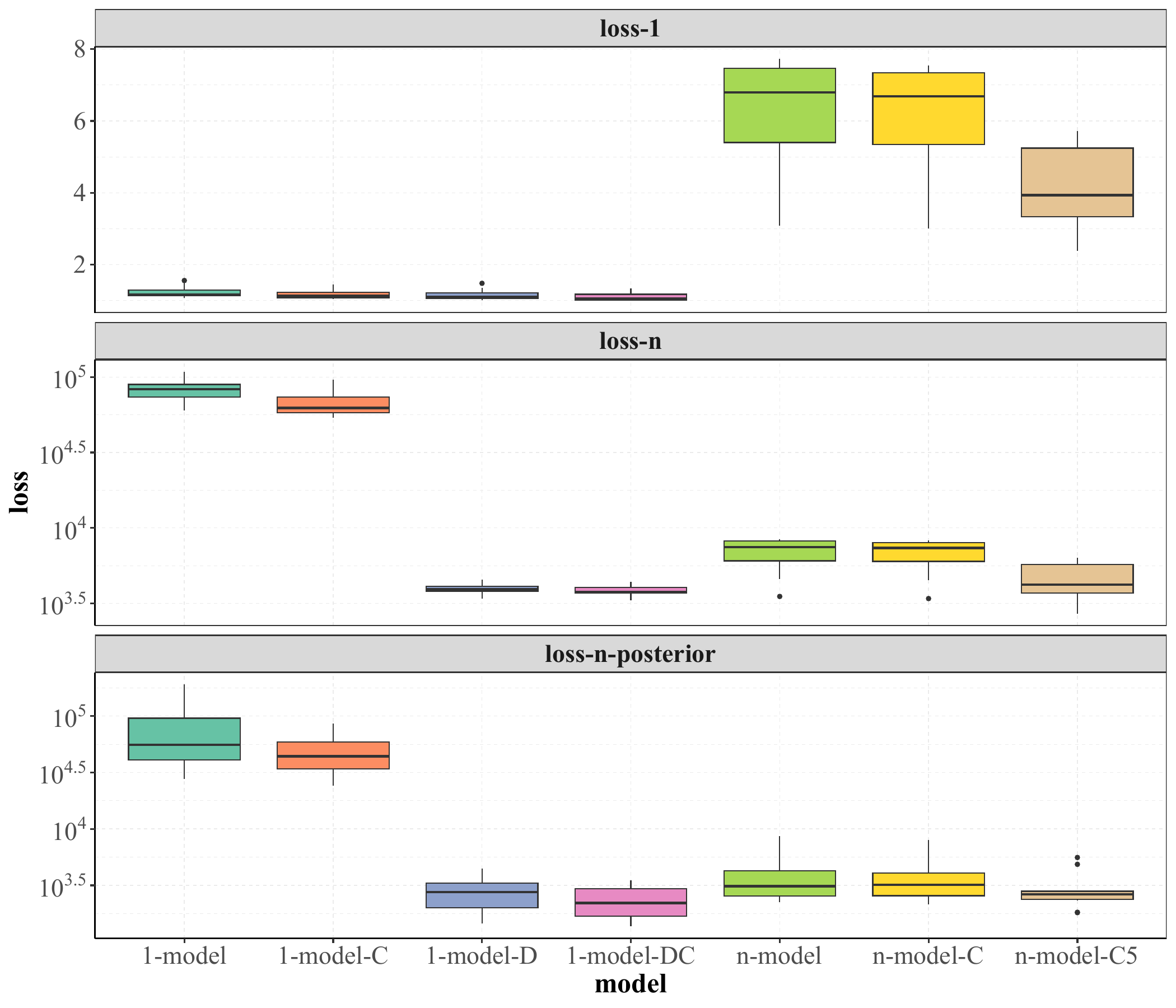}
    \caption{Score estimation loss. ``D" indicates debiasing and ``C" indicates curvature penalty.}
    \label{score_loss_plt}
\end{figure}

From \Cref{score_loss_plt}, we observe that the single-model with both debiasing and curvature penalty (1-model-DC) consistently achieves the lowest losses across all three criteria, followed by the model with debiasing only (1-model-D). Among the four single-observation variants, the debiasing step contributes much more substantially to error reduction than the curvature penalty. For the $n$-observation models, the curvature penalty alone (n-model-C) yields only a modest improvement over the baseline (n-model). Interestingly, 1-model-DC outperforms n-model-C even with a smaller simulation budget. Increasing the size of the reference table by five times (n-model-C5) reduces the losses further, but the gap with 1-model-DC remains. This suggests that the variation of $\theta$ values plays an important role in controlling out-of-sample loss: since 1-model-DC only requires generating one $X$ per $\theta$, the training dataset $\mD^S$ contains a wider spread of distinct parameter values than the reference table $\mD$ used for the $n$-model variants.

{
\subsection{Details of the mRNA transfection model example}\label{sec:mrna_details}
\subsubsection{Data generating process}
For each cell $i$, let $m^{(i)}(t)$ and $p^{(i)}(t)$ denote the amounts of mRNA and green fluorescent protein (GFP) molecules at time $t$, respectively. Their dynamics are governed by the following stochastic differential equation (SDE):
{\footnotesize
\begin{align*}\label{eq:mRNA_dgp}
\begin{pmatrix}
m^{(i)} \\
p^{(i)}
\end{pmatrix}
(t) &=
\begin{pmatrix}
0 \\
0
\end{pmatrix}
\text{ for } 0\leq t < t_0^{(i)} \text{, and }
\begin{pmatrix}
m^{(i)} \\
p^{(i)}
\end{pmatrix}
(t_0^{(i)}) =
\begin{pmatrix}
m_0 \\
0
\end{pmatrix} \\
d
\begin{pmatrix}
m^{(i)} \\
p^{(i)}
\end{pmatrix}
(t)
&=
\begin{pmatrix}
-\delta^{(i)} \cdot m^{(i)}(t) \\
k^{(i)} \cdot m^{(i)}(t) - \gamma^{(i)} \cdot p^{(i)}(t)
\end{pmatrix}
 dt
+
\begin{pmatrix}
\sqrt{\delta^{(i)} \cdot m^{(i)}(t)} & 0 \\
0 & \sqrt{k^{(i)} \cdot m^{(i)}(t) + \gamma^{(i)} \cdot p^{(i)}(t)}
\end{pmatrix}
dB_t^{(i)}
\end{align*}}
where $(\delta^{(i)}, \gamma^{(i)}, k^{(i)}, t_0^{(i)})
\sim \mN\left((\mu_\delta, \mu_\gamma, \mu_k, \mu_{t_0}), \text{diag}(\tau_\delta^2, \tau_\gamma^2, \tau_k^2, \tau_{t_0}^2) \right)$ and $B_t^{(i)}$ is a standard Brownian motion. The observed data are noisy transformations of the GFP trajectory:
\[
X^{(i)}(t) = \log(\text{scale}\cdot p^{(i)}(t) + \text{offset}) + \varepsilon^{(i)}(t), \qquad \varepsilon^{(i)}(t) \sim \mN(0, \sigma^2)
\]
We observe i.i.d. data from $n$ cells. For each cell, the trajectory is discretized using the Euler–Maruyama scheme and recorded at $180$ time points.

\subsubsection{Implementation details in simulation study}
{\bf Prior.} We choose independent Gaussian priors for
\[
(\log m_0, \log \text{scale}, \log \text{offset}, \log \text{sigma}, \mu_\delta, \mu_\gamma, \mu_k, \mu_{t_0}, \log {\tau_\delta}, \log {\tau_\gamma}, \log {\tau_k}, \log {\tau_{t_0}})
\]
with mean
\[
(5, 1, 3, -1.5, -0.694, -3, 0.027, 0, -0.8, -0.8, -0.8, -0.8)
\]
and standard deviation
\[
(1, 1, 1, 1, 0.6, 0.5, 1, 1, 0.5, 0.5, 0.5, 0.5)
\]

{\bf Localization.} In the generator $\tau(\theta, \Zn)$, $\Zn$ is simply a set of i.i.d. standard gaussian variables. Besides, since the dynamics is muted before $t_0$, the differentiability for the two $t_0$ related parameters $\mu_{t_0}$ and $\tau_{t_0}$ is broken. To address this, we use a smooth version indicator function $\I_{smooth}(t) = \frac{1}{1+\exp(-1000t)}$ to approximate the truncation at $t_0$ and restore differentiability. We use Adam with learning rate $1\times 10^{-1}$ for optimization. Each run converges within $200$ iterations, so the simulation cost to obtain $100$ samples is around $2\times 10^4$ (in unit of $\Xn$).

{\bf single-model training details.} We have $N = 1\times 10^7$ for the reference table $\mathcal{D}^S$, and $(N_R, m_R) = (1\times 10^6, 2\times 10^2)$ for the reference table $\mathcal{D}^R$. We use independent ELU networks for each output dimension to address the scale difference among different dimensions of the score, where each ELU neural network has $3$ hidden layers of $128$ units each. The network is trained with batch size $1\,000$ and learning rate gradually decreasing from $3\times10^{-4}$ to $6.25\times10^{-6}$, for $300$ epochs or till convergence, and another $6$ epochs after adding the curvature regularization. The mean regression is implemented using independent ELU networks for each output dimension, with $3$ hidden layers and $128$ units in each hidden layer for each network. The network is trained with batch size $64$ and learning rate gradually decreasing from $10^{-4}$ to $10^{-5}$, for $150$ epochs or until convergence, and another $1$ epoch after adding the curvature regularization. The curvature penalty parameter is chosen as $10^{-2}$ in both the score matching and the mean regression parts. We obtain $8\,000$ samples from $10$ independent Langevin Markov chains.

{\bf Details of NLE and NPE.} We use the ``sbi" Python package from \citep{tejero-cantero2020sbi} to implement both NLE and NPE, and we choose the flow network as a Masked Autoregressive Flow network \citep{papamakarios2017masked} with $5$ autoregressive transforms, each of which has $2$ hidden layers of $50$ units each. For NLE, we have reference table size $N = 2.1 \times 10^8$ . The network is trained with batch size $200$ and learning rate $5\times 10^{-4}$, for $2$ epochs. We draw $2\,000$ posterior samples from $4$ Markov Chains using the trained NLE model. For NPE, we have a permutation invariant embedding network $f_\phi(\Xn) = f_{\phi_1}(\sum_{i=1}^n f_{\phi_2}(X_i))$ before the flow model, where both $f_{\phi_1}$ and $f_{\phi_2}$ are fully connected neural networks with $3$ hidden layers of $128$ units each and output dimension $128$. We have reference table size $N = 1.05 \times 10^6$. The network is trained with batch size $200$ and learning rate $3\times 10^{-4}$ for $200$ epochs or until convergence. We draw $10\,000$ samples from the estimated posterior and samples are reweighted due to the use of proposal distribution.

{\bf Details of ABC.}
We generate a reference table of size $N=1.05\times10^6$ and keep $10\,000$ samples that have the smallest $\text{W}_1$ distance, and we reweight the samples by $\frac{\pi(\theta)}{q(\theta)}$.

{\bf Simulation cost comparison.} The simulation cost for all the methods is listed in \Cref{simu_budget_SDE_simu}, where one unit is the cost of generating $n$ observations.

\begin{table}[htbp]
    \centering
\caption{Simulation cost in the simulation study of the mRNA transfection example}
\begin{tabular}{ccccc}
\toprule
Localization & single-model & ABC-W1 & NPE & NLE \\
\midrule
$2\times 10^4$ & $1.05 \times 10^6$ & $1.05 \times 10^6$ & $1.05 \times 10^6$ & $1.05 \times 10^6$ \\
\bottomrule
\end{tabular}
 \label{simu_budget_SDE_simu}
\end{table}

\subsubsection{BSL results in simulation study}\label{sec:SDE_BSL}
We implement BSL using synthetic data in the mRNA transfection example. Following \citet{wood2010statistical}, we design a $10$-dimensional summary for each time-series observation, including the mean, standard deviation, skewness, kurtosis, quadratic variation and autocorrelation at several lags. For each simulated dataset, we then compute the component-wise mean and standard deviation of these summaries across the 200 observations, resulting in a 20-dimensional dataset-level summary.

We run a Metropolis–Hastings algorithm with the synthetic Gaussian likelihood, where the mean and covariance are estimated by $200$ simulated datasets at each iteration. We find the acceptance rate and convergence critically depend on the proposal choice and the initialization.

In \Cref{fig:SDE_trace_plots}, we show BSL trace plots initialized at the localization mean and at the true parameter, together with the trace plot produced by our method. Particularly, we construct a Gaussian proposal using the posterior obtained by our method, which significantly increases the acceptance rate. However, poor initializations still make the chain converge slowly.
The BSL chains exhibit slow movement and poor mixing, especially when initialized at the localization mean. In contrast, the chain generated by our method mixes well around the target region, by virture of the learned score information. Since BSL produces a reasonable posterior sample only if initialized at the true parameter and used with a carefully chosen proposal distribution, we exclude it from the main paper.
\begin{figure}[htbp]
    \centering

    \begin{subfigure}{0.99\textwidth}
        \centering
        \includegraphics[width=\textwidth]{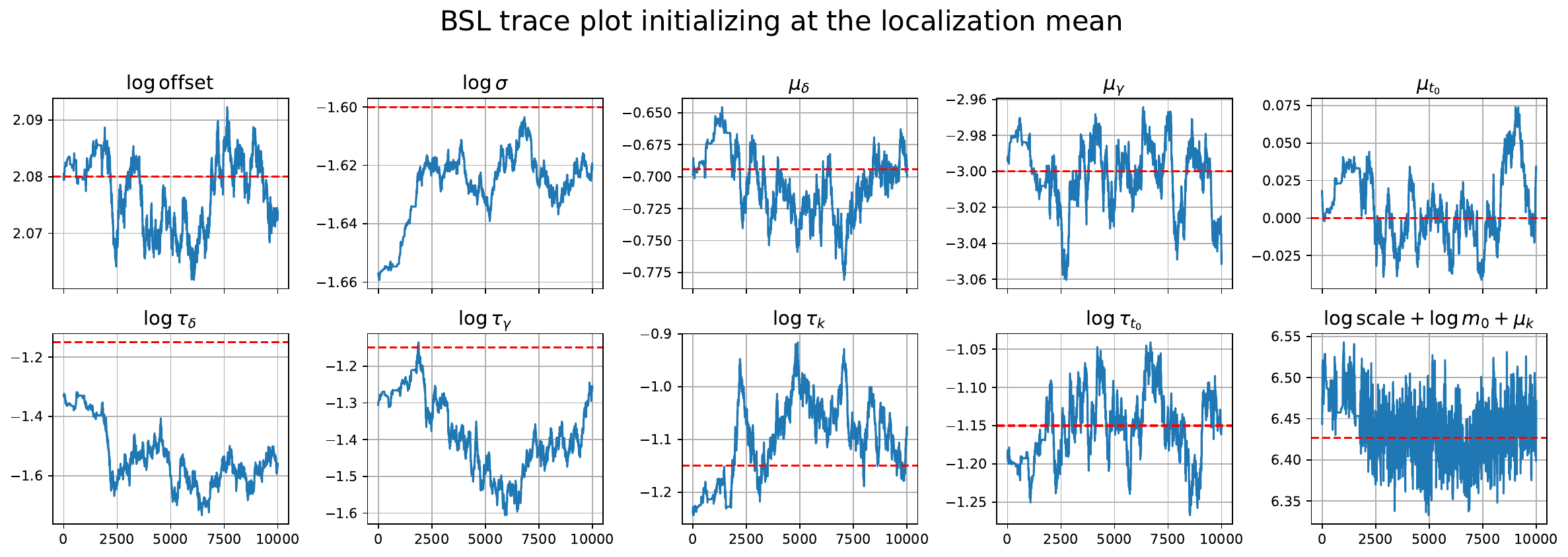}
        \label{fig:trace_plot1}
    \end{subfigure}
    \hfill
    \begin{subfigure}{0.99\textwidth}
        \centering
        \includegraphics[width=\textwidth]{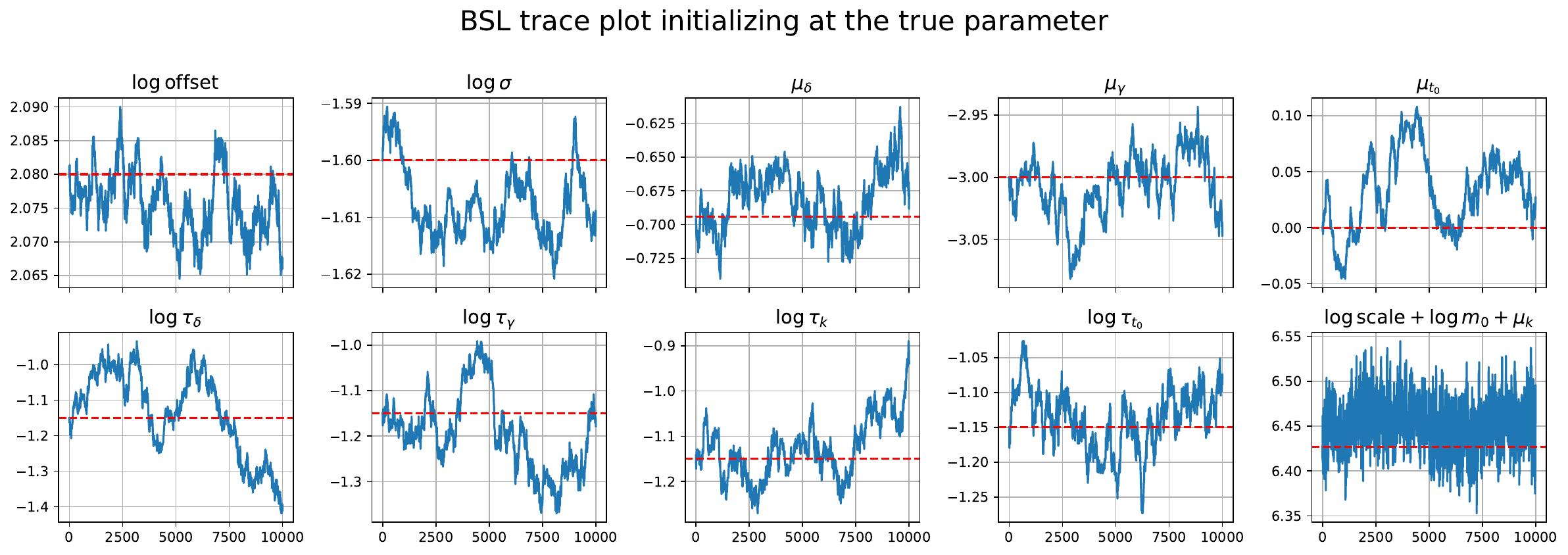}
        \label{fig:trace_plot2}
    \end{subfigure}
    \hfill
    \begin{subfigure}{0.99\textwidth}
        \centering
        \includegraphics[width=\textwidth]{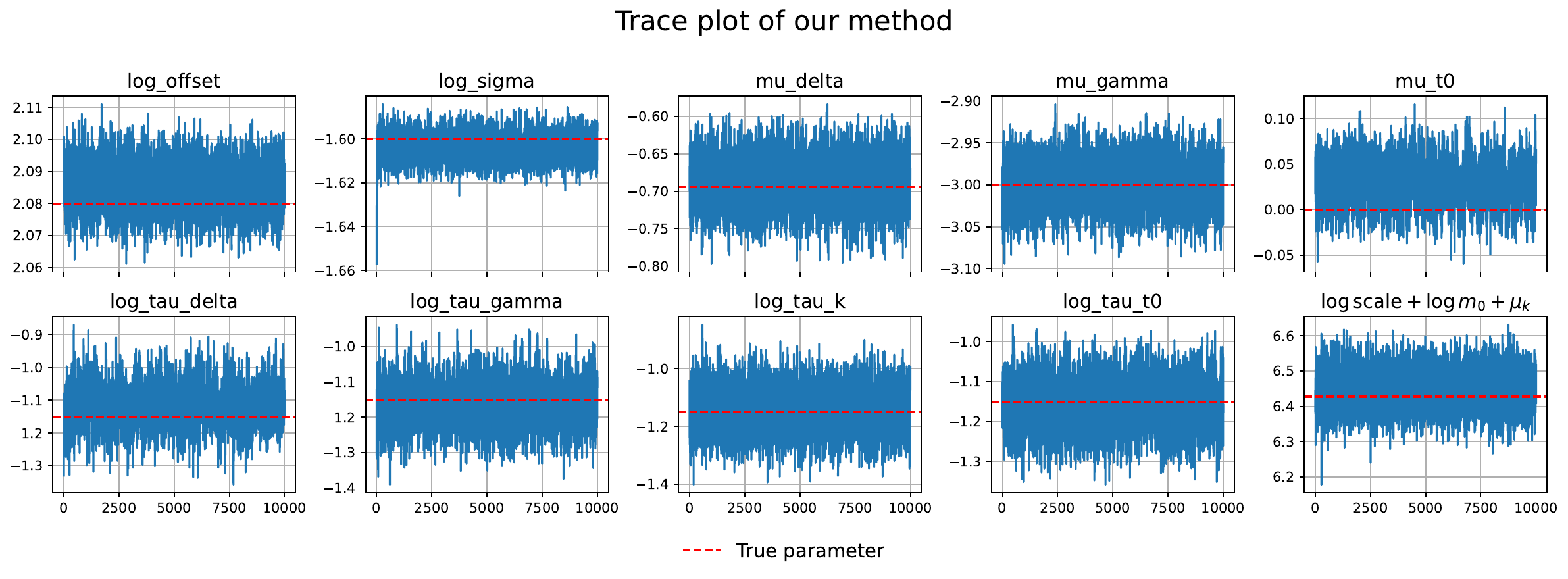}
        \label{fig:trace_plot3}
    \end{subfigure}

    \caption{Comparison of the three trace plots.}
    \label{fig:SDE_trace_plots}
\end{figure}

\subsubsection{Implementation details in real data application}
We choose the same prior as \citet{haggstrom2026simulation}, where the parameters
$(\log m_0, \log \text{scale}, \log \text{offset}, \log \sigma, \mu_\delta, \mu_\gamma, \mu_k, \mu_{t_0})$
have independent Gaussian priors with mean
$(5, 1, 3, -1, -1, -5, 0.5, 0)$
and standard deviation
$(1, 1, 1, 1, 1, 2, 1, 1)$.
The four precision parameters
$(\tau_\delta^{-2}, \tau_\gamma^{-2}, \tau_k^{-2}, \tau_{t_0}^{-2})$
have Gamma priors with shape parameters
$(2, 2, 2, 2)$
and rate parameters
$(0.5, 0.5, 0.5, 0.5)$.

The training settings are mostly the same as those in the simulation study. For the single model, we still use $N = 1\times 10^7$ for the reference table $\mathcal{D}^S$, and $(N_R, m_R) = (1\times 10^6, 2\times 10^2)$ for the reference table $\mathcal{D}^R$. To let all methods have the same simulation budget, we have $N=5\times 10^7$ single data points for NLE, and $N = 1.25\times 10^6$ for NPE and ABC. Besides, we only keep the best $1000$ samples in ABC, since using a larger acceptance rate leads to poorer posterior quality. The simulation cost for localization and all the methods is listed in \Cref{simu_budget_SDE_realdata}, where one unit is the cost of generating $n$ observations.

\begin{table}[htbp]
    \centering
\caption{Simulation cost in the simulation study of the mRNA transfection example}
\begin{tabular}{ccccc}
\toprule
Localization & single-model & ABC-W1 & NPE & NLE \\
\midrule
$1.2\times 10^4$ & $1.25 \times 10^6$ & $1.25 \times 10^6$ & $1.25 \times 10^6$ & $1.25 \times 10^6$ \\
\bottomrule
\end{tabular}
 \label{simu_budget_SDE_realdata}
\end{table}

\subsubsection{Additional results in real data application}\label{sec:more_res_SDErealdata}
\textbf{NLE result.}
We present the posterior density result of NLE in \Cref{fig:NLE_SDE_realdata}. Its estimation of $\mu_\delta, \mu_\gamma, \mu_{t_0}, \log \tau_{t_0}$ are significantly different from the other methods, and could not fit the observed data.

\begin{figure}[!htbp]
    \centering
    \includegraphics[width=0.9\textwidth]{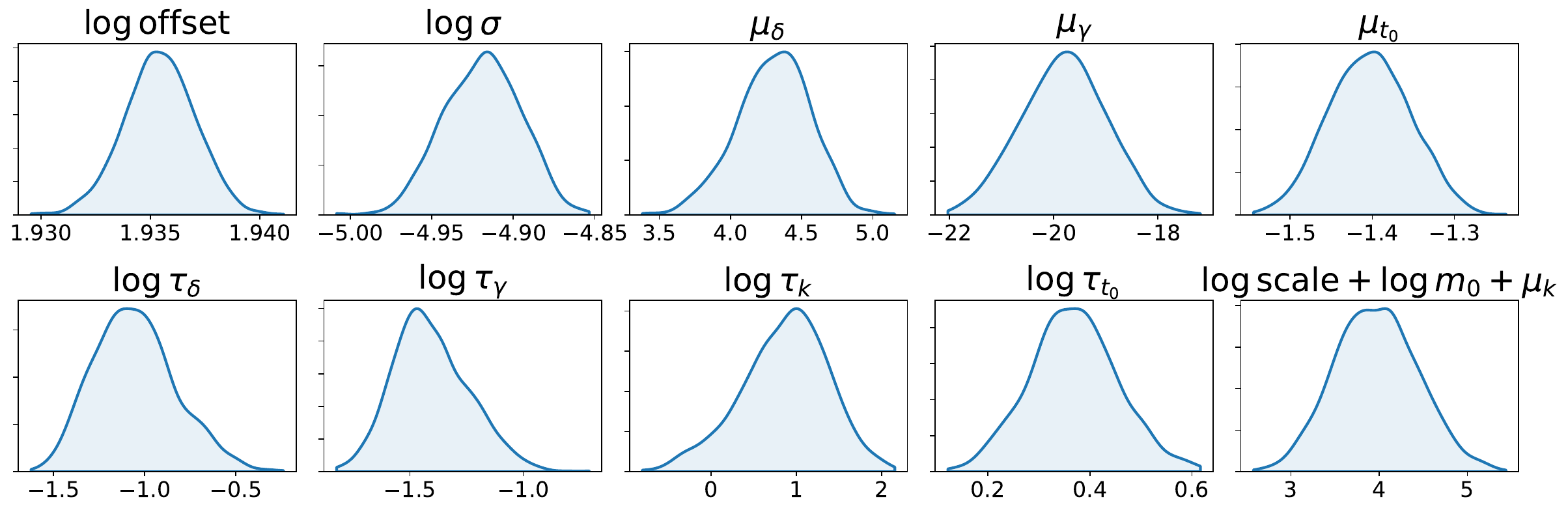}
    \caption{Posterior density plot of NLE results in the real data application.}
    \label{fig:NLE_SDE_realdata}
\end{figure}

\textbf{Posterior predictive distributions from each mode of the single-model result.}
We examine the posterior predictive distributions associated with the two posterior modes in \Cref{fig:SDE_realdata_yt_mode12}. Both modes provide a reasonable fit to the observed data, although the posterior predictions from Mode 2 are slightly higher than those from Mode 1.

\begin{figure}[!htbp]
    \centering
    \includegraphics[width=0.9\textwidth]{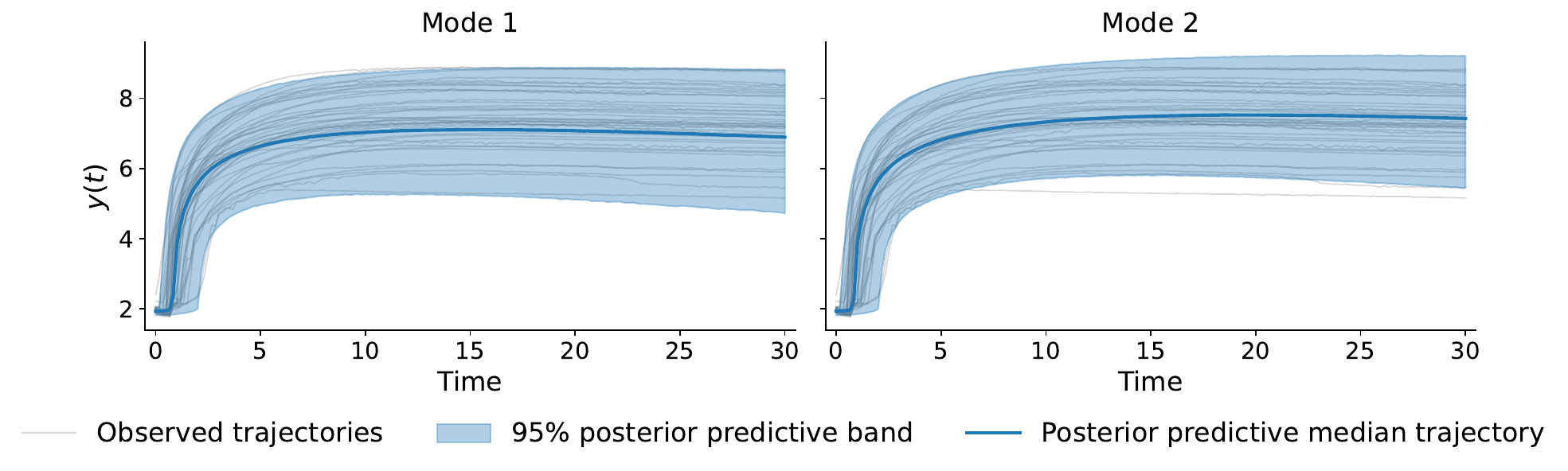}
    \caption{Posterior predictive distributions from each of the two modes of our result in the real data application.}
    \label{fig:SDE_realdata_yt_mode12}
\end{figure}
}

\subsection{Additional example - stochastic epidemic model}\label{sec:sto_epi_details}

In this subsection, we demonstrate the effectiveness of our proposed method on partially observed stochastic susceptible-infected (SI) models introduced by \citet{chatha2024neural}. This model is motivated by real-world problem of healthcare-associated infections (HAIs) that patients acquire infections during their stay in a healthcare facility, often transmitted via healthcare workers. The model has intractable likelihood, and the number of parameters varies according to the healthcare facility, making it a great fit for evaluating our proposed methods. In this example, due to the dependent nature of the data, we only include n-model in this example.

\subsubsection{Data generating process}

We briefly introduce the model; for full details, we refer readers to Section 3 of \citet{chatha2024neural}. The setting involves monitoring a healthcare facility with $n$ individuals over $T$ time steps, and the individuals are distributed across $R$ rooms on $J$ floors.

Each individual $i = 1, \dots,n$ has a binary infection status $Y_{i,t}$ at time $t = 1, \dots, T$, where $Y_{i,t} = 1$ indicates ``infected'', $Y_{i,t} = 0$ denotes ``susceptible'', and there is no ``recovered'' state in the application context. The infection transition is modeled as:
\begin{equation*}
    \mathbb{P}(Y_{i,t} = 1 \mid Y_{i,t-1} = 0) = 1-e^{-\lambda_i(t)},
\end{equation*}
where $\lambda_i(t)$ is the infection risk of individual $i$ at time $t$, determined by all currently infected individuals in the facility:
\begin{align*}
\lambda_{i \leftarrow l} &= \frac{\beta_0}{n} + \frac{\beta_{F_{(i)}}}{n_F}\cdot\bm{1}\{\text{$i$ and $l$ are floormates}\} + \frac{\beta_{J+1}}{n}\cdot\bm{1}\{\text{$i$ and $l$ are roommates}\} \\
\lambda_i(t) &= \sum_{l:Y_{l,t-1}} \lambda_{i \leftarrow l},
\end{align*}
where $F(i) \in \{1, \dots, J\}$ denotes the floor assignment of individual $i$ and $n_F$ is the number of individuals per floor. The model parameters include the facility transmission rate $\beta_0$, floor-specific transmission rate $\beta_j$ for $j=1,\dots,J$, and the roommate transmission rate $\beta_{J+1}$.

 The model also accounts for two real-world complexities. First is the  random intake and outtake. At every $t$, each individual is discharged with  probability $\gamma$ and immediately replaced by a new admission. Each new admission carries the pathogen with probability $\alpha$. The second phenomenon is partial observation of cases. Infections are not always observed and infected individuals may be asymptomatic carriers. Let $X_{i,t}$ denote the observed status. An infection is observed ($X_{i,t} = 1$) if (1) $X_{i,t-1} = 1$ and either the individual exhibits symptoms with probability $\eta$ or the individual is newly admitted and is detected due to entrance screening. Moreover, once $X_{i,t} = 1$ is observed, all future observations $X_{i,s}$ for $s > t$ remain 1, unless the individual is discharged.

In summary, the observed data is $\Xn^* = \{X_{i,t}:\;i = 1,\dots,n, t=1,\dots,T \} \in \mathbb{R}^{nT}$, and the model parameters are the transmission rates $\theta = \{\beta_0, \beta_1, \dots, \beta_{J+1}\}$. Since the dimension of the data can be high due to a large $n$, \citet{chatha2024neural} proposed to do inference based on summary statistics $S\in \mathbb{R}^{T\times (J+2)}$ of $\Xn$, and they showed that empirically this summary statistics captures enough information and can produce high quality inference. Specifically, the summary statistics at time $t$, $S_t$, records the following information at time $t$: the observed infection proportion in the facility and in each floor, and the proportion of rooms in which both people are infected. For all methods included in our simulation study, we use these summary statistics instead of $\Xn$ for inference.

We consider two simulation scenarios. We adopt setting 1 from  \citep{chatha2024neural}, where $T=52, J=5, n=300$ and each room has $2$ individuals. The hyperparameters are $\gamma = 0.05$, $\alpha = \eta = 0.1$, and the true model parameter is ${\theta}^\ast = \{\beta_0^\ast, \beta_1^\ast, \dots, \beta_{J+1}^\ast\} = (0.05, 0.02, 0.04, 0.06, 0.08, 0.1, 0.05)$, and the prior for $\theta$ is log-normal$(-3, 1)$ for all coordinates. For setting 2, we move on to higher dimensions and  a less informative prior. We maintain the setting in Simulation 1, except adding 5 floors with transmission rate $(0.12, 0.14, 0.16, 0.18, 0.2)$, which increase the total number of individuals to $n=600$. We choose the prior to be log-normal$(-3,2)$ on all coordinates. We provide an example of how the ratio of infection evolves under the two simulation settings in \Cref{repre_SIdata}. We compare our proposed method with the NPE method used in \citep{chatha2024neural}, BSL and ABC-W1. {Note that \citet{chatha2024neural} adopt a multivariate Gaussian distribution to parameterize the posterior distribution, instead of using a more standard flow model. This choice is reasonable in this example bacause the posterior is close to Gaussian. We follow the same choice for a more direct comparison.}

\begin{figure}[!ht]
    \centering
    \begin{subfigure}[b]{0.45\textwidth}
        \centering
        \includegraphics[width=\textwidth]{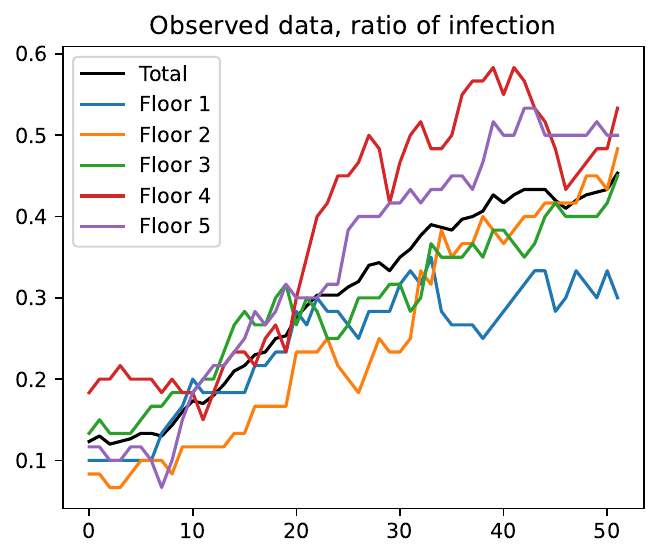}
        \caption{}
        \label{obs_data_SIprior}
    \end{subfigure}
    \hfill
    \begin{subfigure}[b]{0.45\textwidth}
        \centering
        \includegraphics[width=\textwidth]{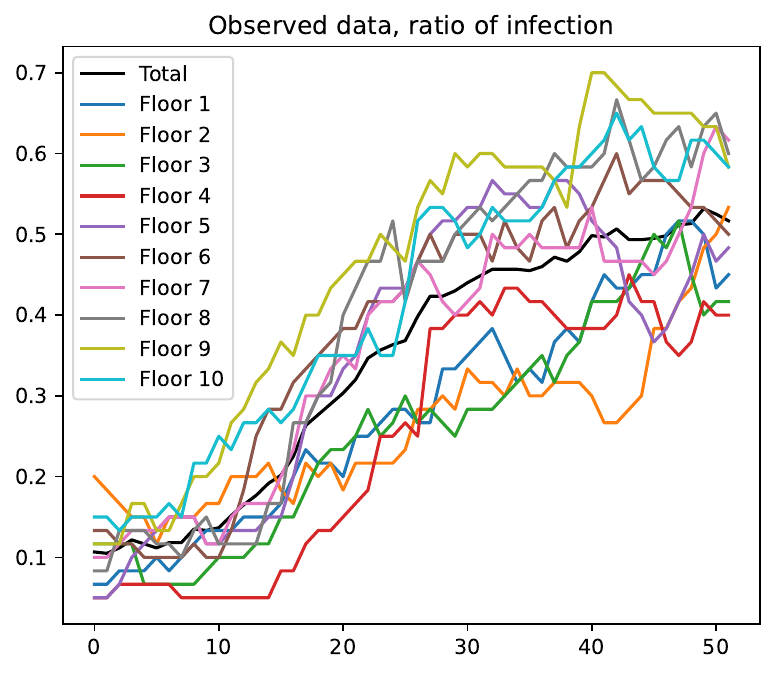}
        \caption{}
        \label{obs_data_SI10F}
    \end{subfigure}
    \caption{Example of observed data under (a) 5-floor setting and (b) 10-floor setting.}
    \label{repre_SIdata}
\end{figure}

\subsubsection{Simulation results in setting 1}

\begin{figure}[!ht]
    \centering
    \includegraphics[width=0.9\textwidth]{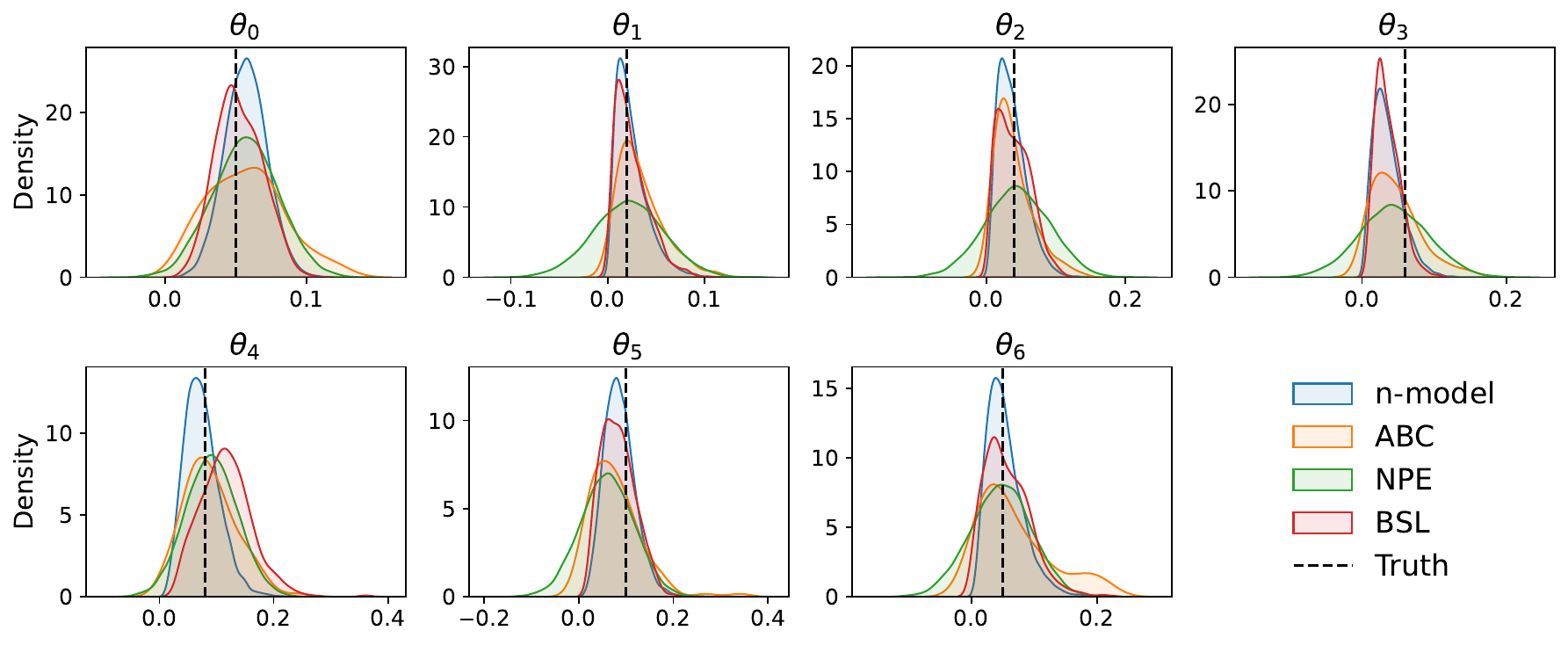}
    \caption{Posterior densities  of different methods under the 5-floor setting}
    \label{kde_SIprior}
\end{figure}

 For setting 1 with 5 floors, averaged results over $100$ experiments are shown in \Cref{SI_simu1_table} and the posterior density plot from one experiment is \Cref{kde_SIprior}. We can see that although all four methods have similar estimation bias, our method has smaller $95\%$ credible intervals for all the parameters, while maintaining high coverage rates. {This suggests that our method provides comparatively better uncertainty quantification. Although the n-model also tends to overestimate posterior uncertainty, as reflected by coverage above the nominal level, its credible intervals are narrower in most cases.}

 This suggests that our method has better uncertainty quantification, as the other methods seem to overestimate posterior uncertainty, with coverage rate almost always equal to $1$.

\begin{table}[!ht]
    \centering
\caption{Averaged results over 100 experiments in simulation setting 1. We report the standard deviations under the average.}
\resizebox{\textwidth}{!}{
\begin{tabular}{l|c|cccc|cccc|cccc}
\toprule
 & \multirow{2}{*}{$\theta^\ast$} & \multicolumn{4}{c|}{$\lvert \widehat{\theta} - \theta^\ast \rvert$} & \multicolumn{4}{c|}{CI95 Width} & \multicolumn{4}{c}{Cover95} \\
 &  & ABC & BSL & NPE & n-model & ABC & BSL & NPE & n-model & ABC & BSL & NPE & n-model \\
\midrule
Facility & 0.05 & \makecell{\bf{0.009} \\ (0.006)} & \makecell{0.010 \\ (0.007)} & \makecell{\bf{0.009} \\ (0.007)} & \makecell{0.010 \\ (0.008)} & \makecell{0.085 \\ (0.013)} & \makecell{0.081 \\ (0.037)} & \makecell{0.108 \\ (0.054)} & \makecell{\bf{0.066} \\ (0.009)} & 1.00 & 1.00 & 1.00 & 0.98 \\
\midrule
Floor 1 & 0.02 & \makecell{0.018 \\ (0.008)} & \makecell{0.016 \\ (0.012)} & \makecell{0.022 \\ (0.015)} & \makecell{\bf{0.014} \\ (0.013)} & \makecell{0.092 \\ (0.021)} & \makecell{0.094 \\ (0.065)} & \makecell{0.143 \\ (0.042)} & \makecell{\bf{0.074} \\ (0.021)} & 1.00 & 1.00 & 1.00 & 0.98 \\
\midrule
Floor 2 & 0.04 & \makecell{\bf{0.009} \\ (0.009)} & \makecell{0.011 \\ (0.009)} & \makecell{0.020 \\ (0.021)} & \makecell{0.014 \\ (0.013)} & \makecell{0.106 \\ (0.027)} & \makecell{0.105 \\ (0.050)} & \makecell{0.193 \\ (0.122)} & \makecell{\bf{0.090} \\ (0.027)} & 1.00 & 1.00 & 1.00 & 1.00 \\
\midrule
Floor 3 & 0.06 & \makecell{\bf{0.013} \\ (0.009)} & \makecell{0.016 \\ (0.011)} & \makecell{0.015 \\ (0.014)} & \makecell{0.016 \\ (0.012)} & \makecell{0.123 \\ (0.031)} & \makecell{0.125 \\ (0.092)} & \makecell{0.199 \\ (0.114)} & \makecell{\bf{0.104} \\ (0.026)} & 1.00 & 1.00 & 1.00 & 1.00 \\
\midrule
Floor 4 & 0.08 & \makecell{\bf{0.023} \\ (0.015)} & \makecell{\bf{0.023} \\ (0.016)} & \makecell{0.025 \\ (0.025)} & \makecell{\bf{0.023} \\ (0.017)} & \makecell{0.148 \\ (0.052)} & \makecell{0.131 \\ (0.040)} & \makecell{0.218 \\ (0.148)} & \makecell{\bf{0.126} \\ (0.037)} & 0.99 & 0.97 & 1.00 & 0.98 \\
\midrule
Floor 5 & 0.10 & \makecell{0.028 \\ (0.017)} & \makecell{0.026 \\ (0.018)} & \makecell{0.030 \\ (0.019)} & \makecell{\bf{0.025} \\ (0.019)} & \makecell{0.179 \\ (0.054)} & \makecell{0.165 \\ (0.071)} & \makecell{0.211 \\ (0.071)} & \makecell{\bf{0.138} \\ (0.033)} & 0.98 & 0.98 & 0.99 & 0.96 \\
\midrule
Room & 0.05 & \makecell{0.014 \\ (0.009)} & \makecell{0.016 \\ (0.015)} & \makecell{\bf{0.013} \\ (0.010)} & \makecell{0.015 \\ (0.012)} & \makecell{0.204 \\ (0.048)} & \makecell{0.133 \\ (0.051)} & \makecell{0.216 \\ (0.098)} & \makecell{\bf{0.123} \\ (0.039)} & 1.00 & 1.00 & 1.00 & 1.00 \\
\bottomrule
\end{tabular}
}
\label{SI_simu1_table}
\end{table}

\subsubsection{Simulation results in setting 2}
For setting 2 with 10 floors, we again compare all methods over 50 experiments. For each run, we apply our localization step  to get the proposal distribution $q(\theta)$ and this proposal distribution is used to generate the reference table for all methods. We increase the size of the reference table to $10\,000$ due to the increased number of parameters. The averaged results are reported in \Cref{SI_simu2_table} and the snapshot of posterior densities from one experiment is provided in \Cref{kde_SI10F}. Our method has the smallest estimation bias for most parameters and again has a significant advantage with respect to the $95\%$ credible interval. For NPE, it has much larger credible interval width than other methods, with coverage rates close to $1$ for most parameters, indicating overestimation of the posterior uncertainty. For ABC, its coverage rates are significantly smaller than the nominal rates for most parameters, indicating underestimation of posterior uncertainty.

\begin{table}[!ht]
    \centering
\caption{Averaged results over 50 experiments in simulation setting 2. We report the standard deviations under the average. Note that we mark the CI width with bold font if it is small and its coverage rate is not too far from the nominal level.}
\resizebox{\textwidth}{!}{
\begin{tabular}{l|c|cccc|cccc|cccc}
\toprule
 & \multirow{2}{*}{$\beta_*$} & \multicolumn{4}{c|}{$\lvert \widehat{\beta} - \beta_* \rvert$} & \multicolumn{4}{c|}{CI95 Width} & \multicolumn{4}{c}{Cover95} \\
 &  & ABC & BSL & NPE & n-model & ABC & BSL & NPE & n-model & ABC & BSL & NPE & n-model \\
\midrule
Facility & 0.05 & \makecell{0.013 \\ (0.011)} & \makecell{0.011 \\ (0.010)} & \makecell{\bf{0.007} \\ (0.007)} & \makecell{0.009 \\ (0.007)} & \makecell{0.050 \\ (0.016)} & \makecell{0.061 \\ (0.010)} & \makecell{0.084 \\ (0.015)} & \makecell{\bf{0.055} \\ (0.007)} & 0.86 & 0.94 & 1.00 & 1.00 \\
\midrule
Floor 1 & 0.02 & \makecell{0.013 \\ (0.017)} & \makecell{\bf{0.010} \\ (0.012)} & \makecell{0.035 \\ (0.021)} & \makecell{0.017 \\ (0.015)} & \makecell{\bf{0.063} \\ (0.037)} & \makecell{0.077 \\ (0.029)} & \makecell{0.271 \\ (0.142)} & \makecell{0.127 \\ (0.030)} & 0.94 & 1.00 & 0.98 & 0.94 \\
\midrule
Floor 2 & 0.04 & \makecell{0.021 \\ (0.014)} & \makecell{0.017 \\ (0.009)} & \makecell{0.021 \\ (0.018)} & \makecell{\bf{0.014} \\ (0.014)} & \makecell{0.067 \\ (0.038)} & \makecell{\bf{0.081} \\ (0.025)} & \makecell{0.250 \\ (0.121)} & \makecell{0.131 \\ (0.025)} & 0.84 & 0.98 & 0.98 & 0.98 \\
\midrule
Floor 3 & 0.06 & \makecell{0.028 \\ (0.020)} & \makecell{\bf{0.021} \\ (0.013)} & \makecell{0.022 \\ (0.017)} & \makecell{0.022 \\ (0.015)} & \makecell{0.085 \\ (0.041)} & \makecell{\bf{0.107} \\ (0.030)} & \makecell{0.223 \\ (0.119)} & \makecell{0.128 \\ (0.031)} & 0.78 & 0.98 & 1.00 & 0.98 \\
\midrule
Floor 4 & 0.08 & \makecell{0.040 \\ (0.031)} & \makecell{0.030 \\ (0.020)} & \makecell{\bf{0.023} \\ (0.018)} & \makecell{0.026 \\ (0.018)} & \makecell{0.113 \\ (0.063)} & \makecell{0.122 \\ (0.031)} & \makecell{0.190 \\ (0.086)} & \makecell{\bf{0.131} \\ (0.022)} & 0.82 & 0.90 & 0.96 & 0.96 \\
\midrule
Floor 5 & 0.10 & \makecell{0.041 \\ (0.035)} & \makecell{0.037 \\ (0.023)} & \makecell{0.031 \\ (0.018)} & \makecell{\bf{0.028} \\ (0.019)} & \makecell{0.126 \\ (0.065)} & \makecell{0.153 \\ (0.047)} & \makecell{0.194 \\ (0.053)} & \makecell{\bf{0.138} \\ (0.021)} & 0.74 & 0.92 & 0.98 & 0.98 \\
\midrule
Floor 6 & 0.12 & \makecell{0.046 \\ (0.030)} & \makecell{0.034 \\ (0.027)} & \makecell{0.029 \\ (0.023)} & \makecell{\bf{0.027} \\ (0.022)} & \makecell{0.132 \\ (0.065)} & \makecell{0.151 \\ (0.040)} & \makecell{0.182 \\ (0.049)} & \makecell{\bf{0.141} \\ (0.018)} & 0.74 & 0.90 & 0.98 & 0.98 \\
\midrule
Floor 7 & 0.14 & \makecell{0.045 \\ (0.040)} & \makecell{0.043 \\ (0.048)} & \makecell{0.032 \\ (0.025)} & \makecell{\bf{0.030} \\ (0.029)} & \makecell{0.170 \\ (0.101)} & \makecell{0.219 \\ (0.138)} & \makecell{0.201 \\ (0.032)} & \makecell{\bf{0.160} \\ (0.035)} & 0.86 & 0.96 & 0.96 & 0.94 \\
\midrule
Floor 8 & 0.16 & \makecell{0.054 \\ (0.044)} & \makecell{0.055 \\ (0.040)} & \makecell{0.045 \\ (0.041)} & \makecell{\bf{0.043} \\ (0.041)} & \makecell{0.201 \\ (0.096)} & \makecell{0.229 \\ (0.103)} & \makecell{\bf{0.229} \\ (0.069)} & \makecell{0.172 \\ (0.047)} & 0.78 & 0.92 & 0.94 & 0.90 \\
\midrule
Floor 9 & 0.18 & \makecell{0.055 \\ (0.044)} & \makecell{0.052 \\ (0.048)} & \makecell{0.044 \\ (0.030)} & \makecell{\bf{0.039} \\ (0.027)} & \makecell{0.195 \\ (0.103)} & \makecell{0.249 \\ (0.140)} & \makecell{0.224 \\ (0.059)} & \makecell{\bf{0.174} \\ (0.037)} & 0.86 & 0.96 & 0.94 & 0.94 \\
\midrule
Floor 10 & 0.20 & \makecell{0.066 \\ (0.055)} & \makecell{0.045 \\ (0.045)} & \makecell{\bf{0.035} \\ (0.030)} & \makecell{0.038 \\ (0.030)} & \makecell{0.234 \\ (0.103)} & \makecell{0.278 \\ (0.125)} & \makecell{0.255 \\ (0.061)} & \makecell{\bf{0.197} \\ (0.040)} & 0.90 & 0.98 & 1.00 & 0.96 \\
\midrule
Room & 0.05 & \makecell{0.022 \\ (0.023)} & \makecell{\bf{0.016} \\ (0.010)} & \makecell{0.027 \\ (0.020)} & \makecell{0.020 \\ (0.013)} & \makecell{0.126 \\ (0.056)} & \makecell{0.116 \\ (0.038)} & \makecell{0.367 \\ (0.160)} & \makecell{\bf{0.110} \\ (0.041)} & 0.96 & 0.98 & 1.00 & 0.94 \\
\bottomrule
\end{tabular}
}
\label{SI_simu2_table}
\end{table}

\begin{figure}[!ht]
    \centering
    \includegraphics[width=0.9\textwidth]{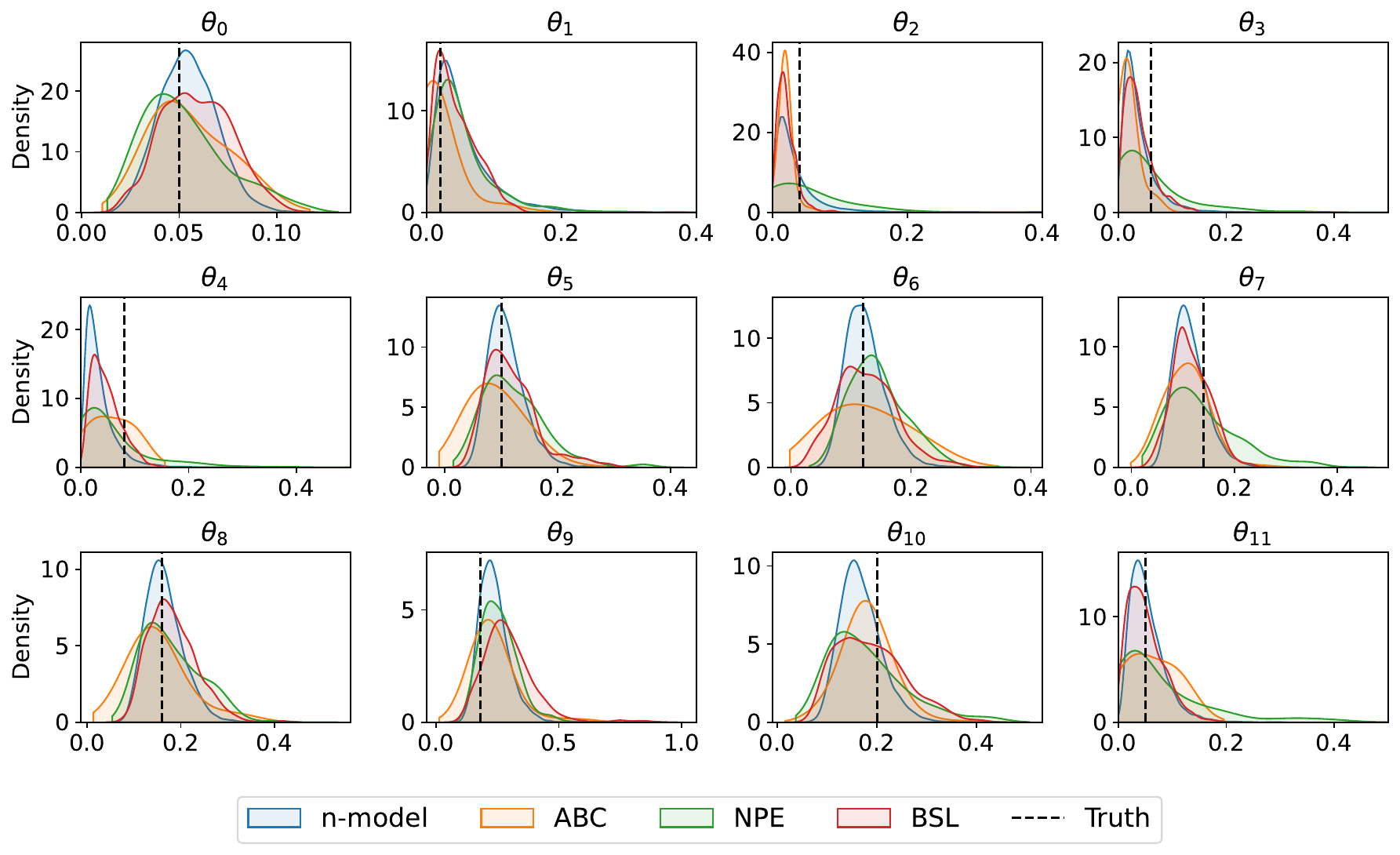}
    \caption{Posterior density plots of different methods in one experiment of 10-floor setting}
    \label{kde_SI10F}
\end{figure}

\subsubsection{Implemtation details under setting 1}
We use the prior $\pi(\theta)$ in this setting and do not use localization step to construct a proposal distribution $q(\theta)$.

{\bf Training details of n-model.}
We have reference table size $N = 8\,000$. We use an ELU neural network with $3$ hidden layers of dimension $(d_\theta + d_S=371, 512, 256, 128, d_\theta = 7)$. The network is trained with batch size $5$ and learning rate gradually decreasing from $10^{-4}$ to $10^{-5}$, for $100$ epochs or until convergence. We obtain $10\,000$ samples from $10\,000$ independent Langevin Markov chains.

{\bf Training details of NPE.}
We have reference table size $N = 8\,000$. We use the code\footnote{https://github.com/epibayes/np-epid} from \citep{chatha2024neural} to implement the NPE method. There, they specify the posterior as a multivariate Gaussian distribution, and estimate its mean and covariance using a neural network via the maximum likelihood criterion. We adopt the same configurations as in \citep{chatha2024neural}. We use a ReLU network with $3$ hidden layers and $32$ units in each hidden layer. The network is trained with full batch size $4\,000$ and learning rate $10^{-3}$, for $1\,000$ epochs or until convergence.

{\bf Details of ABC and BSL.} For BSL, we use the sample mean and element-wise standard deviation of $S$ as the summary statistics to do inference, because otherwise it requires much more simulations to estimate the covariance of the synthetic likelihood. We obtain $8\;000$ samples from $10$ independent Markov chains, where each chain is drawn using Metropolis–Hastings algorithm, with length $1\;000$ and discarding the initial $200$ iterations, and we use $100$ simulations at each iteration to estimate the mean and covariance of the synthetic Gaussian likelihood.

For ABC, we use $\text{W}_1$ distance based on summary statistics. we generate $8\,000$ and keep $100$ samples that have the smallest $\text{W}_1$ distance. We find the results similar regarding treating rows of $S$ as i.i.d. or dependent.

{\bf Simulation cost comparison.} The simulation cost for all the methods is listed in \Cref{simu_budget_SI5}, where one unit is the cost of generating a whole dataset $\Xn$. For the n-model, NPE and ABC, the cost is the size $N$ of the reference table. For BSL, its cost is the product of the number of chains, the length per chain, and the number of simulations at each iteration within each chain.

\begin{table}[!ht]
    \centering
\caption{Simulation cost in the epidemic model example (5-floor setting)}
\begin{tabular}{cccc}
\toprule
ABC & BSL & NPE & n-model \\
\midrule
 $8\times 10^3$ & $1\times 10^6$ & $8\times 10^3$ & $8\times 10^3$ \\
\bottomrule
\end{tabular}
    \label{simu_budget_SI5}
\end{table}

\subsubsection{Implementation details under setting 2}

{\bf Localization.} In the generator $\tau(\theta, \Zn)$, $\Zn$ contains a set of i.i.d. Bernoulli random variables and a set of i.i.d. $\text{Uniform}(0, 1)$ random variables. The Bernoulli random variables are indicators of replacement in the data generation, and the Uniform random variables are quantiles to generate the binary states $X$ and $Y$. It is worth mentioning a detail of the generator here, which is due to the binary nature of the data. For example, to generate a binary state $Y$, it is natural to use the generator $\tau(\theta, Z) = \I(Z \le p(\theta))$ to sample $Y \mid \theta$ where $\mathbb{P}(Y = 1 \mid \theta) = p(\theta)$ is know. Here $\I(\cdot)$ is the indicator function and $Z$ follows $\text{Uniform}(0, 1)$.
However, the indicator function disables using gradient based optimization methods to solve the optimization problem \eqref{eq:smm_sol} as its gradient is $0$ almost everywhere. Therefore, we use a smooth version indicator function $\I_{smooth}(t) = \frac{1}{1+\exp(-500t)}$ as a substitution in the localization procedure. To optimize \eqref{eq:smm_sol}, we apply the Adam optimizer with learning rate $10^{-1}$ for 100 iterations, so the simulation cost for obtaining $100$ samples is $10\,000$ (in unit of $\Xn$).

{\bf Training details of our method.}
We have reference table size $N = 20\,000$. We use an ELU neural network with $3$ hidden layers of dimension $(d_\theta + d_S=636, 512, 256, 128, d_\theta = 12)$. The network is trained with batch size $5$ and learning rate gradually decreasing from $10^{-4}$ to $10^{-5}$, for $100$ epochs or till convergence.
We obtain $10\,000$ samples from $10\,000$ independent Langevin Markov chains.

{\bf Training details of NPE.}
We have reference table size $N = 20\,000$. We still use the code from \citep{chatha2024neural} to implement the NPE method, where the posterior is specified as a multivariate normal distribution, and the mean and covariance are estimated using a neural network via the maximum likelihood criterion. We use a ReLU network with $3$ hidden layers and $32$ units in each hidden layer. The network is trained with batch size $5\,000$ and learning rate $10^{-3}$, for $5\,000$ epochs or until convergence. We draw $10\;000$ samples from the estimated posterior and samples are reweighted by $\frac{\pi(\theta)}{q(\theta)}$.

{\bf Details of ABC and BSL.}
For the ABC method, we generate $N=20\;000$ and keep $100$ samples that have the smallest $\text{W}_1$ distance. We also reweight the samples by $\frac{\pi(\theta)}{q(\theta)}$. For the BSL method, we obtain $8\,000$ samples from $10$ independent Markov chains, where each chain is drawn using Metropolis–Hastings algorithm, with length $1\,000$ and discarding the initial $200$ iterations, and we use $200$ simulations at each iteration to estimate the mean and covariance of the synthetic Gaussian likelihood.

{\bf Simulation cost comparison.} The simulation cost for all methods is listed in \Cref{simu_budget_SI10}, where one unit is the cost of generating a whole dataset $\Xn$. For the n-model, NPE and ABC, the cost is the size $N$ of the reference table. For BSL, its cost is the product of the number of chains, the length per chain, and the number of simulations at each iteration within each chain.

\begin{table}[!ht]
    \centering
\caption{Simulation cost in the epidemic model example (10-floor setting)}
\begin{tabular}{ccccc}
\toprule
Localization& ABC & BSL & NPE & n-model \\
\midrule
$1\times 10^4$  & $2\times 10^4$ & $2\times 10^6$ & $2\times 10^4$ & $2\times 10^4$ \\
\bottomrule
\end{tabular}
    \label{simu_budget_SI10}
\end{table}

\end{document}